%% file: thesis.tex
\PassOptionsToPackage{hidelinks}{hyperref}
\pdfoutput=1
\documentclass{nottsthesis}

\usepackage[T1]{fontenc}
\usepackage{lmodern}
\usepackage{amssymb,amsmath}
\usepackage{ifxetex,ifluatex}
\usepackage{fixltx2e} 
\IfFileExists{upquote.sty}{\usepackage{upquote}}{}
\usepackage[utf8]{inputenc}
\newcommand{\set}{\cat{Set}}
\newcommand{\eq}[2]{\mathsf{Eq}\left(#1, #2\right)}

\usepackage{lipsum}

\thesisauthor{Paolo Capriotti}
\thesistitle{Models of Type Theory with Strict Equality}
\degree{Doctor of Philosophy}
\degreedate{July 2016}
\school{Computer Science}

\input{header.tex}

\begin{document}

\maketitle

\frontmatter

\begin{Abstract}
This thesis introduces the idea of \emph{two-level type theory}, an extension of Martin L\"{o}f type theory~\cite{martin-lof:type-theory} that adds a notion of \emph{strict equality} as an internal primitive.

A type theory with a strict equality alongside the more conventional form of equality, the latter being of fundamental importance for the recent innovation of \emph{homotopy type theory} ($\hott$), was first proposed by Voevodsky~\cite{hts}, and is usually referred to as HTS.

Here, we generalise and expand this idea, by developing a semantic framework that gives a systematic account of type formers for two-level systems, and proving a conservativity result relating back to a conventional type theory like $\hott$.

Finally, we show how a two-level theory can be used to provide partial solutions to open problems in $\hott$. In particular, we use it to construct semi-simplicial types, and lay out the foundations of an internal theory of $(\infty,1)$-categories.
\end{Abstract}

\begin{Acknowledgements}
I would like to thank the members of the Functional Programming Laboratory of the University of Nottingham, and in particular my supervisor Venanzio Capretta, for their advice and support, and for providing an excellent and stimulating research environment.

Thorsten Altenkirch, Nicolai Kraus and Christian Sattler deserve particular thanks for countless discussions, type theory meetings, and reading groups, all of which made the time of my PhD extremely enjoyable and fruitful.

And most of all, I would like to express my gratitude to my wife, Elisa, without whose support and encouragement, I would have certainly not made it to this point. Thank you!
\end{Acknowledgements}

\tableofcontents

\mainmatter

\chapter{Introduction}\label{chap:introduction}

\emph{Type theory} is a foundational framework for mathematics which can also be regarded as a programming language.  The central concept of type theory is of course that of a \emph{type}: an entity that plays the double role of a logical statement (a ``proposition'') and of a collection (a ``set'').

The sort of type theory developed in this thesis is more specifically referred to as \emph{Martin-Löf dependent type theory}~\cite{martin-lof:type-theory}, because it is based on the idea that types can depend on values.  This, together with a few basic primitives, makes the corresponding calculus powerful enough to express most fundamental mathematical ideas, including universal and existential quantification, functions, ordered pairs, etc.

In recent years, a new branch of type theory, called \emph{homotopy type theory} (HoTT) (\cref{sec:hott}) has arisen. The main revolution of HoTT consists in embracing the higher dimensional structure of equality, and using it to interpret types not just as sets, but as \emph{topological spaces} up to homotopy equivalence.

This has made it possible to formalise classical results of homotopy theory \emph{synthetically}, that is, without reference to the underlying representation of topological spaces - be it as sets equipped with a collection of open subsets, or any other formulation, possibly more well-behaved in a constructive setting.  Instead, spaces are studied abstractly, their features and properties derived simply from those of the types that represent them.  This has made the formulations and proofs of homotopical facts extremely elegant and streamlined, cast new light on seemingly well-understood results, and suggested new directions of research.

The interest of homotopy type theory lies in the fact that, by using its underlying type theoretic language, one is restricted to constructions that are \emph{automatically} homotopy-invariant: any concept, or definition, or result, by the mere fact of having been expressed ``internally'', is guaranteed to remain valid when spaces are replaced with equivalent ones.

This fundamental homotopy-invariance property crystallises into the principle of \emph{univalence}, probably the most important technical innovation of HoTT, which roughly states that (homotopy) equivalent types are \emph{equal}. Equality here is not meant in a \emph{strict} sense (i.e. equal types will not be interpreted as the same object in a model), but rather as the existence of some kind of ``path'' in the universe connecting the two types.

If follows that equality, despite still adhering to its defining property of being preserved by \emph{all} constructions, is not a \emph{mere proposition} anymore: it may possess non-trivial structure.  Paths themselves form a type, and their notion of equality is also subject to the same considerations.  Here we see how directly some of the most familiar constructions in classical homotopy theory, such as \emph{homotopy groups}, arise internally in the language of HoTT.

Unfortunately, HoTT, and its homotopy invariant nature, impose some fundamental constraints on the kind of constructions that we are allowed to perform internally.  Any classical definition whose (possibly ultimately irrelevant) details depend on more than the homotopy type of the spaces involved, needs to be reworked to fit into the framework of HoTT.

Sometimes, of course, this is not possible, as not all classical results hold in HoTT (a famous example is Whitehead's theorem, that only holds for truncated types in HoTT \cite{hott-book}).  Other times, it looks as though it \emph{should} be possible to provide an internal analogue of a classical notion, but all naive attempts fail.

The most prominent example of such a notion is that of \emph{semi-simplicial types}, which we explain in \cref{sec:infinite-structures}.  Giving a satisfactory account of this and similar ``infinite coherence'' problems is the main motivation behind this thesis.

\section{Overview}\label{sec:background}

Type theory, especially if directly introduced as a mathematical foundation, is usually presented as a collection of \emph{inference rules}.  Where the usual foundations of mathematics are based on some (often not clearly specified) form of first order logic, on top of which the well-known axioms of Zermelo-Fraenkel set theory~\cite{zf} are laid out, the rules of type theory form a single corpus that describes both the logical and the set-theoretical aspects of mathematics (and much more, as will be clear later when we will describe homotopy type theory).

In this thesis, we follow a slightly unconventional path: we define an algebraic notion of \emph{model of type theory}, as a category equipped with the logical structure necessary to talk about types.  The \emph{syntax} of type theory, then, instead of being implicitly defined by a set of rules, is taken to be the \emph{initial} model in our setting.

The advantage of our approach is that we do not have to deal with all the syntactical complications of name binding, type derivations or congruence rules for definitional equality (see for example \cite{hofmann:syntax-and-semantics}).
In fact, the initial model (provided it exists), is in particular \emph{a} model, hence it comes equipped with all the structure and satisfies all the axioms that we require.  Furthermore, and even more obviously, there is no ``initiality theorem''~\cite{streicher:semantics} to be proved, as the syntax is initial \emph{by definition}.

The disadvantage is that, since the syntax doesn't natively possess a notion of name binding, writing out terms explicitly in the language of the model can be cumbersome, and it makes for expressions that are extremely hard to read.  We will subvert this issue by devising a number of notational conventions (\cref{sec:notation,sec:notational-conventions}) that will make it possible to work in models of type theory \emph{as if} they had name binding, making constructions in a generic model essentially indistinguishable from their completely syntactical counterparts.

The thesis consists of an introductory chapter (\cref{chap:introduction}), followed by three main chapters.
In \cref{chap:models} we lay out our algebraic approach to models of type theory.
In \cref{chap:two-levels}, we extend our framework to cover \emph{two-level} models and prove a conservativity result.
Finally, in \cref{chap:type-theory}, we fix a particular instance of two-level type theory, and give some examples of what can be achieved by working internally in such a theory.

Our definition of model of type theory is based on \emph{categories with families} (CwF, \cite{cwf}), although our definition differs slightly from the original in non-essential ways (\cref{sec:cwf}).  We introduce and motivate a number of \emph{basic type formers} (\cref{sec:basic-type-formers}) using presheaves, and the fact that presheaf categories have a natural CwF structure (\cref{sec:presheaves}).

Once we have enough basic type formers under our belt, we will give a general definition of ``type former'' (\cref{sec:rule-framework}), and show how the basic ones defined previously can also be regarded as instances of the general definition.

We will then introduce two-level models of type theory, where two different type theories are combined in a single system.  This kind of structure naturally arises when studying certain homotopical models of type theory: types can be divided into \emph{fibrant} and \emph{strict}, resulting in two ``parallel'' type theories, with possibly different sets of type formers.

Perhaps surprisingly, with enough assumptions on the type formers involved, a two-level type theory is \emph{conservative} over its fibrant fragment (\cref{sec:conservativity}), meaning that proofs and constructions using the full two-level theory can always be reworked so that they only use the fibrant fragment, as long as the end result is itself fibrant. The proof mimics that of a similar result on the conservativity of the Logical Framework \cite{hofmann:syntax-and-semantics}.

The idea of the conservativity proof is straightforward, but is unfortunately made complicated by issues of strictness of coercion of fibrant types into strict types (\cref{sec:two-level-presheaves}).  We work around the strictness issues by defining the notion of \emph{regularity} for models (\cref{sec:regularisation}).

Finally, we will move completely inside a two-level model, and work in the internal language of the corresponding type theory (\cref{chap:type-theory}), in the style of \cite{hott-book}.  We choose a two-level type theory inspired by Voevodsky's HTS \cite{hts}, but more minimalistic (\cref{sec:hts}).

In our flavour of two-level type theory, we develop the notion of Reedy fibrant diagram, and show how they can be classified by fibrant types.  In particular, this yields a definition of \emph{semi-simplicial type}, a notion that has so far eluded all attempts at formalisation in conventional HoTT.

Our construction resembles the one in \cite{herbelin:semi-simplicial}, however, in the latter, a specific consequence of the existence of strict equality has to be assumed in order for the construction to go through.  We, instead, build on the general idea of Reedy fibrancy, and make no ad-hoc assumption beyond the general setup of two-level type theory.

From that, we lay out the foundations of an internal development of higher category theory, starting from the definition of \emph{complete semi-Segal type} (\cref{def:complete-semi-segal-type}), and showing why this is a good candidate for a notion of category that is powerful enough to include all the reasonable ``categorical'' structures present in HoTT, while at the same time allowing all the familiar categorical constructions to be performed within the constraints of type theory.

Most of the mathematical content of this thesis is based on a constructive meta-theory.  We do not make use of classical principles like the law of excluded middle or the axiom of choice.  One exception is the overview of the simplicial model of HoTT given in \cref{sec:simplicial-model}, since the construction referenced in \cite{simplicial-model} is explicitly non-constructive.
\footnote{There do exist attempts at building models of HoTT in a constructive setting \cite{sattler-gambino:frobenius} \cite{coquand:cubical}, but they are still relatively incomplete and poorly understood, hence we do not rely on them in this thesis.}

\section{Contributions}\label{sec:contributions}

The main contributions of this thesis are as follows:
\begin{itemize}
\item We develop a systematic and generic theory of \emph{type formers}: a single notion that can be instantiated to cover all known examples of what are usually referred to as type formers.  This is inspired by the ideas of the \emph{Logical Framework} \cite{logical-framework}, but our presentation is completely semantic in nature, and can be used to state and prove metatheoretical results about models of type theory without fixing a particular set of type formers in advance.
\item We define the notion of \emph{two-level type theory}, making precise and generalising the ideas underlying the HTS theory proposed by Voevodsky \cite{hts}.  We prove a conservativity result, which implies, among other things, that two-level type theory can be used as a ``schematic'' language for working with infinite families of objects in a conventional type theory.
\item We show how a particular minimalistic flavour of two-level type theory, similar to HTS, can be used to give partial solutions to some of the most pressing open problems in HoTT. In particular, we give a definition of semi-simplicial type, and use it to lay out the foundation of an internal theory of $(\infty,1)$-categories in type theory.
\end{itemize}

In particular, this thesis contains proofs of the following results:
\begin{itemize}
\item \cref{thm:fibrant-lift}, showing that any type former on a CwF can be lifted to the \emph{fibrant universe} of its presheaf category;
\item \cref{thm:regular-yoneda,thm:two-level-presheaves}, drawing a correspondence between a \emph{regular} model of type theory and the two-level model given by its presheaf category;
\item \cref{thm:conservativity}, providing a way to prove statements in HoTT using a two-level system;
\item \cref{thm:reedy-fibrant-replacement}, showing how to construct a Reedy fibrant replacement for any inverse diagram in a two-level system;
\item \cref{thm:fibrant-replacement-inconsistent}, exhibiting an inconsistency of a general \emph{fibrant replacement} operator in a two-level system with non-0-truncated fibrant types (see \cref{sec:hott}).
\end{itemize}

\subsection{Declaration of authorship}

\Cref{sec:fundaments} and \cref{sec:cwf} contain background material about semantic models of type theory.  Most of the definitions and results of these sections can be found in the literature, but their presentation has been reworked to fit with the constructions introduced later.

Most of the material of \cref{chap:type-theory} is joint work with Thorsten Altenkirch and Nicolai Kraus.  The definition of semi-simplicial types and, more generally, Reedy fibrant diagrams, and most of the preliminary content leading up to that, including parts of \cref{sec:infinite-structures}, have been published in \cite{two-level-type-theory}.

The rest of the thesis is original work of the present author.

\section{Related Work}\label{sec:related-work}

The main ideas of this thesis are inspired by Voevodsky's proposal of a \emph{homotopy type system} (HTS), which can be found in~\cite{hts}.

In~\cite{boulier:hts}, the authors present a version of a two-level type theory with a fibrant replacement operator, which would be inconsistent in the formulation of this thesis (\cref{thm:fibrant-replacement-inconsistent}), to derive a model structure on the universe of strict types.

A two-level type theory is developed in \cite{maietti-sambin:two-levels}.  Their motivation, however, is substantially different, hence the resulting theory has little resemblance with
the two-level type theory developed in this thesis.

A lot of work from several authors has recently gone into trying to develop a systematic and rigorous framework for working with models of type theory. \Cref{chap:models} contains one such (partial) attempt.  Similar work going in the same general direction can be found in~\cite{ahrens:monad-modules,ahrens:categorical-structures,palmgren:folds,bsystems,voevodsky:lawvere}.

\section{Fundaments of type theory}\label{sec:fundaments}

To motivate the definitions of \cref{chap:models} we will begin by exploring the basic concepts of intuitive type theory, and show how their desired properties translate directly into categorical structures.

\subsection{Contexts}

The fundamental notion of type theory is that of \emph{dependent type}.  For the idea of dependent type to even make sense, however, we first need to state what it is exactly that a type can depend on.  This is how we arrive to the notion of \emph{context}.

A context represents a list of assumptions, each assumption being essentially made up of variable name and a type.  Every theorem is always stated and proven relatively to some context.

Whenever, in informal mathematics, we say something like \emph{``let $n$ be a natural number, $R$ a commutative ring, and $M$ a free $R$-module of rank $n$''}, we are effectively defining a context $\Gamma$ containing the three variables $n$, $R$, and $M$, having the stated types.

This simple example already shows one important characteristic of contexts: the type of a variable is allowed to depend on previously introduced variables. That is, of course, essential if we want to model the idea of \emph{dependent} types.

Despite the intuition of contexts being essentially lists of pairs, in the following we will take a more axiomatic approach: we will take a collection of contexts $\C$ as given, and work out the structure that this collection ought to possess in order to model the intuitive idea described above.

\subsection{Morphisms}

It is natural to require that contexts form a category.

In fact, assumptions can intuitively be instantiated in the context given by some other assumptions.  For example, if $\Gamma$ denotes the context defined above, with variables $n$, $R$, $M$, and $\Delta$ is the context in which we have a natural number $m$, and field $k$, we can ``interpret'' $\Gamma$ into $\Delta$ by setting, for example,

\begin{morphism}\label{example:morphism}
n \mapsto m \\
R \mapsto k \\
M \mapsto k^m
\end{morphism}

This would define a morphism from $\Delta$ to $\Gamma$ in the category $\C$.  It will be clear in \cref{chap:models}, once we have a complete definition of CwF, how to make morphism definitions like \cref{example:morphism} precise.

The category $\C$ should have a (distinguished) terminal object $\term$.  We call $\term$ the \emph{unit context}, and think of it as the context where no assumptions have been made.  This is consistent with our interpretation, as there should be a unique way to instantiate the unit context in any other context.

\subsection{Types}

Now we can finally move to the central concept: types.  Given a context $\Gamma$, a type $A$ over $\Gamma$ should be defined as something that allows one to talk about:
\begin{itemize}
\item the \emph{context extension} $\Gamma.A$, which is to be thought of as the result of adding a new variable of type $A$ to the existing context $\Gamma$
\item the \emph{display map} $p_A : \Gamma.A \to \Gamma$, which is the interpretation of the extended context into the original one obtained by simply ``forgetting'' about the extra variable.
\end{itemize}

Note that the above data is exactly what is required to give an object of the slice category $\C / \Gamma$.  Therefore, any type should determine such an object.

This will be made precise in \cref{chap:models} in the context of a CwF. However, to motivate the general definition, we will first leave things at an intuitive level, assume that we have a way to map types over $\Gamma$ (whatever they are) to objects in $\C / \Gamma$, and investigate the structure and properties that this mapping should have.

\subsection{Terms}

Given a type $A$ over the context $\Gamma$, a \emph{term} $a$ of type $A$ is a morphism
$$
a : \Gamma \to \Gamma.A
$$
that is a section of the display map $p_A$, i.e. such that $p_A \circ a = \id$.

The idea of this definition is that a term of type $A$ is defined to be exactly what is required to give an interpretation of the extended context $\Gamma.A$ in the context $\Gamma$.  The property of being a section says that the interpretation does not touch any of the other assumptions.

To express the fact that $a$ is term of type $A$ over the context $\Gamma$, we will write the judgement \[ \Gamma \vdash a : A \] or simply $a : A$, when the context is clear.

For technical reasons, although terms can be regarded as a defined notion, we will take them as primitive in \cref{def:cwf} below.  Of course, the characterisation as sections is still valid, and will be proved as \cref{prop:tm-sections}.

\subsection{Substitutions}

Given a morphism $\sigma : \Delta \to \Gamma$, which we regard as a way to interpret the assumptions in $\Gamma$ in terms of the assumptions in $\Delta$, there should be a way to transport types and terms over $\Gamma$ to, respectively, types and terms over $\Delta$. In fact, if the context $\Gamma$ can be interpreted in $\Delta$, then everything we can state and prove in $\Gamma$ should make sense in $\Delta$ as well.

In particular, given a type $A$ over $\Gamma$, there should exist a type $A[\sigma]$ over $\Delta$, and a morphism $\sigma^+ : \Delta.A[\sigma] \to \Gamma.A$, which we refer to as $\sigma$ \emph{extended with} $A$.

The property of being able to transport terms of type $A$ to terms of type $A[\sigma]$ can be expressed concisely by requiring that the following square
\begin{equation} \label{eq:ext_pullback}
\xymatrix{
\Delta.f^* A \ar[r]^-{\sigma^+} \ar[d]_{p_{A[\sigma]}} &
\Gamma.A \ar[d]^{p_A} \\
\Delta \ar[r]^-\sigma &
\Gamma }
\end{equation}
be a pullback.

In fact, the commutativity property states that the extended morphism behaves like $\sigma$ on the assumptions in $\Delta$, while the universal property of the pullback is equivalent to saying that terms $a$ of type $A$ can be uniquely transported to terms of type $A[\sigma]$ in a way that is compatible with the extended morphism $\sigma^+$.

If $\C$ has (distinguished) pullbacks, every $\sigma : \Delta \to \Gamma$ determines a functor $-[\sigma] : \C / \Gamma \to \C / \Delta$, so the condition above can be expressed in any such category.  We refer to $-[\sigma]$ as the \emph{substitution} (or \emph{pullback}, or \emph{reindexing}) functor.

\subsection{Dependent products}\label{sec:intro-pi}

In order to define a notion of ``function'' internal to our system, we need to be able, given types $A$ and $B$ over some context $\Gamma$, to define a type $A \to B$, whose terms can be thought of as functions from $A$ to $B$.

More generally, given a type $A$ over $\Gamma$, and a type $B$ over $\Gamma.A$, we want to define a type of \emph{dependent functions} from $A$ to $B$, the so called \emph{dependent product} of $A$ and $B$, which we denote by $\Pi_A B$.

Terms of $\Pi_A B$ can be thought of as functions whose result \emph{type} depends on the argument.  Alternatively, one can think of $\Pi_A B$ as an internalised form of the categorical product of a family of types.

We define dependent products rigorously in \cref{def:pi-structure}, but for now, we can think of $\Pi_A B$ as defined by the fact its terms are in natural bijective correspondence with terms of type $B$ in the context $\Gamma.A$.  This expresses the idea that a function is completely characterised by its value on a ``generic'' element of its domain.

\subsection{Dependent sums}\label{sec:intro-sigma}

The idea of \emph{dependent sums} generalises the notion of binary product.

Given a type $A$ over $\Gamma$, and a type $B$ over $\Gamma.A$, the dependent sum of $A$ and $B$, denote $\Sigma_A B$, intuitively represents the type of all pairs of terms $a$ and $b$, where $a : A$ and $b : B[a]$.  Dually to dependent products, dependent sums can be thought of as an internal version of the coproduct of a family of types.

Again, we will later give a precise definition of $\Sigma$ (\cref{def:sigma-structure}), but for now, we can think of $\Sigma_A B$ as a type characterised by the fact that its terms are in bijective correspondence with pairs of terms as above.

\subsection{Equality}\label{sec:intro-equality}

The final essential idea that we will require in order to replicate basic logic and set theoretical constructions in our system is that of \emph{equality}.

The ``structural'' nature of the kind of system that we are set to create implies that we should only be allowed to consider equality between terms of the \emph{same} type.

Given a type $A$ in the context $\Gamma$, and terms $a, b : A$, we can then form an \emph{equality type} $a = b$.  This is the first point in our development where the type-theoretic incarnation of a concept differs substantially with its conventional set-theoretic counterpart.

In the usual classical foundations of mathematics (e.g. ZFC over some form of first-order logic), equality of sets is not itself a set, but a meta-theoretic entity. In other words, equality of mathematical objects is not itself a mathematical object.

One can of course remedy this somewhat by reifying equality into a set as follows: define the \emph{equality set} $[a=b]$ of $a$ and $b$ as the equaliser
\footnote{When working in a non-constructive meta-theory like ZF, the above definition can be simplified as follows: $[a,b]$ is defined to be 1 if $a = b$, and the empty set otherwise.}
$$
\xymatrix{
[a=b] \ar[r] & 1 \ar@<-.5ex>[r]_b \ar@<.5ex>[r]^a & X
}
$$
regarded as a subobject of a canonically specified terminal object 1 in $\set$ (e.g. the ordinal 1).

This is indeed one way to interpret type theoretic equality in terms of sets, but the advantage (or the curse, depending on how one looks at it) of the type-theoretic account is that it is much more general, and the notion of equality set described above is but one of many possible interpretations.

We will define one version of equality precisely in \cref{def:equality-structure}, and another one later in \cref{sec:rf-examples}. Again, for this introductory discussion, we will limit ourselves to an informal characterisation: the equality type $a = b$ is defined by the following two features:
\begin{itemize}
\item a canonical term $\refl : a = a$;
\item a ``substitution'' principle: a term $p : a = b$ can be used to reduce any construction involving $b$ (and possibly $p$ itself) into one in terms of $a$
(and $\refl$).
\end{itemize}

The first feature simply expresses the fact that every object should be equal to itself (and gives us a concrete \emph{witness} of the fact).  The second formalises the idea that equal objects are indistinguishable from within the theory.

\subsection{Propositions as types}\label{sec:propositions-as-types}

Equality as a type is but one example of a general pattern in type theory: propositions, i.e. statements \emph{about} mathematical objects, are themselves mathematical objects and can be studied as such.

The idea is that if a type $A$ is thought of as a proposition, then its terms are interpreted as \emph{witnesses} of the truth of $A$, or, in other words, as pieces of evidence for $A$.

Interestingly, all the structures introduced above have a sensible interpretation in terms of operations over propositions.  For example, if $A$ and $B$ are propositions, the type $\Pi_A B$ can be interpreted as the proposition stating that $A$ implies $B$: a witness of $\Pi_A B$, in fact, is a function that turns evidence for $A$ into evidence for $B$.

Similarly $\Sigma_A B$ corresponds to the logical conjuction of $A$ and $B$: a witness of $\Sigma_A B$ is a pair of witnesses for $A$ and $B$ respectively.

Furthermore, one can use dependent products and sums to reproduce the ideas of universal and existential quantification of logical theories.  For example, if $A$ is any type, and $B$ is thought of as a family of propositions indexed over $A$ (or, equivalently, a ``predicate'' over $A$), the dependent product $\Pi_A B$ corresponds to the assertion that $B$ holds for all the elements of $A$ (i.e. $\forall x : A, B$).  Dually, $\Sigma_A B$ serves as the assertion that there exists an element of $A$ for which $B$ holds (i.e. $\exists x : A, B$).

\subsection{Other structures}\label{sec:other-structures}

Unfortunately, the structures of dependent products, dependent sums and equality defined above, although very powerful and versatile, are often not enough to express certain mathematical ideas.  Examples of constructions that are not covered by those basic operations are: induction, disjoint unions, logical negation, quotients, and others.

For this reason, type theories usually include extra structures designed to deal with those requirements.  In the following, after giving precise definitions of the basic structures defined above, we will give a generic definition of \emph{type former} (\cref{sec:rule-framework}), encompassing most of the type-theoretic structures that are encountered in the literature.

This will allow us to work in a type theory (or model thereof) where the set of type formers is arbitrary, and does not need to be specified in advance.  That in turn will make some of our results very general, only subject to certain conditions on the type formers involved, which can then be verified separately and independently.

We will not discuss those extra type formers in detail.  We will define some of them in \cref{sec:rf-examples}, but only give a brief explanation.  We refer the interested reader to \cite[Chapter~1]{hott-book}.

\section{Homotopy Type Theory}\label{sec:hott}

The equality type $x = y$ introduced in \cref{sec:intro-equality} expresses the idea that the two elements $x$ and $y$ are ``identified'' in some sense, and they can be substituted for each other.

However, by itself, it has somewhat awkward features, which make it hard to use it effectively when formalising mathematics in type theory.

First of all, it is not well-behaved when it comes to describing equality of functions and equality of types.  For example, we cannot derive the principle of \emph{function extensionality}, stating that two functions are equal whenever they are equal at every point.  Therefore, this principle is usually taken as an axiom in most incarnations of type theory.

Secondly, the following question may come quite naturally after reading the informal definition of \cref{sec:intro-equality}: is every witness of equality equal to $\refl$?

A superficial reading of the substitution principle of equality (corresponding to the so-called $J$-eliminator, which we will introduce rigorously in \cref{sec:rf-examples}) would suggest this to be the case, since it says that proving a property of equality can be reduced to proving the corresponding property for $\refl$.

A careful examination, however, reveals a fault in this straightforward argument: given arbitrary terms $a, b : A$, and $p : a = b$, we cannot internally express the property of $p$ of being equal to $\refl$, because $p$ and $\refl$ have different types.  If we restrict ourselves to terms $p : a = a$, then our premise is not general enough, and we are not allowed to use the substitution principle.

In fact, it turns out that the question cannot be answered internally: it is consistent to assume that there exist proofs of equality which are not themselves equal to $\refl$ \cite{groupoid-model}.  This implies that equality cannot be simply thought of as a ``mere'' proposition, since it carries potentially non-trivial internal structure.

From here, one can either dismiss this limitation as a failure of the definition of equality, and address it by adding the missing component as an extra assumption (see \cref{eq:uip}), or embrace it, and fully explore its consequences.

Both approaches are viable, and have been pursued with great success.  The first makes it possible to encode most, if not all, of existing informal mathematics (at least, if we also assume certain classical principles such as the axiom of choice or the excluded middle).  It is very close in spirit to working within the Mitchell-B\'{e}nabou language of topoi, and it exists on a similar level of generality.  We will call such a theory \emph{strict}.

The second approach is embodied by HoTT \cite{hott-book}.  When no assumptions on the triviality of equality types is made, we can observe that types arrange themselves into a cumulative hierarchy of \emph{truncation levels}, starting with \emph{$-1$-types} (also called \emph{propositions}), whose equality types are completely trivial, followed by \emph{$0$-types}, or \emph{sets}, having propositions as equality types, and in general $n$-types, defined as those types whose equality types are $(n-1)$-types.

One appeal of $\hott$ is that equalities can be seen as paths in a space, and it is even possibly to develop substantial amounts of homotopy theory synthetically (see for example~\cite{brunerie:thesis} for an extensive account).
An important fact to keep in mind is that, when doing homotopy theory in type theory, every statement that one can make holds up to homotopy, and every construction respects (homotopy) equivalence.

This means that whatever we do will be ``invariant'', in the sense that it can only take the homotopy type of spaces, and homotopy equivalence classes of maps, into account, and not the concrete representations of spaces or maps.
This is often considered a selling point of $\hott$: one might perform constructions using representatives of homotopy classes in traditional homotopy theory, which make it necessary to show that those constructions are well-defined, i.e.\ do not depend on the choice of the representative.

In $\hott$, everything is automatically well-defined up to homotopy as we are simply not able to talk about non-homotopy-invariant notions like strict equality internally.

\section{The problem of ``infinite structures''} \label{sec:infinite-structures}

It is not hard to imagine that the blessing of having only constructions up to homotopy can turn out to be a curse: the inability to reflect a notion of ``strict equality'' into the theory can sometimes make certain ideas much harder to express.

For example, we cannot form a type expressing that a given diagram commutes strictly; all we can do is stating that it commutes up to homotopy.
Unfortunately, depending on the shape of the diagram, this will only be sufficient in the simplest cases.
More often than not, it will be necessary to say that the different ``pieces'' (the equalities expressing commutativity) fit together.

For instance, the fact that a certain sub-diagram commutes can be part of the proof that the diagram commutes, but it may at the same time be derivable as the composition of the fact that other sub-diagrams commute.
In this case, it is natural to require these different ways of getting a certain proof to be equal.
It does not stop here; these new proofs can themselves be required to be coherent, and so on.

This phenomenon is of course not something that can only be observed in type theory.  The first step becomes already apparent in the theory of monoidal categories in the form of ``Mac Lane's Pentagon''.
On higher dimensions, it is exactly the same issue that is discussed as \emph{homotopy commutativity versus homotopy coherence} by Lurie~\cite{lurie-topos}.

In general, homotopy coherence corresponds to infinite towers of coherence data, and it is a major open problem (and commonly believed to be unsolvable) to express such towers internally in $\hott$.
One way to avoid the problem altogether is to restrict constructions to types of low truncation levels.
As an example, the category theory developed in~\cite{aks} only considers $1$-truncated types to develop a theory of ordinary categories.
This is in many situations not satisfactory: we know that types are $\infty$-groupoids~\cite{lumsdaine:phd,bg:type-wkom}, and similarly, the universe should be an $(\infty,1)$-category.
Unfortunately, there does not seem be a way to express this internally in $\hott$.

Of course, it is always possible to take one of the existing models of higher categories and replicate it internally in HoTT.
However, since all of the existing models are ultimately built out of sets, this would force the HoTT version to be based on \emph{sets} as well (i.e. $0$-truncated types), which means that many specific structures that are expected to be $(\infty,1)$-categories would not qualify.  One notable example is provided by universes, which cannot in general be assumed to be truncated (as shown in \cite{kraus-sattler:universes}), hence cannot possibly be given a categorical structure for any notion of higher category which is based on sets.
On the other hand, we define an $(\infty,1)$-category structure for a universe in \cref{sec:segal-examples}.

The crucial shortcoming of $\hott$ is that we are unable to encode certain constructions which would appear to be harmless, as they only require finite amounts of coherence data at every step.
An example that has received considerable attention in the $\hott$ community is the construction of Reedy fibrant $n$-semi-simplicial types (simply referred to as \emph{semi-simplicial types}).

Let us start with $\deltplus$, the category of finite non-zero ordinals and strictly monotone functions.
Let us write $\ordinal n$ for the ordinal with $n+1$ elements.
A type-valued diagram over $\deltop$ is a \emph{strict} functor from $\deltop$ to the category of types.
It would correspond to a type $X_{\ordinal n}$ (for simplicity written $X_n$) for every $n$, and face maps $d_i : X_{n+1} \to X_{n}$ for $0 \leq i \leq n$, as it is well-known that any map in $\deltop$ can be written as a composition of face maps.
The problem is that we need the semi-simplicial identities (essentially a
representation of the functor laws) to be strict, a fact which we cannot express in type theory.

The considered approach to avoid this problem is to only attempt internalising \emph{Reedy fibrant} diagrams over $\deltop$, essentially ensuring that the face maps are simple projections.

Using the correspondence between fibrations and type families, a (Reedy fibrant) semi-simplicial type then corresponds to a type $X_0$ (the ``points'') on level $0$.
On level $1$, we need a family
$$
X_1 : X_0 \to X_0 \to \U,
$$
where $\U$ is the universe of types.
We think of $X_1$ as lines between types.
Next, we need
$$
X_2 : \prd{a,b,c: X_0} X_1(a,b) \to X_1(b,c) \to X_1(a,c) \to \U,
$$
the type of fillers for triangles.

Writing down the type of $X_4$ is already rather tedious, but nevertheless straightforward: $X_4$ is a family which gives a type for any collection of four points, six lines and four triangles that form a boundary of a tetrahedron.

A long-standing open problem of homotopy type theory is then to write down the type of $X_n$, or something equivalent to it, for a general natural number $n$.  This has revealed to be much harder than one might expect, and it is actually conjectured to be impossible.

What is definitely possible is to generate an expression $X_n$ for every externally fixed numeral $n$, such that the expressions $X_0, X_1, X_2, \ldots$ all ``fit together''.
If one attempts to perform the same construction for a \emph{variable} $n : \nat$, the types do not match up anymore.
The reason is that some strict equalities that hold in the case of a numeral $n$ fail to hold in the case of a variable.
One could try to prove that the required equalities hold up to homotopy, but one quickly realises that one would also need to show that these equalities are coherent, and that the coherence proofs are coherent themselves, and so on; even only expressing the coherence data that is required to make the construction go through seems to be as hard, if not harder, than the original problem.

\section{Internalising strict equality}
\label{sec:internalising-strict-equality}

In some sense, the equalities needed when attempting to construct semi-simplicial types, as explained in \cref{sec:infinite-structures}, \emph{should} hold and be fully coherent, because they are trivially satisfied for each externally fixed natural number.
If only we had a way to reason about strict equalities \emph{within} the system, there would be no problem at all; however, this would require strict equalities to be reified into a type.

We could take the equality of a strict theory to be the internalised version of strict equality.  In that case, it would be possible to construct Reedy fibrant semi-simplicial types internally.  However, we can also simply define categories and functors in the usual sense, and all coherences will be satisfied automatically thanks to the strictness assumptions in the theory.

Using this approach, we would bypass all the coherence problems, but have to give up all the advantages of $\hott$, like univalence and higher inductive types.  The idea of a two-level system is to combine strict type theory and $\hott$, instead of viewing them as two alternative extensions of the basic underlying type theory.

A two-level type theory consists of two ``parallel'' type theory, with possibly different structures, sharing a small common \emph{core} consisting of dependent products and sums.  We call the two fragments \emph{strict} and \emph{fibrant} respectively.
The strict fragment is, unsurprisingly, a strict form of type theory, while the fibrant fragment is an incarnation of $\hott$.  Every fibrant type can be canonically regarded as a strict type, but not vice versa.

The reason why two-level type theory has to be set up in this way, rather than just having two equality types, is \cref{lem:all-fibrant-collapse}, showing that if there is no distinction between fibrant and strict types, then the two equalities necessarily collapse into one.

The idea a type theory with two equality types is not new.  Such a system was first suggested by Voevodsky~\cite{hts}, who referred to it as HTS, but the theory developed in this thesis (specifically in \cref{chap:two-levels}) presents substantial differences with HTS (see \cref{sec:hts}).  In particular, it requires no form of equality reflection in its strict fragment. Thus, we can avoid all the problems that are usually connected to equality reflection, such as undecidability of type checking.

In contrast, the two-level system presented in this thesis is well-behaved, very close to the standard formulation of $\hott$, and has straightforward semantics.
One could expect that a downside of our system might be reduced expressibility compared to a theory that features equality reflection.
However, we can achieve in our system what HTS was suggested for: a definition of semi-simplicial types, and other constructions based on them.

Furthermore, by being careful about the relationship between strict and fibrant type formers, we can prove a conservativity result (\cref{thm:conservativity}).
This means that, in some sense, the fibrant fragment corresponds exactly to $\hott$ as presented in~\cite{hott-book}.
In a proof assistant which supports this theory, we could in principle implement results that so far can only be stated meta-theoretically.
To give an example, it is shown in~\cite{kraus:general-universal-property} that constant functions from $A$ to $B$ which satisfy $n$ coherence conditions correspond to maps $\|A\| \to B$, provided that $B$ is $n$-truncated.  Here $n$ is a natural number, external to the theory, so the result has to be formalised as a \emph{sequence} of internal statements, which means that it can only be stated and proved meta-theoretically.
In a two-level system, we can formalise it by taking $n$ to be an element of the strict type of natural numbers, then show the required equivalence in the fibrant fragment.  Conservativity would then allow us to conclude the the corresponding statement is valid in $\hott$ for all choices of the parameter $n$, and all the complications of meta-theoretic reasoning would be encapsulated in the proof of \cref{thm:conservativity}.

\chapter{Type theory and type formers}\label{chap:models}

This chapter contains the fundamental definitions and constructions that will be used throughout the rest of the thesis.
We will start from the intuitive ideas presented in \cref{chap:introduction}, and make them precise in terms of \emph{categories with families}, which we choose as the primary basic notion of model of type theory.

Our presentation of basic type formers ($\Pi$, $\Sigma$, equality and unit type) is based on the same ideas as in \cite{natural-models}, which will make it easier to extend the notion of type former to more general operations, as well as to the context of \cref{chap:two-levels}.

\section{Categories with families}\label{sec:cwf}

\begin{defn}[see \cite{cwf}]\label{def:cwf}
A \emph{category with families} (CwF) is given by:
\begin{itemize}
\item a category $\C$, equipped with a distinguished terminal object $1$;
\item a presheaf $\Ty : \C \to \op\set$;
\item a presheaf $\Tm : \left(\int \Ty\right) \to \op\set$;
\item for all $\Gamma : \C$ and $A : \Ty(\Gamma)$, an object $(\Gamma.A, \pi_A)$ representing the functor $\left(\C/\Gamma\right) \to \op\set$ defined by:
\begin{equation}\label{eq:ext-repr-functor}
(\Delta, \sigma) \mapsto \Tm_\Delta(A[\sigma]).
\end{equation}
\end{itemize}

Here and in the following, if $X : \C \to \op\set$ is a presheaf on a category $\C$, $\sigma : \C(\Delta, \Gamma)$ is a morphism, and $x : X_\Gamma$ is an element of $X$, we write $x[\sigma]$ instead of $X(\sigma)(x)$.

The objects of $\C$ are called \emph{contexts}. Given a context $\Gamma$, the elements of $\Ty(\Gamma)$ are called \emph{types}, and given a type $A$, the elements of $\Tm_\Gamma(A)$ are called \emph{terms}.

The context $\Gamma.A$ is called the \emph{context extension} of $\Gamma$ by the type $A$, and $\pi_A$ is the \emph{display map} of $A$.

The action of $\Ty$ and $\Tm$ on morphisms is called \emph{substitution}. \end{defn}

Note that, given a morphism $\sigma : \C(\Delta, \Gamma)$, a type $A : \Ty(\Gamma)$, and a term $a : \Tm_\Delta(A[\sigma])$, the definition of CwF gives a corresponding morphism $\C(\Delta, \Gamma.A)$ which we will denote by $\langle \sigma, a \rangle$.

\begin{prop}\label{prop:tm-sections}
For all contexts $\Gamma : \C$ and types $A : \Ty(\Gamma)$, there is a natural isomorphism:
\begin{equation}\label{eq:tm-sections}
\Tm_\Gamma(A) \cong \C/\Gamma(\Gamma, \Gamma.A).
\end{equation}
\end{prop}
\begin{proof}
Equation \cref{eq:tm-sections} follows directly from the definition of context extension, by taking $\Delta :\equiv \Gamma$ and $\sigma :\equiv \id$ in \cref{eq:ext-repr-functor}.
\end{proof}

\Cref{prop:tm-sections} says that terms of type $A$ can be equivalently regarded as sections of the display map $\pi_A : \C(\Gamma.A, \Gamma)$.

\begin{prop}\label{prop:morphism-weakening}
Let $\sigma : \C(\Delta, \Gamma)$ be any morphism, and $A : \Ty(\Gamma)$. There exists a morphism $\sigma^+ : \C(\Delta.A[\sigma], \Gamma.A)$ that makes the square
\begin{equation}\label{eq:ctx-ext-pullback}
\xymatrix{
\Delta.A[\sigma] \ar[r]^-{\sigma^+} \ar[d]_-{\pi_{A[\sigma]}} & \Gamma.A \ar[d]^-{\pi_A} \\
\Delta \ar[r]_-{\sigma} & \Gamma
}
\end{equation}
into a pullback.
\end{prop}
\begin{proof}
Diagram \ref{eq:ctx-ext-pullback} being a pullback is equivalent to the condition that, for all contexts $\Phi$ and morphisms $\tau : \C(\Phi, \Delta)$, there is a natural isomorphism:
$$
\C/\Gamma(\sigma_* \Phi, \Gamma.A) \cong \C/\Delta(\Phi, \Delta.A[\sigma]),
$$
where we write $\Phi$ to mean the pair $(\Phi, \tau)$ in the slice category $\C/\Delta$, and similarly for $\Gamma.A$ and $\Delta.A[\sigma]$.

But clearly, the isomorphism holds, since both sides are naturally isomorphic to $\Tm_\Phi(A[\sigma\circ\tau])$, by the defining property of context extension.
\end{proof}

\Cref{prop:morphism-weakening} allows us to turn the context extension operation into a functor $\ext : \int\Ty \to \C$.

\begin{defn}\label{def:cwf-morphism}
Let $\C$, $\D$ be CwFs. A \emph{CwF morphism} $\C \to \D$ is given by:
\begin{itemize}
\item a functor $F : \C \to \D$;
\item a natural transformation $F^{\Ty} : \int_{\Gamma} \Ty(\Gamma) \to \Ty(F \Gamma)$;
\item a natural transformation $F^{\Tm}: \int_{\Gamma, A} \Tm_{\Gamma}(A) \to \Tm_{F\Gamma}(F^{\Ty} A)$;
\end{itemize}
such that $F1$ is a terminal object in $\D$, and, for all $\Gamma : \C$ and $A : \Ty(\Gamma)$, the map
$$
\phi^F_A : \D(F(\Gamma.A), F\Gamma.F^\Ty A)
$$
defined below is an isomorphism.
\end{defn}

The map $\phi^F_A$ is obtained as follows. First, by applying the functor $F$ to the display map $p_A$, the context $F(\Gamma.A)$ can be regarded as an element of $\D / F\Gamma$.
Then, the term $F^\Tm(v_A) : \Tm_{F(\Gamma.A)}(F^\Ty A [F(p_A)])$ determines a morphism $\D/F\Gamma(F(\Gamma.A), F\Gamma.F^\Ty A)$ by the defining property of context extension, and $\phi^F_A$ is taken to be the corresponding underlying morphism $\D(F(\Gamma.A), F\Gamma.F^\Ty A)$.

We will usually omit the superscripts $\Ty$ and $\Tm$ when referring to the
action of a morphism on types and terms respectively.

\begin{defn}
A CwF morphism $F : \C \to \D$ is said to be \emph{split} if it preserves the distinguished terminal objects and context extension ``on the nose'' and the map $\phi^F_A$ is the identity for all types $A$.
\end{defn}

\begin{defn}
A CwF morphism $F : \C \to \D$ is said to be a \emph{CwF equivalence} if it is an equivalence of categories, and it induces isomorphisms on types.
\end{defn}

Note that a CwF equivalence automatically induces isomorphisms on terms.

\subsection{Notation}\label{sec:notation}

In the following, let $\C$ be a CwF.

If $\Gamma$ is a context, and $A : \Ty(\Gamma)$, the universal property of the context extension, applied to the identity substitution $\C/\Gamma(\Gamma.A, \Gamma.A)$, yields a canonical term $v_A : \Tm_{\Gamma.A}(A[\pi_A])$.  We call $v_A$ the \emph{variable} of type $A$.

Weakenings, i.e. substitutions along display maps, will often be omitted from the notation, as they can usually be unambiguously reconstructed, and leaving them implicit simplifies the syntax considerably. In particular, the variable of type $A$ can be regarded simply as a term in $\Tm_{\Gamma.A}(A)$.

Sometimes, when building contexts using context extension, we will associate ``names'' to certain types.  These names will be used to refer to their corresponding variables, and weakenings thereof.  For example, the context $\Gamma(a : A)$ denotes the context $\Gamma.A$, with the convention that the name $a$ refers to the variable $v_A : \Tm_{\Gamma(a : A)}(A)$.

The terminal object of $\C$ is referred to as the \emph{unit} context.%
\footnote{In traditional type-theoretic terminology, the term \emph{empty context} is more often found. This is because contexts are usually built explictly by chaining a finite number of context extensions, and 1 is the base case of this process, where no extensions have been performed yet. However, ``empty'' is more suggestive of an initial, rather than terminal, object, so we will keep consistency with the corresponding terminology for types, and use the term \emph{unit context} instead}
We will identify types in the unit context with the corresponding contexts obtained by context extension.  So, for example, if $A : \Ty(1)$, we will write $A$ to denote $1.A$, and if $B : \Ty(A)$, we can form the context extension $A.B$.

Finally, thanks to \cref{prop:tm-sections}, terms in $\Tm_\Gamma(A)$ correspond bijectively with sections of the display map $\pi_A$. We will therefore identify a term with its corresponding section.

With those syntactical conventions, working in an arbitrary CwF is basically indistinguishable from working in the corresponding type theory (i.e. its internal language). For that reason, we are able to avoid giving a precise definition of \emph{syntax} of type theory. Our definitions and constructions exist purely within the semantics realm of CwFs, and that is sufficient for our purposes.

We will also implicitly assume the existence of a hierarchy of an arbitrary finite number of universes of sets $\set_0 \subseteq \set_1 \subseteq \set_2 \ldots$, but remove the indices from the notation. In particular, we will simply write $\set$ instead of $\set_0$ or $\set_1$. This is in line with a widespread convention in type theory called ``typical ambiguity''~\cite{feferman:typical-ambiguity}, and is used, for example, in \cite{hott-book}.

The existence of this hierarchy of universes may depend on certain large cardinal axioms (like the existence of a corresponding chain of innaccessible cardinals) in a foundations like ZFC. Alternatively, if we assume that the metatheory that we are working in is itself some form of type theory, then all we need is a tower of universes (as in \cref{def:universe}) in the outer theory.

\subsection{Presheaves}\label{sec:presheaves}

The prototypical example of a CwF is the category of presheaves over $\C$, where $\C$ is an arbitrary (small) category.  We will denote this category by $\presheaf\C$.  For any presheaf $P$, let $\bTy(P)$ be the category of presheaves over $\int^{\C}P$, and let $\Ty(P)$ be the underlying set of objects of $\bTy(P)$.

Clearly, $\bTy$ defines a functor $\op{\C} \to \cat{Cat}$, hence $\Ty$ is a functor $\op{\C} \to \set$. The corresponding term functor is given by:
$$
\Tm_P(A) :\equiv \bTy(P)(1, A),
$$

where 1 is the terminal object of $\bTy(P)$, i.e. the functor which is constantly equal to the terminal object 1 of $\set$.  Substitutions are defined in the obvious way via precomposition.

To define context extension, we will need the following
\begin{prop}\label{prop:presheaf-families}
Let $\C$ be any category, and $P : \presheaf{C}$ a presheaf on $\C$. There is an equivalence of categories:
$$
\Phi : \presheaf{C} / P \cong \presheaf{\int P}
$$
such that, for all presheaves $Q$ over $P$, there is an isomorphism of categories:
\begin{equation}\label{eq:presheaf-families}
\int \Phi(Q) \cong \int Q
\end{equation}
\end{prop}
\begin{proof}
Given a presheaf $Q$ over $P$, define a presheaf $\Phi(Q)$ on $\int^{\C}P$ by assigning to every object $(\Gamma, x)$ of $\int^{C}P$, where $\Gamma : \C$ and $x : P_\Gamma$, the fibre of $Q$ over $x$.

Conversely, given a presheaf $F : \presheaf{\int^{C}P}$, define $Q_\Gamma$ as the set of pairs $(x, y)$, where $x : P_\Gamma$, and $y : F(\Gamma, x)$.

It is easy to see that $\Phi$. defines an equivalence of categories.  As for equation \cref{eq:presheaf-families}, it follows immediately from the definition of $\Phi$.
\end{proof}

Now, given a presheaf $P$ and a type $A$ over $P$, define $P.A$ to be the presheaf over $P$ corresponding to $A$ through the equivalence of
\cref{prop:presheaf-families}, so that we have equivalences:
\begin{equation}\label{eq:presheaf-extension}
\bTy(P.A) \cong \presheaf{\int A} \cong \bTy(P)/A,
\end{equation}
where the first is a consequence of the isomorphism \cref{eq:presheaf-families}, and the second is obtained by applying \cref{prop:presheaf-families} to the category $\int^{C} P$.  We will call $P.A$ the \emph{total space} of $A$.

Therefore, we can associate, to any type in $B : \Ty(P.A)$, a corresponding type in $\Ty(P)$, which we will denote by $\Sigma_A B$.  Note that $P.\Sigma_A B \cong P.A.B$.

\begin{lem}
The map $B \mapsto \Sigma_A B$ defines a left adjoint for the substitution functor $\bTy(P) \to \bTy(P.A)$ along $\pi_A$.
\end{lem}
\begin{proof}
The functor $\Sigma_A$ can be regarded as the composition:
$$
\Sigma_A : \bTy(P.A) \to \bTy(P)/A \to \bTy(P),
$$
where the first functor is the equivalence \ref{eq:presheaf-extension}, and the second is the forgetful functor.

The latter has a right adjoint, mapping a type $C : \Ty(P)$ to the product $A \times C$, together with the first projection.

Therefore, all is left to do is to verify that $A \times C$ corresponds to $C[\pi_A]$ through the equivalence \ref{eq:presheaf-extension}, which is easy to see.
\end{proof}

Note that $\bTy(P)$, being a presheaf category, is a cartesian closed category with all small limits and colimits.  In particular, given two types $A, B$, we can form their exponential $B^A$, which we can think of as the ``function type'' between $A$ and $B$.

We will now generalise this notion of function type to the situation where $B$ ``depends on $A$'', i.e. when $B$ is not in $\bTy(P)$, but in $\bTy(P.A)$.

Given $B : \Ty(P.A)$, we can obtain a type $\Sigma_A B : \Ty(P)$, together with a projection $\pi_1 : \Ty(P)(\Sigma_AB, A)$.  Since $\Ty(P)$ has limits, we can form a pullback square:
\begin{equation}\label{eq:presheaf-pi-pullback}
\xymatrix{
\Pi_A B \ar[r] \ar[d] &
\left(\Sigma_A B\right)^A \ar[d] \\
1 \ar[r] &
A^A,
}
\end{equation}
where the bottom arrow selects the identity morphism $A \to A$.

This determines a type $\Pi_A B : \Ty(P)$.

\begin{lem}\label{lem:pi-adjunction}
The map $B \mapsto \Pi_A B$ defines a right adjoint for the substitution
functor $\bTy(P) \to \bTy(P.A)$ along $\pi_A$.
\end{lem}
\begin{proof}
Let $X$ be an arbitrary type in $\Ty(P)$, and consider the homset $\bTy(P.A)(X[\pi_A], B)$. Through the equivalence \cref{eq:presheaf-extension}, this is isomorphic to $\left(\bTy(P)/A\right)(A \times X, \Sigma_A B)$, which fits into a pullback square:
$$
\xymatrix{
\left(\bTy(P)/A\right)(A \times X, \Sigma_A B) \ar[r] \ar[d] &
\bTy(P)(A \times X, \Sigma_A B) \ar[d] \\
1 \ar[r] &
\bTy(P)(A \times X, A).
}
$$

Using the adjunction defining the exponential, this diagram is isomorphic to:
$$
\xymatrix{
\left(\bTy(P)/A\right)(A \times X, \Sigma_A B) \ar[r] \ar[d] &
\bTy(P)(X, \left(\Sigma_A B\right)^A) \ar[d] \\
1 \ar[r] &
\bTy(P)(X, A^A).
}
$$

However, by applying the limit-preserving functor $\bTy(P)(X, -)$ to \ref{eq:presheaf-pi-pullback}, we get the same diagram, but with $\bTy(P)(X, \Pi_A B)$ in the top left corner.  Therefore, it follows that there is a natural isomorphism
$$
\bTy(P)(X, \Pi_A B) \cong \bTy(P.A)(X[\pi_A], B),
$$
hence $\Pi_A$ is right adjoint to substitution along $\pi_A$.
\end{proof}

As an immediate consequence of \cref{lem:pi-adjunction}, there is a natural isomorphism:
\begin{equation}\label{eq:lambda-abstraction}
\lambda : \Tm_{P.A}(B) \to \Tm_P(\Pi_A B),
\end{equation}
which is often referred to as \emph{lambda abstraction}.  Furthermore, given terms $f : \Tm_P(\Pi_A B)$ and $a : \Tm_P(A)$, we get a term $\lambda^{-1}(f)[a] : \Tm_P(B[a])$.  It is customary to denote this term simply by $f\ a$, and call this operation \emph{application}.

Alternatively, we can regard application as a morphism $\epsilon_{A,B}$:
$$
\xymatrix{
  \Gamma.A.\Pi_A B \ar[rr]^{\epsilon_{A,B}} \ar[rd] & &
  \Gamma.A.B \ar[ld] \\
  & \Gamma.A.
}
$$

Since the type $B$ appearing in a $\Pi_A B$ is defined over an extended context, it is often convenient to introduce a name for the variable of type $A$, when constructing such an expression.  Therefore, we will employ the notation:
$$
\Pi_{a : A} B,
$$
to mean the exact same thing as $\Pi_A B$, with the addition that $B$ is assumed to be a type in the context $P(a : A)$, i.e. the name $a$ refers to the variable of type $A$ within the expression that defines $B$.  A similar notation will be used for $\Sigma$.

We will now define a very simple notion of \emph{equality type} for presheaves.

Let $P$ be a presheaf, and $A : \Ty(P)$ a type over it. Consider the diagonal morphism $\bTy(P)(A, A \times A)$ and map it through the equivalence of \cref{prop:presheaf-families} to get a morphism in $\presheaf{C}(P.A, P.(A\times A))$, which is isomorphic to $\presheaf{C}(P.A, P.A.A)$.  Using \cref{prop:presheaf-families} again, this morphism determines a type over $P.A.A$ which we will denote by $\Eq_A$, and refer to as the \emph{equality type} of $A$.

In particular, given terms $a_1, a_2 : \Tm_P(A)$, we can form a type $\Eq_A[a_1, a_2]$ by substitution.  Terms of this type are witnesses of equality betwee $a_1$ and $a_2$, hence this type is inhabited (i.e. it has a global section) if and only if $a_1$ and $a_2$ are equal terms.

\begin{lem}\label{lem:eq-subterminal}
The type $\Eq_A$ is a subterminal object of $\bTy(P.A.A)$.
\end{lem}
\begin{proof}
Since equivalence of categories preserves subterminality, it is enough to show that the diagonal $A \to A \times A$ is subterminal in $\bTy(P)/(A\times A)$.

Let now $\C$ be any category, and $A : \C$ an object such that the product $A \times A$ exists.  The diagonal $\delta : A \to A \times A$ is the equaliser of the two projections $A \times A \to A$, hence it is monic.  Since the forgetful functor $\C/(A\times A) \to \C$ is faithful, it follows that $\delta \to \id$ is monic in $\C/(A \times A)$, i.e. $\delta$ is subterminal.
\end{proof}

\subsection{Basic type formers}\label{sec:basic-type-formers}

In the previous section, we defined the operations $\Sigma$, $\Pi$ and $\Eq$ on types of a presheaf category.  We will now define what it means for a general CwF to support those operations.

The following definitions are standard (see for example \cite{hofmann:syntax-and-semantics}).

\begin{defn}\label{def:pi-structure-standard}
We say that a CwF \emph{supports $\Pi$-types} if for any two types $A : \Ty(\Gamma)$ and $B : \Ty(\Gamma.A)$ there is a type $\pi(A,B) : \Ty(\Gamma)$, and for each $b : \Tm_{\Gamma.A}(B)$ there is a term $\lambda(b)$, and for each $f : \Tm_{\Gamma}(\pi(A,B))$ and $a : \Tm_{\Gamma}(A)$ there is a term $f \cdot a : \Tm_\Gamma(B[a])$ such that the following equations (appropriately quantified) hold:
\begin{equation*}
\begin{aligned}
& \lambda(b) \cdot a = b[a] \\
& \lambda (f \cdot v_A) = f \\
& \pi(A,B)[\tau] = \pi(A[\tau], B[\tau^+]) \\
& (\lambda(b))[\tau] = \lambda (b[\tau]) \\
& (f \cdot a)[\tau] = f[\tau] \cdot a[\tau].
\end{aligned}
\end{equation*}
\end{defn}

\begin{defn}\label{def:sigma-structure-standard}
We say that a CwF \emph{supports $\Pi$-types} if for any two types $A : \Ty(\Gamma)$ and $B : \Ty(\Gamma.A)$ there is a type $\sigma(A,B) : \Ty(\Gamma)$, and for each $a : \Tm_\Gamma(A)$ and $b : \Tm_\Gamma(B[a])$ there is a term $\langle a, b \rangle : \Tm_\Gamma(\sigma(A,B))$, and for all terms $x : \Tm_\Gamma(\sigma(A,B))$ there are terms $\pi(x) : \Tm_\Gamma(A)$ and $\pi'(x) : \Tm_\Gamma(B[\pi(x)])$ such that the following equations (appropriately quantified) hold:
\begin{equation*}
\begin{aligned}
& \pi(\langle a, b \rangle) = a \\
& \pi'(\langle a, b \rangle) = b \\
& \langle \pi(x), \pi'(x) \rangle = x \\
& \sigma(A,B)[\tau] = \sigma(A[\tau], B[\tau^+]) \\
& \langle a , b \rangle[\tau] = \langle a[\tau], b[\tau] \rangle \\
& \pi(x)[\tau] = \pi(x[\tau]) \\
& \pi'(x)[\tau] = \pi(x'[\tau]).
\end{aligned}
\end{equation*}
\end{defn}

\begin{defn}\label{def:equality-structure-standard}
We say that a CwF \emph{supports equality types} if for all types $A : \Ty(\Gamma)$ there is a type $\eq(A) : \Ty(\Gamma.A.A)$, such that two terms $a, b : \Tm_\Gamma(A)$ are equal if and only if there is a term $p : \Tm_\Gamma(\eq(A)[a, b])$, and furthermore:
$$
\eq(A)[\tau^{++}] = \eq(A[\tau]).
$$
\end{defn}

\begin{defn}\label{def:unit-structure-standard}
We say that a CwF \emph{has a unit type} if there exists a type $1 : \Ty(1)$ with a unique term.
\end{defn}

The purpose of this section is to develop equivalent formulations of the above definitions based on presheaves.  In \cref{sec:rule-framework}, we will introduce the \emph{rule framework}, and that will help us generalise the presheaf-based definitions (\cref{def:pi-structure}, \cref{def:sigma-structure}, \cref{def:equality-structure} and \cref{def:unit-structure}) to cover a wide variety of ``type formers''.

\begin{defn}\label{def:universe}
Let $\C$ be a CwF. A \emph{universe} in $\C$ is given by:
\begin{itemize}
\item a type $\U$ in the unit context;
\item a type $\El$ in the context $\U$.
\end{itemize}
\end{defn}

We will see later how universes of sets determine universes in presheaf categories for an arbitrary $\C$ (\cref{sec:presheaf-universes}).  For now, we will focus on the case where $\C$ is itself a CwF.  In that case, the presheaf category $\presheaf{\C}$ has a canonical universe, given by the functors $\Ty$ and $\Tm$, part of the CwF structure of $\C$.  For reasons that will be clear later, we will call this the \emph{fibrant universe} of $\presheaf\C$.

Since now we have two CwFs in play, in an attempt to avoid confusion, we will use the notation $\presheaf{\Ty}$ and $\presheaf{\Tm}$ when discussing the CwF structure on $\presheaf\C$.

In the following, we will write $y$ for the Yoneda embedding $\C \to \presheaf\C$.

\begin{lem}\label{lem:yoneda-terms}
Let $P$ be a presheaf on $\C$, $A$ a term of type $\Ty$ in the context $P$ of $\presheaf\C$, and $x$ an element of $P$ over some $\Gamma : \C$.  Let us write $\pi$ for the display map of the type $\Tm[A]$ over $P$.

There is an isomorphism of types over $P.\Tm[A]$:
\begin{equation}\label{eq:yoneda-terms}
y(\Gamma, x)[\pi] \cong y(\Gamma.A_\Gamma(x), x[\pi], v_{A_\Gamma(x)}),
\end{equation}
natural in $(\Gamma, x) : \int P$.
\end{lem}
\begin{proof}
We will construct the required isomorphism by using \cref{prop:presheaf-families} to transport all the presheaves involved to $\presheaf \C$.

By the Yoneda lemma, we can regard $x$ as a morphism $y(\Gamma) \to P$.  The left side of \ref{eq:yoneda-terms} is then isomorphic to the type over $y(\Gamma)$ obtained by substituting $\Tm[A]$ along $x$.

As for the right side, its total space can also be regarded as a presheaf over $y(\Gamma)$ through the Yoneda embedding of the display map $\Gamma.A(x) \to \Gamma$.

By \cref{prop:presheaf-families}, presheaves over $y(\Gamma)$ correspond to presheaves on $\int y(\Gamma)$, which is isomorphic to $\C / \Gamma$.  Applying the isomorphism of \cref{prop:presheaf-families} explicitly, it is easy to see that the left side is mapped to the functor given by \cref{eq:ext-repr-functor} for the type $A(x)$, so the conclusion follows from the defining property of context extension. 
\end{proof}

\begin{cor}\label{cor:ty-ext}
Let $P$ be a presheaf on $\C$, and $A$ a term of type $\Ty$ in the context $P$ of $\presheaf\C$.  The type
$$
\Pi_{\Tm[A]} \Ty
$$
is isomorphic to the presheaf on $\int P$ given by:
\begin{equation}\label{eq:ty-ext}
(\Gamma, x) \mapsto \Ty(\Gamma.A_\Gamma(x)).
\end{equation}
\end{cor}
\begin{proof}
Again, let us write $\pi$ for the display map of $\Tm[A]$.

Fix an arbitrary $(\Gamma, x) : \int P$. By \cref{lem:pi-adjunction}, there is a natural isomorphism:
$$
\presheaf\bTy(P)(y(\Gamma, x), \Pi_{\Tm[A]} \Ty) \cong
\presheaf\bTy(P.\Tm[A]) (y(\Gamma,x)[\pi], \Ty).
$$

By \cref{lem:yoneda-terms}, the weakened type $y(\Gamma,x)[\pi]$ is isomorphic to the representable presheaf $y(\Gamma, x, A(x))$, hence the conclusion follows from the Yoneda lemma.
\end{proof}

In the setting of \ref{cor:ty-ext}, if $B$ is a term of type $\Pi_{\Tm[A]}\Ty$ in context $P$, we will denote by $\widetilde B_\Gamma(x)$ the element of $\Ty(\Gamma.A_{\Gamma}(x))$ corresponding to $B_\Gamma(x)$ through the isomorphism \ref{eq:ty-ext}.  Expanding the definition of the isomorphism, one can show that:
$$
\widetilde B_\Gamma(x) = (\lambda^{-1}B)_{\Gamma.A_\Gamma(x)}(x[\pi], v_{A_\Gamma(x)}).
$$

\begin{cor}\label{cor:tm-ext}
Let $P$ be a presheaf on $\C$, $A$ a term of type $\Ty$, and $B$ a term of type $\Pi_{\Tm[A]} \Ty$, both in the context $P$.  The type
$$
\Pi_{\Tm[A]} \Tm[B\ a]
$$
is isomorphic to the presheaf on $\int P$ given by:
$$
(\Gamma, x) \mapsto \Tm_{\Gamma.A_\Gamma(x)}(\widetilde B_{\Gamma}(x))
$$
\end{cor}

The universe $\Ty$ allows us to use the CwF structure on $\presheaf \C$ to give definitions that work across all types of $\C$.  However, to generalise $\Pi$ and $\Sigma$, we need to access \emph{pairs} of dependent types.  For that reason, we define the context $\Ty^{(2)}$ as:
$$
(\b A : \Ty)(\b B : \Pi_{\Tm[\b A]} \Ty).
$$

Here we are using the syntactical conventions introduced in \cref{sec:notation}.
Let us take a minute to explain in detail what this expression means.

First of all, since $\Ty$ is a type in the unit context of $\presheaf{C}$, we can form a context $P_0 :\equiv (\b A : \Ty)$ by extension from the unit context, and use $\b A$ to refer to the corresponding term of type $\Ty$, i.e. $\b A : \presheaf{\Tm}_{P_0}(\Ty)$.

In the context $P_0$, the morphism corresponding to the variable $\b A$ is just the identity $P_0 \to P_0$, hence $\Tm[\b A]$ could have simply been written as $\Tm$.  However, using an explicit substitution makes it clear that we are referring to the variable $\b A$, and generalises better to situations where the context contains more than one variable.

Since $T :\equiv \Pi_{\Tm[\b A]} \Ty$ is a type in the context $P_0$, we can perform another context extension and obtain the context $P_0(\b B : T)$.  If we make weakenings explicit, now $\b A$ refers to the variable of type $\Ty[\pi_{\Ty}][\pi_T]$, and $\b B$ to the variable of type $T[\pi_T]$.

\begin{cor}\label{cor:ty2-ext}
There is an isomorphism, natural in $\Gamma : \C$:
$$
\Ty^{(2)}(\Gamma) \cong \coprod_{A : \Ty(\Gamma)}\Ty(\Gamma.A).
$$
\end{cor}
\begin{proof}
Immediate consequence of \cref{cor:ty-ext} and the definition of context extension of presheaves.
\end{proof}

Thanks to \cref{cor:ty2-ext}, we are free to identify elements of $\Ty^{(2)}(\Gamma)$ with pairs of types $(A, B)$, where $A : \Ty(\Gamma)$ and $B : \Ty(\Gamma.A)$.  However, using $\Ty^{(2)}$ can sometimes be preferable, since it avoids referring to context extension at all.

\begin{defn}\label{def:pi-structure}
A $\Pi$-type structure on $\C$ is given by:
\begin{itemize}
\item a term
\begin{equation}\label{eq:pi-structure-pi}
\pi : \presheaf{\Tm}_{\Ty^{(2)}}(\Ty),
\end{equation}
\item an isomorphism
\begin{equation}\label{eq:pi-structure-iso}
\Tm[\pi] \cong \Pi_{a : \Tm[\b A]} \Tm[\b B\ a]
\end{equation}
of types over $\Ty^{(2)}$.
\end{itemize}
\end{defn}

Note that a $\Pi$-type structure on $\C$ is given entirely in terms of the CwF structure on $\presheaf{\C}$ and its fibrant universe.

\Cref{def:pi-structure} can be stated more explictly: giving the term \ref{eq:pi-structure-pi} is the same as giving a natural transformation $\pi : \Ty^{(2)} \to \Ty$, and, thanks to \cref{cor:tm-ext}, the isomorphism \ref{eq:pi-structure-iso} is equivalent to an isomorphism: \begin{equation}\label{eq:pi-structure-iso2} \Tm_\Gamma(\pi_\Gamma(A, B)) \cong \Tm_{\Gamma.A}(B).  \end{equation}

It is then easy to verify that $\C$ supports $\Pi$-types (\cref{def:pi-structure-standard}) if and only if it has a $\Pi$-type structure.  In particular, we get the following:

\begin{prop}
For any category $\C$, the presheaf category $\presheaf\C$ is equipped with a canonical $\Pi$-type structure.
\end{prop}
\begin{proof}
It looks like one could simply take $\pi$ to be the $\Pi$ operation on presheaves.  However, $\Pi$, regarded as a family of functions $\presheaf\Ty^{(2)}(\Gamma) \to \presheaf\Ty(\Gamma)$, is not natural in $\Gamma$.

In fact, keeping in mind that $\presheaf\Ty(\Gamma)$ is a \emph{category}, and not just a set, one would only be able to prove that $\Pi$ is a \emph{pseudonatural} transformation of functors $\C \to \op{\cat{Cat}}$.  Fortunately, there is a way to give an alternative equivalent definition of $\Pi$ that is indeed strictly natural.

Let $P : \presheaf\C$, $A : \presheaf\Ty(P)$, and $B : \presheaf\Ty(P.A)$.  We will define $\pi(A,B)$ as a functor $\int P \to \op{\set}$.  For $(\Gamma,x) : \int P$, we will write $x : y(\Gamma) \to P$ for the morphism corresponding to $x$ through the isomorphism of the Yoneda lemma.  Then set:
$$
\pi(A,B)_\Gamma(x) :\equiv \left(\Pi_{A[x]}B[x^+]\right)_\Gamma(\id).
$$

Pseudonaturality of $\Pi$ implies that $\pi(A,B) \cong \Pi_A B$.  Furthermore, it is easy to check directly that $\pi : \Ty^{(2)} \to \Ty$ is (strictly!) a natural transformation.

The isomorphism \cref{eq:pi-structure-iso2} can now be obtained from $\lambda$ abstraction for $\Pi$, and the fact that $\Pi$ and $\pi$ are pointwise isomorphic.
\end{proof}

\begin{defn}\label{def:sigma-structure}
A $\Sigma$-type structure on $\C$ is given by:
\begin{itemize}
\item a term
\begin{equation}\label{eq:sigma-structure-sigma}
\sigma : \presheaf{\Tm}_{\Ty^{(2)}}(\Ty),
\end{equation}
\item an isomorphism
\begin{equation}\label{eq:sigma-structure-iso}
\Tm[\sigma] \cong \Sigma_{a : \Tm[\b A]} \Tm[\b B\ a]
\end{equation}
of types over $\Ty^{(2)}$.
\end{itemize}
\end{defn}

Like in the case of $\Pi$-type structures, $\Sigma$-type structures have a more direct characterisation: giving a $\Sigma$-type structure on $\C$ is the same as giving a natural transformation $\sigma : \Ty^{(2)} \to \Ty$, together with a natural isomorphism between $\Tm_\Gamma(\sigma_\Gamma(A, B))$ and the set of pairs $(a,b)$, where $a : \Tm_\Gamma(A)$ and $b : \Tm_\Gamma(B[a])$.  Clearly, this is just a reformulation of \cref{def:sigma-structure-standard}, hence $\C$ supports $\Sigma$-types if and only if it has a $\Sigma$-type structure.

From this characterisation, we get:

\begin{prop}
For any category $\C$, the presheaf category $\presheaf\C$ is equipped with a canonical $\Sigma$-type structure.
\end{prop}
\begin{proof}
The morphism $\sigma : \Ty^{(2)} \to \Ty$ can now be taken to be the $\Sigma$ operation on presheaves, which in this case is automatically natural. The required isomorphism follows directly from the definition of $\Sigma$.
\end{proof}

\begin{defn}\label{def:equality-structure}
An equality type structure on $\C$ is given by:
\begin{itemize}
\item a term
\begin{equation}\label{eq:equality-structure-eq}
\eq : \presheaf{\Tm}_{(A : \Ty).\Tm[A].\Tm[A]}(\Ty),
\end{equation}
\item an isomorphism
\begin{equation}\label{eq:equality-structure-iso}
\Tm[\eq] \cong \Eq_{\Tm[A]}
\end{equation}
of types over $(A:\Ty).\Tm[A].\Tm[A]$.
\end{itemize}
\end{defn}

By \cref{cor:ty-ext}, a term like $\eq$ in \cref{def:equality-structure} is given by a map that assigns, to every $A : \Ty(\Gamma)$ a type $\eq(A) : \Ty(\Gamma.A.A)$, naturally in $(\Gamma,A)$.

Isomorphism \cref{eq:equality-structure-iso} is equivalent to an isomorphism between sections of the morphism $\Gamma.A \to \Gamma$ (display map of $A$), and of the morphism $\Gamma.A.A.\eq(A) \to \Gamma$ (composition of display maps).

\begin{prop}\label{prop:presheaf-equality-structure}
For any category $\C$, the presheaf category $\presheaf\C$ is equipped with a canonical equality type structure.
\end{prop}
\begin{proof}
As for $\Pi$ and $\Sigma$, we want to define $\eq$ using the $\Eq$ operation on presheaves, but once again we have the problem that $\Eq$, as defined, is not strictly natural.  However, thanks to \cref{lem:eq-subterminal}, we can easily define a stricter version of $\Eq$.

For $P : \presheaf\C$, and $A : \presheaf\Ty(P)$, let $\eq(A)$ be the image of the unique map $\Eq(A) \to 1$ in $\presheaf\bTy(P.A.A)$.  Since $\Eq(A)$ is subterminal by \cref{lem:eq-subterminal}, it follows that $\eq(A) \cong \Eq(A)$, and $\eq$ is clearly natural in $A$.

The required isomorphism is now easy to construct.
\end{proof}

The construction in \cref{prop:presheaf-equality-structure} may appear more involved than necessary, since one might be tempted to simply define $\eq$ as:
\begin{equation}\label{eq:nonconstructive-eq}
\eq(A)_\Gamma(x, a, a') = \left\{
\begin{aligned}
& 1 \qquad \mathrm{if}\ a = a' \\
& 0 \qquad \mathrm{otherwise}.
\end{aligned}
\right.
\end{equation}

However, a definition like \cref{eq:nonconstructive-eq} presumes that we are able to decide the equality of arbitrary functions.  Classically, \cref{eq:nonconstructive-eq} is equivalent to the definition given in \cref{prop:presheaf-equality-structure}, but the way we phrased it makes it valid in a constructive setting as well.

Similarly to $\Pi$ and $\Sigma$-type structures, the existence of an equality structure is equivalent to the fact that $\C$ supports equality structures (\cref{def:equality-structure-standard}).

Finally, we will define one last structure. This one is fortunately much simpler than the previous three.

\begin{defn}\label{def:unit-structure}
A \emph{unit type structure} on $\C$ is given by:
\begin{itemize}
\item a term
\begin{equation}\label{eq:unit-structure-eq}
u : \presheaf\Tm_1(\Ty)
\end{equation}
\item an isomorphism
\begin{equation}\label{eq:unit-structure-iso}
\Tm[u] \cong 1
\end{equation}
of types in the unit context.
\end{itemize}
\end{defn}

And correspondingly:

\begin{prop}
For any category $\C$, the presheaf category $\presheaf\C$ is equipped with a canonical unit type structure.
\end{prop}
\begin{proof}
The type $u$ can be set to the unit presheaf $1$. The required isomorphism obviously follows from the fact that $1$ is terminal.
\end{proof}

Again, unit type structures and the existence of unit types (\cref{def:unit-structure-standard}) are equivalent.

\subsection{Morphisms}\label{sec:basic-type-former-morphisms}

Given a morphism $F : \C \to \D$ between CwFs, if $\C$ and $\D$ are equipped with one of the structures defined in \cref{sec:basic-type-formers}, we can ask whether $F$ \emph{preserves} those structures.

\begin{defn}\label{def:related-pairs}
Let $\Gamma : \C$, $(A,B) : \Ty^{(2)}_\Gamma$ and $(A', B') : \Ty^{(2)}_{F\Gamma} $.  We say that $(A,B)$ and $(A',B')$ are $F$-related if:
\begin{itemize}
\item $FA = A'$
\item for all $(\Delta,\sigma) : \C/\Gamma$, and all terms $a : \Tm_\Delta(A[\sigma])$, we have that $F(B(a)) = B'(Fa)$.
\end{itemize}
\end{defn}

The following is a direct consequence of \cref{def:related-pairs}:

\begin{lem}\label{lem:related-pairs}
Two pairs $(A,B)$ and $(A',B')$ as in \cref{def:related-pairs} are $F$-related if and only if:
\begin{itemize}
\item $FA = A'$
\item $\phi^F_A(F\widetilde B) = \widetilde B'$, where $\phi^F_A$ is as in \cref{def:cwf-morphism}, $\widetilde B$ is the type in $\Ty(\Gamma.A)$ corresponding to $B$ through the isomorphism of \cref{cor:ty-ext}, and $\widetilde B'$ is defined similarly.
\end{itemize}
\end{lem}

In particular, for all pairs $(A,B)$ in $\C$ there is exactly one pair $(A',B')$ in $\D$ that is related to it.  The advantage of formulating the following definitions in terms of related pairs rather than using the characterisation of \cref{lem:related-pairs} directly is that we need no mention of context extension.

\newcommand{\piAB}{\Pi_{a : \Tm[\b A]} \Tm[\b B\ a]}
\newcommand{\sigmaAB}{\Sigma_{a : \Tm[\b A]} \Tm[\b B\ a]}

\begin{defn}\label{def:related-functions}
Let $(A,B)$ and $(A',B')$ be $F$-related pairs, $u : (\piAB)_\Gamma(A, B)$ and $u' : (\piAB)_{F\Gamma}(A', B')$.  We say that $u$ and $u'$ are $F$-related if for all $(\Delta,\sigma) : \C/\Gamma$, and all terms $a : \Tm_\Delta(A[\sigma])$, we have that $F(u(a)) = u'(Fa)$.
\end{defn}

Note that the equality between $F(u(a))$ and $u'(Fa)$ in \cref{def:related-functions} makes sense because $(A,B)$ and $(A',B')$ are themselves related.

\begin{defn}\label{def:preservation-pi}
Suppose both $\C$ and $\D$ are equipped with $\Pi$-type structures. We say that $F$ \emph{preserves} $\Pi$-types if, for all related pairs $(A,B)$ and $(A',B')$:
\begin{itemize}
\item $F(\pi(A,B)) = \pi(A', B')$,
\item for all terms $f : \Tm_\Gamma(\pi(A, B))$, the element of $(\piAB)_\Gamma(A, B)$ corresponding to $f$ through the $\Pi$-type structure on $\C$ is related to the element of $(\piAB)_{F\Gamma}(A', B')$ corresponding to $Ff$ through the $\Pi$-type structure on $\D$.
\end{itemize}
\end{defn}

The definition of preservation of $\Sigma$-types is similar, but simpler, because we don't need to define a notion of relatedness for elements of $\sigmaAB$, as we can simply map them using $F$ directly:

\begin{defn}
Suppose both $\C$ and $\D$ are equipped with $\Sigma$-type structures. We say that $F$ \emph{preserves} $\Sigma$-types if, for all related pairs $(A,B)$ and $(A',B')$:
\begin{itemize}
\item $F(\sigma(A,B)) = \sigma(A', B')$,
\item the following diagram commutes:
$$
\xymatrix{
\Tm_\Gamma(\sigma(A,B)) \ar[r]^-\cong \ar[d]_F &
(\sigmaAB)_\Gamma(A,B) \ar[d]^F \\
\Tm_{F\Gamma}(\sigma(A', B')) \ar[r]^-\cong &
(\sigmaAB)_{F\Gamma}(A',B'),
}
$$
where the horizontal arrows are the isomorphisms given by the $\Sigma$-type structures on $\C$ and $\D$ respectively.
\end{itemize}
\end{defn}

For equality types, the definition is entirely analogous:

\begin{defn}
Suppose both $\C$ and $\D$ are equipped with equality type structures. We say that $F$ \emph{preserves equality} if, for all $\Gamma : \C$, $A : \Ty(\Gamma)$:
\begin{itemize}
\item $F(\eq(A, a, b)) = \eq(FA, Fa, Fb)$,
\item the following diagram commutes:
$$
\xymatrix{
\Tm_\Gamma(\eq_\Gamma(A, a, b)) \ar[r]^-\cong \ar[d]_F &
(\Eq_{\Tm[\b A]})_\Gamma(A, a, b) \ar[d]^F \\
\Tm_{F\Gamma}(\eq_{F\Gamma}(FA, Fa, Fb)) \ar[r]_-\cong &
(\Eq_{\Tm[\b A]})_{F\Gamma}(FA, Fa, Fb)
}
$$
\end{itemize}
\end{defn}

Finally, we say that $F$ preserves the unit type simply if $Fu = u$ over the unit context.

Replacing equality with isomorphism in the above definitions yields the notions of \emph{weak preservation} of the various type structures.

\begin{remark}\label{remark:application-through-morphism}
Let $F : \C \to \D$ be a CwF morphism. Suppose $\C$ is equipped with a $\Pi$ structure.  Then the application morphism $\epsilon_{A,B} : \Gamma.A.\Pi_A B \to \Gamma.A.B$ can be mapped to $\D$ through $F$, which implies that we can apply terms of type $F(\Pi_A B)$ to terms of type $F A$, even though $\D$ might not even have a $\Pi$-type structure.
\end{remark}

\subsection{The Yoneda embedding for CwFs}\label{sec:yoneda}

If $\C$ is a CwF, the Yoneda embedding $y : \C \to \presheaf\C$ is a functor between CwFs, so it is natural to ask whether it can be extended to a CwF morphism.

\begin{defn}\label{def:y0}
Let $\Gamma : \C$ be a context, and $A : \Ty(\Gamma)$ a type over $\Gamma$.  Define the presheaf $y_0(A) : \presheaf\Ty(y\Gamma)$ as follows:
$$
y_0(A)_\Delta(\sigma) :\equiv \Tm_\Delta(A[\sigma]).
$$
\end{defn}

\begin{prop}\label{prop:y0-extension}
For all $\Gamma : \C$ and $A : \Ty(\Gamma)$, there is a natural isomorphism
$$
y(\Gamma).y_0(A) \cong y(\Gamma.A)
$$
over $y(\Gamma)$.
\end{prop}
\begin{proof}
Immediate consequence of the defining isomorphism of context extension.
\end{proof}

\begin{lem}\label{lem:y0-terms}
For all $\Gamma : \C$ and $A : \Ty(\Gamma)$, we have:
$$
\Tm_\Gamma(A) \cong \presheaf\Tm_{y\Gamma}(y_0 A).
$$
\end{lem}
\begin{proof}
It follows from \cref{prop:y0-extension} and \cref{prop:tm-sections} that $\presheaf\Tm_{y\Gamma}(y_0 A)$ is naturally isomorphic to the set of sections of $y\pi_A : y(\Gamma.A) \to y(\Gamma)$.  By the Yoneda lemma, this is isomorphic to the set of sections of $\pi_A : \C(\Gamma.A, \Gamma)$, which, by \cref{prop:tm-sections} again, is isomorphic to $\Tm_\Gamma(A)$.
\end{proof}

\begin{prop}\label{prop:y0-embedding}\label{prop:yoneda}
For any CwF $\C$, the Yoneda Embedding $y : \C \to \presheaf\C$ can be extended to a CwF morphism, where $y_0$ is the action of the morphism on types, and the isomorphism of \cref{lem:y0-terms} is its action on terms.
\end{prop}
\begin{proof}
Naturality of $y_0$ is easy to verify.  It only remains to check that the map
$$
\phi^y_A : \presheaf\C(y(\Gamma.A), y(\Gamma).y_0(A))
$$
as in \cref{def:cwf-morphism} is an isomorphism, but this follows immediately from the fact that it is the inverse of the isomorphism of \cref{prop:y0-extension}.
\end{proof}

The reason for the subscript $0$ in our notation for the action of $y$ on types is that, when $\C$ possesses $\Pi$ and $\Sigma$ type structures, the map $y_0$, as defined, does not preserve them.

We will later define in certain cases a stricter version of $y_0$ that does indeed preserve the extra structure, and we reserve the name $y$ for that.

\subsection{Presheaf universes}\label{sec:presheaf-universes}

Using a universe of sets $\set_i$, we can build a universe in any presheaf model. This construction follows closely the one in \cite{universe-lifting}. Let $\C$ be any small category, and consider the CwF structure on $\presheaf\C$ defined in \cref{sec:presheaves}.

\begin{defn}
Let $P$ be a context in $\presheaf\C$. A type $A : \presheaf\Ty(P)$ is said to be \emph{small} (with respect to $\set_i$), if it factors through $\set_i$ when regarded as a functor $\int P \to \set$.
\end{defn}

For all object $\Gamma : \C$, let $\U_\Gamma$ be the set of small types over $y\Gamma$.  This defines a presheaf $\U$ on $\C$.

For all $\Gamma : \C$ and $P : \U_\Gamma$, define
$$
\El_\Gamma(P) :\equiv P_\Gamma(\id).
$$

We now have a universe $(\U, \El)$ in $\presheaf\C$.

\begin{prop}
The universe $(\U, \El)$ classifies small types, i.e. a type $A$ over $P$ is small if and only if there exists a term $\widetilde A$ of type $\U$ over $P$ such that $A = \El[\widetilde A]$.
\end{prop}
\begin{proof}
Clearly, $\El$ is small, hence $\El[\widetilde A]$ is small for all $\widetilde
A: P \to \U$.

Conversely, if $A$ is small, define $\widetilde A : P \to \U$ as follows:
$$
\widetilde A_\Gamma(x) :\equiv A[x],
$$
where $x : y(\Gamma) \to P$ denotes the morphism corresponding to $x : P_\Gamma$
through the isomorphism of the Yoneda lemma.  We have:
\begin{align*}
\El[\widetilde A]_\Gamma(x)
& = \El_\Gamma(\widetilde A_\Gamma(x)) \\
& = \El_\Gamma(A[x]) \\
& = A[x]_\Gamma(\id) \\
& = A_\Gamma(x).
\end{align*}
\end{proof}

\subsection{More notational conventions}
\label{sec:notational-conventions}

In the following, we will make heavy use of nested $\Pi$ and $\Sigma$ types, building complicated type expressions with them.  It is therefore convenient to adopt a ``flatter'' notation, one that is more symmetric with the respect to the two arguments of a $\Pi$ or $\Sigma$ type.

This notation is inspired by the syntax of the proof assistant \textsc{agda}~\cite{norell:agda}, and it works as follows: a type like $\Pi_{a : A} B$ is written as:
$$
(a : A) \to B,
$$
mimicking the usual notation for (non-dependent) function types.

Similarly, the type $\Sigma_{a : A} B$ will be written as follows:
$$
(a : A) \times B,
$$
making it explicit that $\Sigma$-types can be thought of as a generalised form of products.

Chained $\Pi$ types will be written by omitting all the intermediate arrows, and if the same type is present more than once, the corresponding variables can be grouped within one bracket. For example:
$$
(a : A)(b, b' : B)(c : C) \to D
$$
represents the type:
$$
\Pi_{a : A}\Pi_{b : B}\Pi_{b' : B}\Pi_{c : C} D.
$$

Finally, if $(\U, \El)$ is a universe, we will sometimes omit uses of $\El$, as they can be inferred very easily: if a term is used in place of a type, it means that there is an implicit application of $\El$ there.

\subsection{Fibrations and contextuality}\label{sec:fibrations}

\begin{defn}
Let $p : \C(\Delta, \Gamma)$ be a morphism in a CwF. We say that $p$ is a \emph{fibration} if there is a type $A : \Ty(\Gamma)$ such that $p$ and $p_A : \C(\Gamma.A, \Gamma)$ are isomorphic in the slice category $\C/\Gamma$.

We say that a context $\Gamma$ is \emph{fibrant} if the unique morphism $\C(\Gamma,1)$ is a fibration.
\end{defn}

\begin{lem}\label{lem:pullbacks-of-fibrations}
In any CwF, pullbacks of fibrations exist and are fibrations.
\end{lem}
\begin{proof}
Immediate consequence of \cref{prop:morphism-weakening}.
\end{proof}

\begin{defn}\label{def:contextual}
A CwF $\C$ is said to be \emph{contextual} if every context of $\C$ is fibrant.
\end{defn}

The idea of \cref{def:contextual} is to express the idea that in certain CwFs contexts are none other than types in the unit context.  For example, this holds for syntactical models like $\RF_0$, introduced in \cref{sec:rule-framework} (see \cref{lem:rf0-contextual}).

If $\C$ is a CwF, and $\Gamma$ is any context of $\C$, we can put a category structure on $\Ty(\Gamma)$ by defining a morphism between types $A$ and $B$ to be a morphism between $p_A$ and $p_B$ in the slice category $\C/\Gamma$.  We denote with $\bTy(\Gamma)$ the resulting category of types over $\Gamma$.

Note that the notation $\bTy(\Gamma)$ is consistent with how we denoted the category of types over a presheaf in \cref{sec:presheaves}.

\begin{prop}\label{prop:contextual-equivalence}
A CwF $\C$ is contextual if and only if the canonical functor $j : \bTy(1) \to \C$ is an equivalence of categories.
\end{prop}
\begin{proof}
The functor $j$ is always fully faithful, and $\C$ being contextual is clearly equivalent to $j$ being essentially surjective.
\end{proof}

\begin{cor}
A presheaf category is a contextual CwF.
\end{cor}

Contextual CwFs are similar to C-systems (also called contextual categories) \cite{cartmell-contextual}.  There are, however, two important differences:
\begin{itemize}
\item the identification between types and contexts is not canonical, and only up to isomorphism;
\item we require that every context can be obtained out of a single type, rather than a chain of types.
\end{itemize}

In particular, the second condition implies that our notion of contextuality is only well-behaved when $\C$ has a $\Sigma$-type structure.  It would be possible to formulate \cref{def:contextual} in a way that doesn't implicitly require the existence of $\Sigma$-types, using the idea of a \emph{telescope} (i.e. a finite sequence of types, each depending on the previous ones), but doing so is cumbersome, and will not be required in the following, so we avoid it.

\begin{prop}\label{prop:ty-cwf}
If $\C$ is a CwF equipped with a $\Sigma$-type structure, then the category $\bTy(\Gamma)$ is itself a CwF with a $\Sigma$-type structure, and the canonical functor $j : \bTy(\Gamma) \to \C$ is a split CwF morphism preserving $\Sigma$-types.
\end{prop}
\begin{proof}
Define a type over $A : \bTy(\Gamma)$ to simply be an element of $\Ty(\Gamma.A)$.
Context extension and $\Sigma$-types can be defined directly using the $\Sigma$-type structure of $\C$.

Verifying that $j$ is a split CwF morphism is then straightforward, and the preservation of $\Sigma$-types is a direct consequence of the definitions.
\end{proof}

Contextuality has a useful category-theoretic consequence:

\begin{prop}\label{prop:contextual-finite-products}
Let $\C$ be a contextual CwF.  Then $\C$ has finite products.
\end{prop}
\begin{proof}
The existence of a terminal object is part of the definition of a CwF, so we only need to show that $\C$ has binary products.

Let $\Gamma, \Delta : \C$ be any two contexts.  By contextuality, we can replace $\Delta$ with a type $A$ over the unit context.  By \cref{prop:morphism-weakening}, the following square is a pullback:
$$
\xymatrix{
\Gamma.A \ar[r]\ar[d] &
A \ar[d] \\
\Gamma \ar[r] &
1,
}
$$
which means that $\Gamma.A$ is the product of $\Gamma$ and $A$.
\end{proof}

We conclude this \namecref{sec:fibrations} with a construction that will occasionally be useful later.

\begin{prop}\label{prop:cwf-slice}
Let $\C$ be a CwF, and $\Gamma : \C$ a context. The slice category $\C / \Gamma$ can be equipped with a CwF structure.
\end{prop}
\begin{proof}
If $(\Delta,\sigma)$ is an object of $\C$, we simply define types and terms over $(\Delta,\sigma)$ to be the types and terms over $\Delta$ in $\C$.
\end{proof}

\section{The Rule Framework}\label{sec:rule-framework}

We will use the type structures defined above to ``bootstrap'' a more general definition of structure for CwF.  To that end, we give the following definition:

\begin{defn}
An $\RF$-category\footnote{$\RF$ stands for \emph{rule framework}} is a CwF $\C$, equipped with $\Pi$, $\Sigma$, equality and unit type structures, and a universe $\U$, $\El$.  An $\RF$-morphism is a CwF morphism preserving all the structure.
\end{defn}

$\RF$-categories and $\RF$-morphisms form a category $\mathcal{RF}$. Denote by $\strict{\mathcal{RF}}$ the subcategory of $\mathcal{RF}$ consisting of only split morphisms. We will need the following:

\begin{lem}\label{lem:rf-limits}
The category $\strict{\mathcal{RF}}$ has all small limits.
\end{lem}
\begin{proof}
Let $I$ be a small category, and $F : I \to \mathcal{RF}$ a functor.  Denote by $\C_i$ the underlying category of $F(i)$.

We construct the limit of $F$ by first taking the limit $\C$ (in $\cat{Cat}$) of the $\C_i$, and then defining a CwF structure on $\C$, equipped with all the required type structures.

For a context $\Gamma : \C$, denote by $\Gamma_i$ the context of $\C_i$ obtained from $\Gamma$ through the projection of the universal cone $\C \to \C_i$.  Types over $\Gamma$ are defined to be simply the limit of $\Ty(\Gamma_i)$ over $I$.

Similarly, if $A$ is a type over $\Gamma$, we write $A_i$ for the projection of $A$ to $\Ty(\Gamma_i)$, and define terms of type $A$ as the limit of $\Tm_{\Gamma_i}(A_i)$.

Context extension is defined pointwise.  This is the crucial point where we use the fact that the diagram $F$ is composed solely of split morphisms.

Verifying that this gives a CwF structure on $\C$ is straightforward.

As for the $\Pi$, $\Sigma$, equality and unit type structures, they can all be defined pointwise, and the resulting RF-category is easily seen to satisfy the universal property of the limit.  \end{proof}

\begin{thm}\label{thm:rf0}
The category $\strict{\mathcal{RF}}$ has an initial object $\RF_0$.
\end{thm}

\Cref{thm:rf0} can be proved by giving an explicit inductive definition of $\RF_0$: types are expressions generated from base types like $\U$, $\El$ and the unit type, by applying the operations of the $\RF$ structures: $\Pi$, $\Sigma$ and equality.  Similarly, terms are generated from variables and their weakening by applying the various isomorphisms of the $\RF$ structures. Contexts are defined as tuples of types, and morphisms as tuples of terms.

Making this sort of definition precise is, however, far from a straightforward task, as is proving that it in fact gives an initial object of $\mathcal{RF}$.  Intuitively, initiality follows because we can regard every context (resp. type, term, morphism) in $\RF_0$ as a ``recipe'' to build a context (resp. type, term, morphism) in an arbitrary $\RF$-category $\C$.  This gives, for any such $\C$, a uniquely determined functor $\RF_0 \to \C$ that clearly preserves all the structures.

We follow a slightly more indirect approach, based on the ideas underlying the proof of the adjoint functor theorem. Indeed, the following proof could be adapted to show the more general fact that the forgetful functor $\mathcal{RF} \to \cat{Cat}$ has a left adjoint.  However, we will not need the extra generality.

\begin{proof}[Proof of \cref{thm:rf0}]
Since $\strict{\mathcal{RF}}$ has all small limits (\cref{lem:rf-limits}), it is enough to show that it has a weakly-initial family.  We say that a small $\RF$-category is \emph{countable} if the set of objects is countable, all the homsets are countable, and $\Ty(\Gamma)$ and $\Tm_\Gamma(A)$ are countable for all $\Gamma$ and $A$.

We will show that every $\RF$-category contains a countable $\RF$-subcategory. From this fact, the existence of a weakly-initial family easily follows (for example, fix a countably infinite set $\Omega$ and take the family of all $\RF$-categories whose contexts, morphisms, types and terms are all elements of $\Omega$).  

Let $\C$ be an $\RF$-category. We define a chain of subsets $\D_n$ of $\C$, each equipped with subfamilies of morphisms, types and terms, arranged just like in a CwF, but with no further structure.  The morphisms of $\D_n$ between contexts $\Delta$ and $\Gamma$ will be denoted $\D_n(\Delta, \Gamma)$, just like in a category, and they will form a subset of $\C(\Delta, \Gamma)$.  We will write $\Ty^n(\Gamma)$ for the types of $\D_n$ over $\Gamma$, which will form a subset of $\Ty(\Gamma)$, and similarly for terms.

The starting point $\D_0$ is just the empty subset. Given $\D_n$ and its associated structures, define $\D_{n+1}$ as the subset of $\C$ containing $\D_n$, plus all the contexts, morphisms, types and terms that are obtained from those of $\D_n$ by applying any of the operations of the $\RF$-category $\C$.  In detail:

\begin{itemize}
\item the set $\D_{n+1}$ contains all the elements of $\D_n$, plus the unit context, and the context $\Gamma.A$, for all choices of $\Gamma : \D_n$ and $A : \Ty^n(\Gamma)$;
\item morphisms of $\D_{n+1}$ are obtained from those of $\D_n$ by adding the canonical morphism to the unit context, identity morphisms, compositions of morphisms in $\D_n$, projections of types in $\D_n$ and substitutions of the form $\langle \sigma, a \rangle$, where $\sigma : \D_n(\Delta, \Gamma)$, and $a : \Tm^n_\Delta(A[\sigma])$;
\item the set $\Ty^{n+1}(\Gamma)$ contains all the types of $\D_n$, plus the unit type, types of the form $\Pi_A B$ and $\Sigma_A B$, where $A : \Ty^n(\Gamma)$, and types of the form $a = b$, where $a,b : \Tm^n_\Gamma(A)$, and $A : \Ty^n(\Gamma)$;
\item the set $\Tm^{n+1}(\Gamma)$ contains all the terms of $\D_n$, plus the unique inhabitant of the unit type, and the images of the isomoprhisms defining $\Pi$, $\Sigma$ and equality types and their inverses.
\end{itemize}

From the fact that every operation in the definition of $\RF$-category has a finite number of arguments, it easily follows that the union of all the $\D_n$ and corresponding structures forms an $\RF$-subcategory of $\C$.
\end{proof}

The advantage of the proof above over the usual technique of building the initial model purely syntactically is that the iterative construction happens within an existing CwF, hence we only need to concern ourselves with adding the necessary elements to the structures involved, and their required properties will automatically hold, because they do so in the ambient category.

We will write $\RF_0$ to denote the initial object of $\strict{\mathcal{RF}}$.  Since $\RF_0$ is only initial in a subcategory of $\mathcal{RF}$, we cannot conclude that it is initial in $\mathcal{RF}$. In particular, given an $\RF$-category $\C$, we can always give a morphism $\RF_0 \to \C$, but that morphism might not be unique.

Fortunately, we can prove a weaker version of uniqueness.

\begin{defn}
A \emph{weak $\RF$-morphism} is a CwF morphism that weakly preserves all the
structure.
\end{defn}

\begin{thm}\label{thm:rf-2-initiality}
Let $\C$ be an $\RF$-category, and $F, G : \RF_0 \to \C$ two weak $\RF$-morphisms in $\mathcal{RF}$.  Then $F$ and $G$ are isomorphic.
\end{thm}
\begin{proof}
Construct an $\RF$-category $\E$ (the \emph{pseudo-equaliser} of $F$ and $G$) as follows: the objects of $\E$ are contexts $\Gamma$ in $\RF_0$, together with an isomorphism between $F\Gamma$ and $G\Gamma$. Similarly, types (resp. terms) in $\E$ are types (resp. terms) in $\RF_0$, together with an isomorphism between their respective images in $\C$.

The fact that $F$ and $G$ are weak $\RF$-morphisms implies that it is possible to equip $\E$ with a structure of $\RF$-category such that the obvious projection $\pi: \E \to \RF_0$ is a split morphism.

By initiality of $\E$, the morphism $\pi$ has a section, which implies that $F$ and $G$ are isomorphic.
\end{proof}

\begin{lem}\label{lem:rf0-contextual}
The category $\RF_0$ is contextual.
\end{lem}
\begin{proof}
It is easy to see that $\bTy(1)$ can be equipped with an $\RF$-category structure such that the canonical functor $j : \bTy(1) \to \RF_0$ is a split $\RF$-morphism (see \cref{prop:ty-cwf}).  It follows that $j$ is an isomorphism of $\RF$-categories, hence $\RF_0$ is contextual by \cref{prop:contextual-equivalence}.
\end{proof}

\section{Type formers and structures}\label{sec:type-formers}

We know from \cref{sec:cwf} that presheaf categories are equipped with a canonical CwF structure, as well as $\Pi$, $\Sigma$, equality and unit type structures.  If $\C$ is a CwF, then its presheaf category additionally possesses a canonical universe (the \emph{fibrant} universe) given by the presheaves of types and terms.  Therefore, we have that for any CwF $\C$, the presheaf category $\presheaf{\C}$ is an $\RF$-category.

\begin{defn}\label{def:type-former}
A \emph{type former} is a context in $\RF_0$.
\end{defn}

The idea behind \cref{def:type-former} is that we can use the language of $\RF_0$ as a meta-theoretical framework to describe structures on a generic CwF $\C$.  The universe $\U$ in $\RF_0$ intuitively stands for the collection of types of $\C$.  Given some $A : \U$, the $\RF$-type $\El[A]$ corresponds to the terms of $A$ regarded as a type on $\C$.

Making this intuition precise is relatively straightforward: denote by $\interp{-}^{\presheaf\C}$ the unique split morphism $\RF_0 \to \presheaf{\C}$. Using $\interp{-}^{\presheaf\C}$, any type former can be interpreted as a presheaf on $\C$ constructed from $\Ty$ and $\Tm$, using the operations of $\RF$-categories in $\presheaf{\C}$.

\begin{defn}
Let $\Phi$ be a type former, and $\C$ a CwF.  A $\Phi$-structure on on $\C$ is a global element of $\interp{\Phi}^{\presheaf\C}$.  A CwF equipped with a $\Phi$-structure will be referred to as a $\Phi$-CwF.
\end{defn}

\begin{lem}\label{lem:phi-slice}
Let $\Phi$ be a type former. A $\Phi$-structure $\phi$ on $\C$ can be transported to a $\Phi$-structure on the slice category $\C/\Gamma$ for any context $\Gamma$.
\end{lem}
\begin{proof}
The $\Phi$-structure $\phi$ can be regarded as a term of type $\interp{\Phi}$ in the unit context of $\presheaf\C$.  If $! : y(\Gamma) \to 1$ is the unique morphism to the terminal object of $\presheaf\C$, it is not hard to verify that $\interp{\Phi}[!]$ coincides with $\interp{\Phi}^{\presheaf{\C / \Gamma}}$ under the isomorphism of \cref{prop:presheaf-families}.  Therefore, $\phi[!]$ is a $\Phi$-structure for $\C / \Gamma$.
\end{proof}

It follows from \cref{lem:phi-slice} that a slice of an $\RF$-category is itself an $\RF$-category.

\section{Examples}\label{sec:rf-examples}

All of the commonly employed type structures on CwFs can be expressed using the notion of type former developed in \cref{sec:type-formers}.

In particular, we can now revisit the definitions of the type structures of an $\RF$-category, as given in \cref{sec:cwf}, and reformulate them in terms of type formers.

For example, a $\Pi$-type structure is none other than a $\Phi^\Pi$-structure, where $\Pi$ is the following type former:
\begin{equation} \label{eq:pi-types-rf}
\begin{aligned}
\Phi^\Pi & = (A : \U)(B : A \to \U) \\
  & \to (P : \U) \times (P \cong ((a : A) \to B\ a)),
\end{aligned}
\end{equation}

where we are making use of the notation described in \cref{sec:notational-conventions} to represent nested $\Pi$ and $\Sigma$ types in $\RF$, and uses of $\El$ are implicit.  The symbol $\cong$ refers to a notion of \emph{isomorphism} internal to $\RF$, defined in the natural way:
\begin{align*}
A \cong B
& :\equiv (f : A \to B) \\
& \times (g : B \to A) \\
& \times ((x : A) \to g(f(x)) = x) \\
& \times ((y : B) \to f(g(y)) = y).
\end{align*}

Note that the equality symbol used here and in following type formers refers to the equality type structure that is part of the definition of $\RF$-category.

Expanding the definition of isomorphism into \cref{eq:pi-types-rf} brings it closer to the traditional formulation of $\Pi$-types: the return $\Sigma$-type in \cref{eq:pi-types-rf} consists of five components, corresponding to the formation, elimination and introduction rule, plus $\beta$ and $\eta$ equalities \cite{hofmann:syntax-and-semantics}.

Similarly, we can define a type former $\Phi^\Sigma$ for $\Sigma$-type structures, a type former $\Phi^{\textsc{eq}}$ for equality type structures, and a type former $\Phi^{\textsc{unit}}$ for unit type structures.

Unfortunately, we cannot use the above characterisations as definitions, because we need to bootstrap the process with a number of basic type structures in order to define $\RF_0$.

However, we can now use $\RF$ to give succint definitions of other commonly employed type structures, and, more importantly, we can prove metatheoretical results on CwFs while remaining agnostic of the particular type structures that they carry.

One of the simplest examples that we haven't covered directly so far is given by \emph{binary sums}.  They can be defined by the following $\RF$ type former:
\begin{align*}
\Phi^{\textsc{sum}}
& :\equiv (A, B: \U) \\
& \to (S : \U) \\
& \times (l : A \to S) \\
& \times (r : B \to S) \\
&
\begin{aligned}
\times (
& (P : S \to \U) \\
& (d_1 : (x : A) \to P(l(x))) \\
& (d_2 : (y : B) \to P(r(y))) \\
& (f : (s : S) \to P(s)) \\
& \times ((x : A) \to f(l(x)) = d_1(x)) \\
& \times ((y : B) \to f(r(y)) = d_2(y))).
\end{aligned}
\end{align*}

Another important example is \emph{intensional equality}, the cornerstone of Martin-L\"of type theory, and $\hott$ in particular.  This is not to be confused with the equality type former introduced in
\cref{sec:basic-type-formers}.
\begin{equation}\label{eq:ieq-rf}
\begin{aligned}
\Phi^{\textsc{IEQ}}
& :\equiv (A : \U) \\
& \to (E: A \to A \to \U) \\
& \times (r : (a : A) \to E(a, a)) \\
&
\begin{aligned}
\times (
& ((a : A)(P : (b : A) \to E(a, b) \to \U) \\
& (d : P(a, r(a))) \\
& \to (J : (b : A)(p : E(a, b)) \to P(b, p)) \\
& \times J(a, r(a)) = d).
\end{aligned}
\end{aligned}
\end{equation}

For comparison, the \emph{extensional} equality type structure of \cref{sec:basic-type-formers} can be represented in $\RF$ as follows:
\begin{align*}
\Phi^{\textsc{EQ}}
& :\equiv (A : \U) \\
& \to (E: A \to A \to \U) \\
& \times ((a, b : A) \to E(a, b) \cong (a = b)).
\end{align*}

Given a $\Phi^{\textsc{EQ}}$-structure, the ability to convert any \emph{propositional} equality, (i.e. a term of type $E(a, b)$), into a \emph{definitional} equality (i.e. an equality of $a$ and $b$ as terms), is often referred to as the ``reflection rule''.

The difference between intensional and extensional equality can then be summarised by the statement that intensional equality does not admit a reflection rule, and instead replaces it with the \emph{$J$ eliminator} and corresponding computation rule given in \cref{eq:ieq-rf}.

We employed extensional equality as a very convenient technical device in the development of our framework of type formers.  Indeed, many ``natural'' models of type theory like $\set$ or any presheaf model come equipped with a straightforward extensional equality structure.

However, in a constructive setting, extensional equality has certain undesirable characteristics (for example, models with extensional equality, such as $\RF_0$, tend to have undecidable equality of terms), hence intensional equality is often preferred.

As a compromise between the two forms of equality, we recall the following rule, depending on some $E : \Phi^{\textsc{IEQ}}$, called \emph{uniqueness of
identity proofs} (UIP).
\begin{equation} \label{eq:uip}
\begin{aligned}
\textsc{UIP}(E)
& :\equiv (A : \U) \\
& \to (a, b : A) \\
& \to (p, q : E(a, b)) \\
& \to E(p,q).
\end{aligned}
\end{equation}

$\textsc{UIP}$ says that any two parallel equalities are themselves equal, which means that types do not possess any higher equality structure.  In HoTT terminology, this can be expressed by saying that \emph{every type is a set}.

We will refer to the type former $(E : \Phi^{\textsc{IEQ}}) \times \Phi^{\textsc{UIP}}(E)$ as \emph{strict equality}.  Note that extensional equality satisfies UIP, hence it can be regarded as a special case of strict equality.

Other type formers that we will need in the following are:
\begin{itemize}
\item $\Phi^{\textsc{empty}}$, for the empty type;
\item $\Phi^\nat$, for the natural numbers;
\item $\Phi^{\textsc{funext}}$, for \emph{function extensionality}.
\end{itemize}

Their definitions can be obtained by encoding in $\RF_0$ the usual rules that concern them.  See for example \cite{hott-book} for a detailed exposition of these type formers and similar ones.

\section{Morphisms}

Similarly to what we did in \cref{sec:basic-type-former-morphisms}, we want to define what it means for a morphism between CwFs $F : \C \to \D$ to \emph{preserve} a $\Phi$-structure.
Unfortunately, due to the presence of $\Pi$-types in the description of a type former as a context in $\RF_0$, this turns out to be quite challenging.

In fact, given a CwF morphism $F : \C \to \D$, it is not possible in general to define a corresponding $\RF$-morphism $\widetilde F$ between $\presheaf\C$ and $\presheaf\D$, in either direction.
If we had such a morphism, we could say that $F$ \emph{preserves} $\Phi$-structures when $\widetilde F$ maps the $\Phi$-structure on $\C$ into the one on $\D$, or vice versa.

However, since this is not the case, our definition of preservation of type structures is much more cumbersome, and requires setting up some infrastructure to be able to talk about a form of ``logical relations'' on type structures.  Then, given an $F$, we will be able to recursively define the preservation relation on $\Phi$-structures on $\C$ and $\D$, essentially by induction on $\Phi$.

\begin{defn}\label{def:lax-rf-morphism}
An \emph{oplax $\RF$-morphism} between $\RF$-categories $\A$ and $\A'$, with universes $(\U, \El)$ and $(\U', \El')$ respectively, is given by:
\begin{itemize}
\item a CwF morphism $F : \A \to \A'$;
\item a morphism $F^{\U} : \A'(\U', F \U)$;
\item a morphism $F^{\El} : \A'(\U'.\El', F(\U.\El))$;
\end{itemize}
such that the following diagram commutes:
$$
\xymatrix{
\U'.\El' \ar[r]^{F^{\El}} \ar[d] &
F(\U.\El) \ar[d] \\
\U' \ar[r]_{F^{\U}} &
F\U.
}
$$
\end{defn}

We will often suppress the superscript $\U$ and $\El$ from our notation when working with an oplax $\RF$-morphism.

Note that $F$ is \emph{not} required to preserve any of the $\RF$-structure.

Given an oplax $\RF$-morphism $F : \A \to \A'$, we can construct an $\RF$-category
$\mathcal{R}_F$.

Objects of $\mathcal{R}_F$ are defined to be triples $\mathbf{\Gamma} = (\Gamma, \Gamma', R)$, where $\Gamma : \A$, $\Gamma' : \A'$, and $R$ is a \emph{span} over $F\Gamma$ and $\Gamma'$, i.e. a diagram in $\A'$ of the form:
\begin{equation}\label{eq:span}
\xymatrix{
F \Gamma & R \ar[l]_-l \ar[r]^-r & \Gamma'.
}
\end{equation}

A type over $\mathbf{\Gamma}$ is itself a triple $\mathbf{A} = (A, A', X)$, where $A : \Ty(\Gamma)$, $A' : \Ty(\Gamma')$, and $X : \Ty(R.FA[l].A'[r])$.
Context extension of $(A',A',X)$ is defined to be the span determined by $R.FA[l].A'[r].X$.

Terms of type $\mathbf{A}$ are defined to be triples $(a, a', x)$, where $a : \Tm_\Gamma(A)$, $a' : \Tm_{\Gamma'}(A')$, and $x : \Tm_R(X[al, a'r])$.

This determines a CwF structure on $\mathcal{R}_F$, and it is easy to see that the two obvious projections $\mathcal{R}_F \to \A$ and $\mathcal{R}_F \to \A'$ are split CwF morphisms.

\begin{prop}\label{prop:oplax-rf}
The CwF $\mathcal{R}_F$ defined above has an $\RF$-structure, and the two projections $\mathcal{R}_F \to \A$ and $\mathcal{R}_F \to \A'$ are split $\RF$-morphisms.
\end{prop}
\begin{proof}
We will only show how to define a $\Pi$-type structure on $\mathcal{R}_F$, since this is the most involved step.

Let $\mathbf{\Gamma} = (\Gamma, \Gamma', R)$ be a context in $\mathcal{R}_F$, $\mathbf{A} = (A,A',X)$ a type over it, and $\mathbf{B} = (B,B',Y)$ a type over $\mathbf{\Gamma}.\mathbf{A}$.  Let $R$ be given by the span in \cref{eq:span}.

The $\Pi$-type $\Pi_{\mathbf{A}} \mathbf{B}$ is defined as the triple $\mathbf{P} = (\Pi_A B, \Pi_{A'} B', P)$, where $P$ is the following type in the context $R_0 = R(u : F(\Pi_A B)[l])(u' : \Pi_{A'}{B'}[r])$:
$$
\Pi_{a : F A[l]} \Pi_{a' : A'[r]}
\Pi_{X[a,a']}
Y[a, u\ a, a', u'\ a'],
$$

and $u\ a$ denotes the application of $u : F(\Pi_A B)[l]$ to $a : F A[l]$, as described in \cref{remark:application-through-morphism}.

Terms of type $\mathbf{P}$ are triples $(u, u', w)$, where $u : \Tm_\Gamma(\Pi_A B)$, $u' : \Tm_\Gamma(\Pi_{A'} B')$, and $w : \Tm_R(P[Fu[l], u'r])$.

Using the defining properties of $\Pi$-type structures in $\A$ and $\A'$, we can see that these are naturally isomorphism to triples $(b, b', y)$, where $b : \Tm_{\Gamma.A} B$, $b' : \Tm_{\Gamma'.A'} B'$, and $y : \Tm_{R(a : F A[l])(a' : A'[r]).X}(Y[a, Fb[l], a', b'[r]])$, which are exactly terms of type $\mathbf{B}$ in the context $\mathbf{\Gamma}.\mathbf{A}$.
\end{proof}

Since the functors $\mathcal{R}_F \to \A$ and $\mathcal{R}_F \to \A'$ are split $\RF$-morphisms by \cref{prop:oplax-rf}, initiality of $\RF_0$ implies that $\interp{\Phi}^{\mathcal{R}_F}$ is a span over $\interp{\Phi}^{\A}$ and $\interp{\Phi}^{\A'}$.  We will write that span as:
$$
\xymatrix{
F \interp{\Phi}^{\A} & \interp{\Phi}^F \ar[l] \ar[r] & \interp{\Phi}^{\A'}.
}
$$

\begin{defn}
Let $\phi$ be a global element of $\interp{\Phi}^{\A}$ and $\phi'$ a global element of $\interp{\Phi}^{\A'}$.  An \emph{element} of $\interp{\Phi}^F$ over $\phi$ and $\phi'$ is defined to be a global element $s$ of $\interp{\Phi}^F$ such that the following diagram commutes:
$$
\xymatrix{
& 1 \ar[dl]_\phi \ar[d]^s \ar[dr]^{\phi'} \\
F \interp{\Phi}^{\A} & \interp{\Phi}^F \ar[l] \ar[r] & \interp{\Phi}^{\A'}.
}
$$
\end{defn}

Let us now fix two CwFs $\C$ and $\D$, both equipped with $\Phi$-structures for some type former $\Phi$, and a CwF morphism $F : \C \to \D$.  The following \namecref{lem:presheaf-functor-lax} is an immediate consequence of \cref{def:lax-rf-morphism}.

\begin{lem}\label{lem:presheaf-functor-lax}
The functor $F^* : \presheaf \D \to \presheaf\C$ is an oplax $\RF$-morphism.
\end{lem}

It follows from \cref{lem:presheaf-functor-lax} that we have an $\RF$-category $\mathcal{R}_{F^*}$, thus we get an interpretation morphism $\RF_0 \to \mathcal{R}_{F^*}$.

\begin{defn}\label{def:phi-morphism}
Let $\phi$ and $\psi$ be the $\Phi$-structures of $\C$ and $\D$ respectively.  We say that $F$ is a \emph{$\Phi$-morphism} (or that $F$ preserves $\Phi$-structures) if there exists an element of $\interp{\Phi}^{F^*}$ over $\phi$ and $\psi$.
\end{defn}

\Cref{def:phi-morphism} is based on the idea of \emph{logical relations} \cite{tait:logical-relations}.  For a fixed CwF morphism $F$, we defined a notion of ``being related through $F$'' for $\Phi$-structures, by induction on $\Phi$.

For the type formers of $\RF$ itself, it is not hard to see that preservation as defined in \cref{sec:basic-type-former-morphisms} coincides with the notion of \cref{def:phi-morphism}, when using the equivalent definitions given in \cref{sec:rf-examples}.

Note that, for a general $\Phi$, for $\Phi$-structures $\phi$ and $\psi$ on $\C$ and $\D$ respectively, being related through $F$ does not mean that $\phi$ can be mapped through $F$ to a $\Phi$-structure on $\D$ that happens to coincide with $\psi$.  In fact, there is no way in general to transport a $\Phi$-structure along an arbitrary functor.

This can be understood in analogy with common algebraic structures. For example, given two monoids $A$ and $B$, and a function between them $f : A \to B$, we know what it means for $f$ to be a monoid homomorphism - meaning that the two monoid structures on $A$ and $B$ are ``related through $f$'' - but there is in general no way to transport a monoid structure from $A$ to $B$.

\section{Composition of morphisms}\label{sec:composition}

Unfortunately, for a general type former $\Phi$, \cref{def:phi-morphism} is not very well behaved.  In fact, it is not even guaranteed that composition of $\Phi$-morphisms is a $\Phi$-morphism, that is, $\Phi$-CwFs do not necessarily form a category.

The problem becomes apparent as soon as we consider certain ``higher order'' type formers, i.e. type formers with $\Pi$ types nested on the left.  The simplest example is:
$$
\Phi :\equiv (\U \to \U) \to \U.
$$

To make our example easier to follow, we observe that, given any set $A$, we can construct a CwF with only one context $1$, $\Ty(1) = A$, and $\Tm_1(a) = 1$ for all types $a : A$, with context extension defined in the only possible way.

If we assume that the set $A$ is equipped with a function $A^A \to A$, then its corresponding CwF can be equipped with a $\Phi$-structure.  Let us call a set equipped with such a structure a $\Phi$-set.

Given a function $f : A \to B$ between $\Phi$-sets, we say that it is a $\Phi$-morphism if it induces a $\Phi$-morphism on the corresponding $\Phi$-CwFs.  If we denote by $\phi_A$ and $\phi_B$ the $\Phi$-structures on $A$ and $B$ respectively, what this means is that for all functions $u : A \to A$ and $v : B \to B$ such that the following diagram commutes:
$$
\xymatrix{
A \ar[r]^u \ar[d]_f &
A \ar[d]^f \\
B \ar[r]_v &
B,
}
$$
we have that $f(\phi_A(u)) = \phi_B(v)$.

To show that $\Phi$-morphisms between $\Phi$-CwFs are not in general closed under composition, it is therefore enough to find $\Phi$-morphisms $f : A \to B$, $g : B \to C$ such that $g \circ f$ is not a $\Phi$-morphism.

We take $A=1$, $B=2$, $C=3$, and $f, g$ to be inclusions. The $\Phi$ structure $\phi_A$ on $A$ is the only possible one, while the $\Phi$-structure $\phi_B$ on $B$ takes a function $u : B \to B$ and returns $u(0)$.

The $\Phi$-structure $\phi_C$ on $C$ is defined as follows: given $u : C \to C$, it distinguishes two cases:
\begin{itemize}
\item if $u(2) \subseteq 2$, then $\phi_C(u) = u(0)$;
\item otherwise, $\phi_C(u) = 1$.
\end{itemize}

It is easy to see that the inclusions $A \to B$ and $B \to C$ are indeed $\Phi$-morphisms. However, if we take for example the function $u : C \to C$ that swaps $1$ and $2$ and fixes $0$, then clearly the following diagram commutes:
$$
\xymatrix{
1 \ar[r]^{\id} \ar[d]_0 &
1 \ar[d]^0 \\
C \ar[r]_u &
C,
}
$$

but $\phi_C(v) = 1 \neq 0 = \phi_A(\id)$.

\section{Special type formers}\label{sec:special-type-formers}

The notion of type formers is very general.  As shown in \cref{sec:composition}, it is possible to define ``higher order'' type formers, for which even the most basic properties are not provable.

In practice, most of the commonly employed type formers are much better behaved than in the general case.  For this reason, it is useful to single out certain specific properties of type formers that make them more suitable to be analysed.

\begin{lem}\label{lem:f-element-prop}
Let $\Phi$ be a type former, $F : \A \to \A'$ an oplax $\RF$-morphism.  Suppose $\phi$ is a global element of $\interp{\Phi}^{\A}$ and $\phi'$ a global element of $\interp{\Phi}^{\A'}$.  Then any two elements of $\interp{\Phi}^F$ over $\phi$ and $\phi'$ are equal.
\end{lem}
\begin{proof}
Let $\mathcal{R}_F'$ be subcategory of $\mathcal{R}_F$ consisting of all those objects
$$
\xymatrix{
F \Gamma & R \ar[l]_-l \ar[r]^-r & \Gamma',
}
$$
where $R$ is subterminal in the category of spans over $F \Gamma$ and $\Gamma'$.

It is not hard to see that $\mathcal{R}_F'$ is itself an $\RF$-category, and consequently the inclusion functor $i : \mathcal{R}_F' \to \mathcal{R}_F$ is a split $\RF$-morphism.

It follows that the interpretation functor $\interp{-}^{\mathcal{R}_F}$ has values in $\mathcal{R}_F'$.  In particular, $\interp{\Phi}^F$ is subterminal over $F \interp{\Phi}^{\A}$ and $\interp{\Phi}^{\A}$, which is exactly what we had to prove.
\end{proof}

\Cref{lem:f-element-prop} ensures that, if a CwF morphism $F : \C \to \D$
between $\Phi$-CwFs is a $\Phi$-morphism, then there is at most one possible choice for the element $s$ of \cref{def:phi-morphism}.

Now, given oplax $\RF$-morphisms $F : \A \to \B$ and $G : \B \to \C$, we can form the pullback $\mathcal{R}_F \times_\B \mathcal{R}_G$, which is an $\RF$-category by \cref{lem:rf-limits}, and is equipped with split morphisms $\pi_A$ and $\pi_C$ to $\A$ and $\C$ respectively.

\begin{defn}\label{def:flat-type-former}
We say that a type former $\Phi$ is \emph{flat} if for all $F, G$ as above, whenever $\interp{\Phi}^{\mathcal{R}_F \times_\B \mathcal{R}_G}$ has a global element $s$, then there is an element of $\interp{\Phi}^{GF}$ over $\pi_A(s)$ and $\pi_C(s)$.
\end{defn}

\Cref{def:flat-type-former} formalises the idea of a type former that is well-behaved with respect to composition, as the following \namecref{prop:flat-composition} shows.

\begin{prop}\label{prop:flat-composition}
If $\Phi$ is a flat type former, composition of $\Phi$-morphisms is a $\Phi$-morphism.
\end{prop}
\begin{proof}
If $F : \A \to \B$ and $G : \B \to \C$ are $\Phi$-morphisms, then $\interp{\Phi}^{\mathcal{R}_F \times_{\B} \mathcal{R}_G}$ has a global element $s$, where $\pi_A(s)$ is the $\Phi$-structure on $\A$, and $\pi_B(s)$ is the $\Phi$-structure on $\C$.

Since $\Phi$ is flat, we get a corresponding element of $\interp{\Phi}^{GF}$, showing that $GF$ is a $\Phi$-morphism.
\end{proof}

\begin{cor}\label{cor:flat-category}
Let $\Phi$ be a flat type former.  $\Phi$-CwFs, together with $\Phi$-morphisms, form a category.
\end{cor}

\begin{defn}\label{def:algebraic-type-former}
A flat type former $\Phi$ is said to be \emph{algebraic} if the category of $\Phi$-CwFs and \emph{split} $\Phi$-morphism has an initial object.
\end{defn}

All the usually considered type formers are algebraic.  In particular, all the type formers involved in the definition of an $\RF$-category are algebraic (as essentially proved by \cref{thm:rf0}), as well as all the examples of \cref{sec:rf-examples}.

\begin{prop}\label{prop:algebraic-2-initiality}
Let $\Phi$ be an algebraic type former, $\mathcal{H}$ the initial $\Phi$-CwF, and $\C$ an arbitrary $\Phi$-CwF.  Then any two $\Phi$-morphisms $F, G: \mathcal{H} \to \C$ are isomorphic.
\end{prop}
\begin{proof}
Let $\mathcal{S}$ be the $\RF$-category whose objects are triples $(P, Q, R)$, where $P$ is a presheaf on $\mathcal{H}$, $Q$ a presheaf on $\C$, and $R$ a span of the form:
$$
\xymatrix{
F^*Q & R \ar[r]\ar[l] & G^*Q.
}
$$

Let $\E$ the pseudo-equaliser of $F$ and $G$, defined like in the proof of \cref{thm:rf-2-initiality}.  Let $\pi : \E \to \mathcal{H}$ be the canonical projection.

We can define $\RF$-morphisms
$$
\xymatrix{
\mathcal{R}_{F^*} \times_{\presheaf\C} \mathcal{R}_{G^*} \ar[r] &
\mathcal{S} \ar[r] &
\mathcal{R}_{\pi^*}.
}
$$

The fact that $F$ and $G$ are both $\Phi$-morphisms determines a global element of the interpretation of $\Phi$ in $\mathcal{R}_{F^*} \times_{\presheaf\C} \mathcal{R}_{G^*}$, which can therefore be transported to $\mathcal{R}_{\pi^*}$.

It follows that $\E$ can be equipped with a $\Phi$-structure such that the CwF morphism $\pi$ is a split $\Phi$-morphism.  The conclusion now follows immediately from the initiality of $\mathcal{H}$.
\end{proof}

\begin{defn}\label{def:set-theoretic-type-former}
A type former $\Phi$ is said to be \emph{set-theoretic} if for all small categories $\A$, the CwF $\presheaf\A$ has a $\Phi$-structure.
\end{defn}

Again, all type formers considered so far are set-theoretic.  In \cref{sec:systems-of-universes} we will define a type former for a \emph{univalent universe} (\cref{def:univalent-universe}), which fails to be set-theoretic.

A type in $\RF_0$ over some type former $\Psi$ will be referred to as a \emph{type former over $\Psi$}.  Given such a type $\Phi$, we will often identify it with the corresponding context extension $\Psi.\Phi$.

\begin{defn}\label{def:phi-structure-over}
Let $\Psi$ be a type former, $\Phi$ a type former over $\Psi$, and $\C$ a CwF equipped with a $\Psi$-structure $\psi$.

A \emph{$\Phi$-structure} on $\C$ is a $\Psi.\Phi$-structure on the underlying CwF such that the induced $\Psi$ structure is equal to $\psi$.
\end{defn}

\section{Systems of universes}\label{sec:systems-of-universes}

If $(\U, \El)$ is a universe in a CwF $\C$, $\U$ induces another CwF structure on $\C$, which we shall denote with the superscript $\U$.
Types of $\C^\U$ over a context $\Gamma$ are defined by:
$$
\Ty^\U(\Gamma) :\equiv \C(\Gamma,\U).
$$
For a type $A : \Ty^\U(\Gamma)$, we define terms of $A$ as follows:
$$
\Tm^\U_\Gamma(A) :\equiv \C(\Gamma,\El[A]).
$$

There is a canonical map $\C^\U \to \C$, which is easily verified to be a CwF morphism.

\begin{defn}\label{def:universe-morphism}
Let $(\U, \El)$ and $(\U', \El')$ be universes in a CwF $\C$.  A \emph{universe morphism} $\U \to \U'$ is a CwF morphism $\C^\U \to \C^{\U'}$ that makes the following diagram commutative:
$$
\xymatrix{
\C^\U \ar[rr] \ar[rd] & &
\C^{\U'} \ar[ld] \\
& \C.
}
$$
\end{defn}

If $\C$ is equipped with a $\Phi$-structure $\phi$, it is not possible in general to restrict $\phi$ to $\C^\U$.  This justifies the following \namecref{def:phi-universe}.

\begin{defn}\label{def:phi-universe}
Let $\C$ be a $\Phi$-category, where $\Phi$ is any type former, and $(\U,\El)$ a universe in $\C$.  We say that $\U$ is a $\Phi$-universe if $\C^\U$ has a $\Phi$-structure $\phi^U$ such that the canonical map $\C^\U \to \C$ is a $\Phi$-morphism.
\end{defn}

Note that if $\Phi$ is flat, then universes over $\C$ form a category, with morphisms given by universe morphisms such that the underlying CwF morphism preserves $\Phi$-structures.

\begin{defn}
Let $\A$ be a category, $\Phi$ a flat type former and $\C$ a $\Phi$-CwF.  A \emph{system of $\Phi$-universes} on $\C$ (indexed by $\A$) is a functor from $\A$ to the category of $\Phi$-universes of $\C$.
\end{defn}

Usually, $\A$ is taken to be a poset, most commonly the ordinal $\omega$. This is the case, for example, in the type theory described in \cite{hott-book}.

\begin{lem}\label{lem:rf-finite-limits}
In any $\RF$-category $\C$, finite diagrams of fibrant objects have a limit.
\end{lem}
\begin{proof}
By \cref{lem:pullbacks-of-fibrations}, all we have to prove is that any morphism between fibrant objects of $\C$ is isomorphic to a fibration.  The following argument appears in \cite{gambino-garner:wfs}.

Let $A$ and $B$ be types over the unit context, and $f : A \to B$ any map. Define:
$$
E :\equiv (a : A) \times (b : B) \times (f(a) = b).
$$

We have a factorisation:
$$
\xymatrix{
A \ar[r]^-i &
E \ar[r]^-p &
B,
}
$$
and it is easy to see that $i$ is an isomorphism, and $p$ is a fibration.
\end{proof}

\newcommand{\UNIV}{\textsc{univ}}

\begin{prop}\label{prop:finite-systems-of-universes}
Let $\A$ be a finite category and $\Phi$ a flat type former.  There is a type former $\UNIV^{\Phi,\A}$ such that systems of $\Phi$-universes indexed by $\A$ are in bijective correspondence with $\UNIV^{\Phi,\A}$-structures on $\C$.
\end{prop}
\begin{proof}
Define:
$$
\Psi :\equiv (U : \U) \times (\el : \El\ U \to \U).
$$

Clearly, a $\Psi$-structure is the same as a universe.  Furthermore, $(\El[U],\El[\lambda^{-1}(\el)])$ is a universe in $\RF_0 / \Psi$, which we will also denote with $U$.  Therefore, $\RF_0 / \Psi$ is an $\RF$-category with universe $U$.

If $\C$ is a CwF equipped with a universe $\V$, we get an interpretation functor $\interp{-}^{\presheaf\C} : \RF_0 / \Psi$ mapping $U$ to $\Ty^V$.

It follows that, if we define $\UNIV^{\Phi,1}$ be the interpretation of $\Phi$ in $\RF_0 / \Psi$, a $\UNIV^{\Phi,1}$-structure in $\C$ is the same as a $\Phi$-universe $\V$ in $\C$.

Now, let $\I$ be the category with two objects 0 and 1, and only one non-identity morphism in $\I(0,1)$. Define a type former $\Psi_2$ as follows:
\begin{align*}
\Psi_2
& :\equiv ((U,\el) : \Psi) \\
& \times ((U',\el' : \Psi)) \\
& \times (f_0 : \El\ U \to \El\ U') \\
& \times (f_1 : (X : \El\ U) \to \El(\el(X)) \to \El(\el'(f_0(X))).
\end{align*}

Clearly, a $\Psi_2$-structure is the same as a pair of universes, together with a universe morphism, i.e. a system of universes indexed by $\I$.

Again, if $\C$ is equipped with universes $\V$ and $\V'$, there is an interpretation functor $\interp{-}^{\presheaf\C} : \RF_0 / \Psi_2$ that maps the two universes $U$ and $U'$ in $\RF_0 / \Psi_2$ to $V$ and $V'$ respectively.

$\RF_0 / \Psi_2$ can be regarded as an $\RF$-category, where the universe is defined to be:
$$
(A : \U) \times (A' : \U') \times (f(A) = A').
$$

Consequently, if we define $\UNIV^{\Phi,\I}$ to be the interpretation of $\Phi$ in $\RF_0 / \Psi_2$, it is easy to see that a $\UNIV^{\Phi,\I}$-structure on $\C$ is the same as a system of $\Phi$-universes indexed by $\I$.

Now the general case follows from \cref{lem:rf-finite-limits} and the fact that every finite category is a finite colimit of 1 and $\I$ in $\cat{Cat}$.
\end{proof}

\subsection{Univalent universes}\label{sec:univalent-universes}

Let $\C$ be a $\Phi_0$-CwF where $\Phi_0$ is defined as:
$$
\Phi_0 :\equiv \Phi^\Pi \times \Phi^\Sigma \times \Phi^{\textsc{ieq}},
$$
and let $(\U, \El)$ be a universe in $\C$.

We can define the property of a function being an equivalence, internally in $\C$, as follows.

Over the context $(A,B : \U)(f : A \to B)$, define a type $\isequiv$:
\begin{align*}
\isequiv
& :\equiv ((g : B \to A) \times (g \circ f = \id)) \\
& \times ((g : B \to A) \times (f \circ g = \id)).
\end{align*}

Here $\id$ and $\circ$ denote the identity function and composition of functions internal to $\C$, respectively, defined in the obvious way using the $\Pi$-type structure on $\C$.

The type $\Equiv$ of equivalences is defined over the context $(A,B : \U)$:
$$
\Equiv :\equiv (f : A \to B) \times \isequiv[f].
$$

It is easy to define a term $\id_E : \Equiv[A,A]$ over the context $(A : \U)$, corresponding to the identity equivalence.  From the properties of equality, it follows that there exists a function $\coerce : A = B \to \Equiv[A,B]$ in the context $(A,B : \U)$.

\emph{Univalence} for $\U$ is the following type, in the unit context:
$$
\ua_\U :\equiv (A, B : \U) \to \isequiv[A = B, \Equiv[A,B], \coerce[A,B]].
$$

\begin{defn}\label{def:univalent-universe}
The universe $\U$ is said to be \emph{univalent} if the corresponding univalence type $\ua_\U$ has a global element.
\end{defn}

\begin{prop}\label{def:ua-type-former}
There is a type former $\Phi^{\textsc{ua}}$ over $\Phi$, such that a $\Phi^{\textsc{ua}}$-structure over $\Phi_0$-CwF $\C$ is the same as a univalent universe.
\end{prop}
\begin{proof}
Univalence can be defined internally in any $\Phi_0$-CwF, hence in particular in $\RF_0 / \Phi_0$.
\end{proof}

\section{Further work}\label{sec:tformers-further-work}

The definitions of special type formers given in \cref{sec:special-type-formers} serve their purpose of allowing a workable theory of type formers to be developed, but could be considered rather unsatisfactory, since they involve quantification over arbitrary functors, and it is thus hard to verify in practice that a given type former possesses those properties.

It seems reasonable that, at least for the case of \emph{flat} and \emph{algebraic} type formers, one should be able to verify that a type formers falls in one of those classes simply by inspecting the type expression in $\RF_0$ that defines it.

For example, it appears to be the case that if a type former is written only using ``first-order'' $\Pi$-types of non-small types, then it is automatically flat.  All the usual type formers, at least the ones that we used or mentioned, have this form, and the example of non-flat type former given in \cref{sec:composition} is indeed higher order.

It also seem likely that there should exist a notion of ``strict positivity'' for type formers, and those type formers that turn out to be strictly positive ought to be algebraic.

Investigating these and similar syntactic characterisations for type formers will be the goal of future research.

\chapter{Two-level type theory}\label{chap:two-levels}

In this \namecref{chap:two-levels}, we will develop the idea of \emph{two-level} type theory, modelled by CwFs with two type functors.
Such systems are motivated by the need to introduce an internalised notion of \emph{strict equality} into the theory.

Since certain type formers will play a special role within a two-level CwF, we single out CwFs with a fixed basic structure:

\begin{defn}\label{def:model}
A \emph{model of type theory} is a CwF equipped with $\Pi$, $\Sigma$ and unit type structures.  Given models of type theory $\C$ and $\D$, a morphism between them is a CwF morphism that preserves the $\Pi$, $\Sigma$ and unit type structures.
\end{defn}

We will write $\T$ to denote the type former corresponding to $\Pi$, $\Sigma$ and unit types, so that a model of type theory is simply a CwF with a $\T$-structure.  In other words:
$$
\T :\equiv \Phi^\Pi \times \Phi^\Sigma \times \Phi^{\textsc{unit}}.
$$

\begin{example}\label{ex:presheaf-model}
If $\C$ is an arbitrary category, the presheaf category $\presheaf\C$ is a model of type theory.
\end{example}

We will often simply say \emph{model} instead of \emph{model of type theory}.  In particular, the structure needed to make a category (or a CwF) into a model will often be referred to as a \emph{model structure}.  Note that our notion of model structure is completely unrelated to that of \emph{Quillen model structure} \cite{quillen-homotopical-algebra}.  No confusion is possible, however, since we will never refer to the latter.

If $\Phi$ is a type former over $\T$, CwFs equipped with a $\Phi$-structure will be referred to as $\Phi$-models.
If $\Phi$ is flat, $\Phi$-models of type theory form a category $\mathcal{M}^\Phi$.
In particular, the trivial type former over $\T$ is flat, and its corresponding category of models will be denoted simply by $\mathcal{M}$.

To incorporate strict equality into type theory, we will need to make a distinction between arbitrary types, and types for which weak equality is well defined.  This is necessary, because, as we will see in \cref{lem:all-fibrant-collapse}, the theory becomes degenerate if we don't make this distinction.

\begin{defn}\label{def:dual-cwf}
A \emph{two-level CwF} is a CwF $\C$, equipped with a functor $\fibrant \Ty : \C \to \op{\set}$, and a natural transformation $|-| : \fibrant \Ty \to \Ty$.
\end{defn}

Given a two-level CwF $\C$, we can define a second CwF structure on $\C$ having $\fibrant \Ty$ as the type functor, and where terms are given by $\fibrant \Tm_\Gamma(A) :\equiv \Tm_\Gamma(|A|)$. Context extension is similarly defined as $\Gamma.A :\equiv \Gamma.|A|$.  We will write $\fibrant \C$ to denote $\C$ equipped with this second CwF structure.  To avoid confusion, and for consistency with notations that we will introduce later, we will write $\strict \C$, $\strict \Ty$ and $\strict \Tm$ when referring to the original CwF structure on $\C$.

The natural transformation $|-|$ induces a split CwF morphism $\strict \C \to \fibrant \C$.

\begin{defn}\label{def:two-level}
A \emph{two-level model of type theory} is a two-level CwF $\C$ such that both $\strict \C$ and $\fibrant \C$ are models of type theory, and $|-|$ is a $\T$-morphism.
\end{defn}

If $\Phi$ and $\Psi$ are type formers over $\T$, we define a $(\Phi,\Psi)$-model to be a two-level model $\C$ where $\fibrant\C$ is equipped with a $\Phi$-structure, and $\strict\C$ is equipped with a $\Psi$-structure.

The simplest way to construct a two-level CwF is with a universe:

\begin{remark}\label{remark:two-level-universe}
Let $\C$ be a CwF equipped with a universe $\U$, $\El$.  Define $\fibrant \Ty(\Gamma) :\equiv \Tm_\Gamma(\U)$, and for $A : \fibrant \Ty(\Gamma)$, let $|A| :\equiv \El[A]$.

Then $\C$, with the above choice of fibrant type functor, is a two-level CwF.
\end{remark}

For all CwFs $\C$, the presheaf category $\presheaf \C$ is a two-level CwF, where we can use the fibrant universe to define fibrant types as in \cref{remark:two-level-universe}.

\section{The simplicial model}\label{sec:simplicial-model}

The reference example of a two-level model is given by the category of \emph{simplicial sets}, whose definition we recall below.

\begin{defn}
The \emph{simplicial category} $\Delta$ has the natural numbers as objects, and morphisms $\Delta(n,m)$ are defined to be monotone functions $[n] \to [m]$, where $[k]$ denotes the set of natural numbers less or equal to $k$.
\end{defn}

\begin{defn}
A \emph{simplicial set} is a presheaf on $\Delta$.
\end{defn}

Simplicial sets form a category $\sset$, that can be regarded as a model of type theory like any presheaf category (\cref{ex:presheaf-model}).

We can then define two-level model structure on $\sset$ as follows: for all contexts $\Gamma$, fibrant types $\fibrant\Ty(\Gamma)$ are defined to be the subset of $\Ty(\Gamma)$ of those types $A$ such that the display map $\Gamma.A \to \Gamma$ is a \emph{Kan fibration}. Since Kan fibrations are closed under $\Pi$ and $\Sigma$ type formation \cite{simplicial-model}, the fibrant fragment of $\sset$ admits $\Pi$ and $\Sigma$ type formers, hence $\sset$ is a two-level model of type theory.

Note that the definition of types used in \cite{simplicial-model} differs from the one we have given here.  However, it can be easily verified that all the constructions carry over to our definition.  Following \cite{simplicial-model}, then, it can be shown that the fibrant fragment of $\sset$ models all of the commonly used type formers, including a univalent universe.

\section{Presheaf models}\label{sec:two-level-presheaves}

In this section, we will show that, given a model of type theory $\C$, its presheaf category $\presheaf \C$ can be regarded as a two-level model. Furthermore, if $\C$ is equipped with a $\Phi$-structure for some type former $\Phi$, one can find the same $\Phi$-structure on the fibrant fragment of $\presheaf \C$.

The idea of the proof is very simple: we start with a model $\C$ and build a two-level model structure on $\presheaf\C$.  The strict fragment of $\presheaf\C$ is obtained from the usual CwF structure on presheaf categories (\cref{sec:presheaves}).  Fibrant types on $\presheaf\C$ are given by the fibrant universe (\cref{sec:basic-type-formers}), and the fibrant model structure is inherited from that of $\C$.

The problem with this approach is that the resulting morphism from the fibrant to the strict fragment does not preserve type formers strictly.  For example, let $A,B : \presheaf\C(1, \Ty)$ be fibrant types over the unit context.  If we form their $\Sigma$-type within the fibrant model structure, then convert it to a strict type, we get the presheaf $P$ given by:
$$
P_\Gamma = \Tm_{\Gamma}(\Sigma_A B).
$$

However, if we convert both $A$ and $B$ to strict types first, then take their $\Sigma$-type, we end up with a presheaf $Q$, where, $Q_\Gamma$ is a set of pairs of terms of $A$ and $B$ over $\Gamma$.

Of course, terms of $\Sigma_A B$ can be identified to the set of such pairs, but the two resulting presheaves, although isomorphic, are not equal on the nose.  A similar problem occurs with $\Pi$-types.  Therefore, the resulting structure on $\presheaf\C$ does not satisfy the definition of two-level model (\cref{def:two-level}).

For this reason, we need to slightly modify the CwF structure on $\C$, so that strict preservation of $\Pi$ and $\Sigma$ can be achieved.  This will be the aim of the following subsections.

\subsection{Lifting type formers}\label{sec:lifting-type-formers}

Let $\C$ be a CwF equipped with a $\Phi$-structure $\phi$.  The fibrant universe $\Ty$ determines a CwF structure on $\presheaf\C$, where types are given by
$$
\fibrant{\presheaf\Ty}(P) :\equiv \presheaf\C(P, \Ty).
$$
Let us denote by $\fibrant{\presheaf\C}$ the corresponding CwF.

\begin{lem}\label{lem:yoneda-fibrant}
The yoneda embedding $y : \C \to \fibrant{\presheaf\C}$ can be extended to a CwF morphism.
\end{lem}
\begin{proof}
By the Yoneda lemma, $\fibrant\Ty(y\Gamma) \cong \Ty(\Gamma)$, hence we can take this isomorphism as the action of $y$ on types.  Consequently, $y$ can be defined to be an isomorphism on terms as well.
\end{proof}

In this section, we will show how to lift $\phi$ to a $\Phi$-structure on $\fibrant{\presheaf\Ty}$, so that the Yoneda embedding of \cref{lem:yoneda-fibrant} is a $\Phi$-morphism.

\begin{lem}\label{lem:equivalence-rf}
Let $f : \C \to \D$ be a CwF equivalence between $\RF$-categories.  If $f$ weakly preserves the universe, then $f$ is a weak $\RF$-morphism.
\end{lem}
\begin{proof}
Since $f$ is a CwF equivalence, we can use it to transport all the type structures from $\C$ to $\D$.  Since all the type structures of an $\RF$-category except the universe are characterised by a universal property, it easily follows that the transported structures are isomorphic to the original ones on $\D$, which amounts to saying that $f$ preserves them.
\end{proof}

\begin{thm}\label{thm:morita-lift}
Let $f : \C \to \D$ be a CwF morphism such that $f^* : \presheaf\D \to \presheaf\C$ is an equivalence of categories, and $f$ is bijective on types.  Suppose $\C$ is equipped with a $\Phi$-structure $\phi$. Then there exists a $\Phi$-structure $\phi'$ on $\D$ such that $f$ is a $\Phi$-morphism.
\end{thm}
\begin{proof}
Since $f^*$ is an equivalence of categories, it induces a equivalences of slice categories, hence an isomorphism of the type functors of $\presheaf\D$ and $\presheaf\C$ thanks to \cref{prop:presheaf-families}.  Therefore, $f^*$ is a CwF equivalence.

Note that $f$ being bijective on types is equivalent to $f^*$ preserving the universe.  Hence, it follows from \cref{lem:equivalence-rf} that $f^*$ is a weak $\RF$-morphism.

Let $f_! : \presheaf\C \to \presheaf\D$ be the left adjoint of $f^*$. Explicitly, $f_!$ is given by the left Kan extension of $y \circ f$ along $y : \C \to \presheaf\C$.  In this case, $f_!$ is also a CwF equivalence, hence a weak $\RF$-morphism.

We define a CwF morphism $\widetilde f : \presheaf\C \to \mathcal{R}_{f^*}$.  On objects, $\widetilde f$ maps $P$ to the triple $(P, f_!(P), \delta_P)$, where $\delta_P$ denotes the span:
$$
\xymatrix{
f^*f_!P & P \ar[l]_-{\eta} \ar[r] & P.
}
$$
The action of $\widetilde f$ on types is defined similarly.  It is not hard to check that $\widetilde f$ is a weak $\RF$-morphism.

Therefore, the diagram:
$$
\xymatrix{
& \RF_0 \ar[ld]\ar[rd] & \\
\presheaf\C \ar[rr]_{\widetilde f} & & \mathcal{R}_f
}
$$
commutes weakly by \cref{thm:rf-2-initiality}.  It follows that $\widetilde f(\phi)$ determines a canonical $\Phi$-structure on $\presheaf\D$, such that $f$ is a $\Phi$-morphism, as required.
\end{proof}

\begin{thm}\label{thm:fibrant-lift}
There is a $\Phi$-structure on $\fibrant{\presheaf\C}$ such that the Yoneda embedding (\cref{lem:yoneda-fibrant}) preserves $\Phi$-structures.
\end{thm}
\begin{proof}
The Yoneda embedding $y : C \to \fibrant{\presheaf\C}$ is bijective on types, and the induced functor $y^*$ is an equivalence.  Therefore, \cref{thm:morita-lift} applies directly.
\end{proof}

\subsection{Regular models}\label{sec:regularisation}

Let $\C$ be a model, and consider the category $\presheaf \C / \Ty$ of presheaves over $\Ty$. If $X$ is such a presheaf, we denote by $|-|_X$ the corresponding morphism to $\Ty$.

For a presheaf $X$ over $\Ty$, regard $X$ as a type in the unit context of the CwF $\presheaf\C$, and denote by $X^{(2)}_\Gamma$ the presheaf corresponding to the type:
$$
\Sigma_{A : X}\Pi_{\Tm[|A|_X]} X.
$$

\begin{lem}\label{lem:X2-pairs}
For all context $\Gamma : \C$, the set $X^{(2)}_\Gamma$ is naturally isomorphic to the set of pairs $(A,B)$, where $A : X_\Gamma$ and $B : X_{\Gamma.|A|_X}$.
\end{lem}
\begin{proof}
Immediate consequence of \cref{lem:yoneda-terms}.
\end{proof}

In the following, we will use the isomorphic representation of $X^{(2)}$ given by \cref{lem:X2-pairs} liberally.

Note that if $\Ty$ is regarded as an element of $\presheaf\C/\Ty$, the presheaf $\Ty^{(2)}$ matches with the one we defined in \cref{sec:basic-type-formers}.  

We can make $X^{(2)}$ into a presheaf over $\Ty$ in at least two ways: using the $\Pi$ or $\Sigma$-type structures on $\C$. In fact, they both can be regarded as morphisms:
$$
\Ty^{(2)} \to \Ty,
$$
from which we obtain the desired morphism $X^{(2)} \to \Ty$ by composing with the obvious map $X^{(2)} \to \Ty^{(2)}$.

We now define an endofunctor $E$ of $\presheaf\C / \Ty$ as:
$$
EX :\equiv X^{(2)} + X^{(2)} + 1,
$$

where the map $EX \to \Ty$ on the first $X^{(2)}$ component is given by the $\Pi$-type structure on $\C$ as explained above, on the second component by the $\Sigma$-type structure, and on the third component it just selects the unit type.

Denote by:
\begin{align*}
& \pi^E : X^{(2)} \to EX \\
& \sigma^E : X^{(2)} \to EX \\
& u^E : 1 \to EX
\end{align*}
the three canonical injections into the coproduct $EX$.

\begin{prop}\label{prop:E-finitary}
The endofunctor $E : \presheaf\C / \Ty \to \presheaf\C / \Ty$ is finitary.
\end{prop}
\begin{proof}
Clear from the characterisation of \cref{lem:X2-pairs}.
\end{proof}

It follows from \cref{prop:E-finitary} that $E$ admits a free monad $E^*$.

An element of $(E^*X)_\Gamma$ is either a \emph{base element} $\eta(A)$, where $A : X_\Gamma$ and $\eta : X \to E^*X$ is the unit of the monad $E^*$, or a \emph{compound element} of the form $\pi^E(A, B)$, $\sigma^E(A, B)$ or $u^E$.

Here, we are abusing notation by writing $\pi^E$ for the canonical map $(EX)^{(2)} \to EX$ given by the free monad construction, and similarly for $\sigma^E$ and $u^E$.

The idea of this construction becomes clear when we try to apply $E^*$ to $\Ty$ itself.  The resulting presheaf $E^*\Ty$ can be regarded as an alternative type functor on $\C$ where types can be uniformly be classified into base types, $\Pi$-types, $\Sigma$-types or unit types.

This is made precise by the following.

\begin{lem}\label{lem:e-star-model}
For any model $\C$, the presheaf $E^*\Ty$ can be extended to a model structure.
\end{lem}
\begin{proof}
For an element $A : E^*\Ty$, define its set of terms simply as $\Tm_\Gamma(|A|)$.  This clearly equips $\C$ with a CwF structure, where context extension is given by $\Gamma.A :\equiv \Gamma.|A|$.

The rest of the structure can be obtained directly from the decomposition of $E^*$: the $\Pi$-type structure is given by $\pi^E$, the $\Sigma$-type structure by $\sigma^E$ and the unit type structure by $u^E$.  Verifying all the required properties is straightforward.
\end{proof}

If $\C$ is a category and $\Phi$ is a flat type former over $\T$, write $\mathcal{M}^\Phi_\C$ for the subcategory of $\mathcal{M}^\Phi$ consisting of $\Phi$-models that have $\C$ as the underlying category, and $\Phi$-morphisms that have the identity as the underlying functor.  We refer to $\mathcal{M}^\Phi_\C$ as the \emph{category of $\Phi$-model structures on $\C$}.

Similarly, $\mathcal{M}_\C$ denotes the category of model structures on $\C$ (without any additional structure).

\Cref{lem:e-star-model} implies that $E^*$ induces an endofunctor on $\mathcal{M}_\C$ for all models $\C$.

\begin{lem}\label{lem:e-star-comonad}
The endofunctor determined by $E^*$ is a comonad on $\mathcal{M}_\C$.
\end{lem}
\begin{proof}
A morphism $\epsilon$, serving as the counit of the comonad, can be obtained directly from the map $|-| : E^*\Ty \to \Ty$.  All we need to do to make $\epsilon$ into the unit of a comonad is to show that it induces a model morphism.  Indeed, this is readily verified, since the model structure corresponding to $E^*\Ty$ is defined in terms of $|-|$ itself.

To define the comonad multiplication $\delta : E^*\Ty \to E^*(E^*\Ty)$, we proceed by induction on the structure of $E^*$, and at the same time show that $|\delta(X)| = |X|$ for all $X$.

Let $\Gamma : \C$, and $X : E^*\Ty(\Gamma)$.
\begin{itemize}
\item If $X = \eta(A)$ for some $A : \Ty(\Gamma)$, set $\delta(X) :\equiv
\eta(\eta(A))$. Then clearly $|\delta(X)| = A = |X|$.
\item If $X = \pi^E(A,B)$, we have by induction hypothesis $\delta(A) :
\E^*(\E^*\Ty)(\Gamma)$, and, modulo an application of the isomorphism of
\cref{lem:X2-pairs}, $\delta(B) : \E^*(\E^*\Ty)(\Gamma.|A|)$.  Since $|A| =
|\delta(A)|$ by the induction hypothesis, we can set $\delta(X) :\equiv
\pi^E(\delta(A), \delta(B))$, and observe that $|\delta(X)| = \Pi_{|A|} |B| =
|\pi^E(A,B)| = |X|$, as required.
\item If $X = \sigma^E(A,B)$, we proceed exactly like for the $\pi^E$ case above.
\item If $X = u^E$, we set $\delta(X) :\equiv u^E$, and the required equation
obviously holds.
\end{itemize}

The fact that $\delta$ is a model morphism follows immediately from its definition.

One of the comonad laws has already been proved as part of the definition of $\delta$, and the others can be easily verified.
\end{proof}

\Cref{lem:e-star-comonad} may seem surprising at first, since $E^*$ is defined as a (free) monad, while it turns out to be a comonad when regarded as an endofunctor of model structures.  However, $E^*$ is already a comonad on $\presheaf\C/\Ty$, so all that \cref{lem:e-star-comonad} states is that this structure carries over.

On the other hand, $E^*$ is \emph{not} a monad on $\mathcal{M}_\C$, since for example the unit $\eta : \Id \to E^*$ cannot be regarded as a model morphism.

\begin{thm}\label{thm:regular-yoneda}
Let $\C$ be a model. The Yoneda embedding $y : \C \to \presheaf\C$ can be extended to a model morphism between $E^*\Ty$ and the canonical model structure on $\presheaf\C$ defined in \cref{sec:presheaves}.
\end{thm}
\begin{proof}
We have already defined an action of $y$ on types, denoted $y_0 : \Ty \to \presheaf\Ty(y -)$.  The function $y_0$ maps a type $A : \Ty(\Gamma)$ to the functor $(\Delta, \sigma) \mapsto \Tm_\Delta(A[\sigma])$ (\cref{def:y0}).  We observed that $y_0$ is not in general a model morphism.

For an element $X : E^*\Ty(\Gamma)$, we will define $y(X) : \presheaf\Ty(y \Gamma)$ by induction on the structure of $E^*\Ty(\Gamma)$, and at the same time we will construct a natural isomorphism $y(X) \cong y_0(|X|)$.
\begin{itemize}
\item if $X = \eta(A)$ for some type $A : \Ty(\Gamma)$, let $y(X) :\equiv y_0(X)$ and the isomorphism be the identity;
\item if $X = \pi^E(A, B)$, we get by induction hypothesis a type $y(A) : \presheaf\Ty(y\Gamma)$; similarly, using \cref{lem:X2-pairs}, we get a type $y'(B) : \presheaf\Ty(y(\Gamma.A))$. Now, $y(\Gamma.A) \cong y(\Gamma).y_0(A) \cong y(\Gamma).y(A)$, hence we can set $y(X) :\equiv \Pi_{y(A)}y'(B)$.  It follows from the definition of $\Pi$-types in $\C$ that $y(X) \cong y_0(|X|)$.
\item if $X = \sigma^E(A, B)$ or $X = u^E$, we proceed similarly to the case of $\pi^E$.
\end{itemize}

We can prove that $y : \Ty(\Gamma) \to \presheaf\Ty(y\Gamma)$ is natural in $\Gamma$ by induction on its argument, and using the fact that $y_0$ is natural (\cref{prop:y0-embedding}) as the base case.

At this point, since the definition of the $\Pi$-type structure on $E^*\Ty$ is given precisely by $\pi^E$, it is easy to verify that $y$ as defined above does indeed preserve $\Pi$-types strictly, and a similar argument shows that $y$ preserves all type formers, hence it is a model morphism.
\end{proof}

\begin{defn}
A $\Phi$-model $\C$ is said to be \emph{(weakly) regular} if it is equipped with a model morphism $\theta : \Ty \to E^*\Ty$ on the category of $\Phi$-model structures of $\C$.  $\C$ is said to be \emph{strongly regular} if $\theta$ is a coalgebra of the comonad $E^*$.
\end{defn}

\begin{defn}\label{def:regular-type-former}
An type former $\Phi$ over $\T$ is said to be \emph{regular} if $E^*$ maps $\Phi$-models into $\Phi$-models.
\end{defn}

If $\Phi$ is flat, we can say that $\Phi$ is regular if and only if $E^*$ can be extended to an endofunctor (hence a comonad) on $\mathcal{M}^\Phi_\C$ for any $\Phi$-model $\C$.  Clearly, the trivial type former over $\T$ is regular by \cref{lem:e-star-model}.

\begin{prop}\label{prop:regular-initial}
If $\Phi$ a regular algebraic type former over $\T$, then the initial $\Phi$-model is strongly regular.
\end{prop}
\begin{proof}
Let $\S$ be the initial $\Phi$-model.  The existence of $\theta : \Ty^\S \to E^*\Ty^\S$ is an immediate consequence of the initiality of $\S$.
\end{proof}

\begin{thm}\label{thm:two-level-presheaves}
Let $\Phi$ and $\Psi$ be type formers over $\T$, with $\Phi$ regular and $\Psi$ set theoretic, and let $\C$ be a regular model. The presheaf category $\presheaf\C$ can be equipped with a $(\Phi,\Psi)$-model structure such that the Yoneda embedding $y : \C \to \presheaf\C$ can be extended to a $\Phi$-morphism between $\C$ and the fibrant fragment of $\presheaf\C$.
\end{thm}
\begin{proof}
Let $\C'$ be the model obtained from $\C$ by replacing $\Ty$ with $E^*\Ty$.  We know from \cref{thm:fibrant-lift} that $\fibrant{\presheaf\C}$ can be made into a $\Phi$-model and $y : \C' \to \fibrant{\presheaf\C}$ can be extended to a $\Phi$-morphism.  Furthermore, $\strict{\presheaf\C}$ is a $\Psi$-model by the assumption that $\Psi$ is set-theoretic.

Since $\theta : \C \to \C'$ is a model morphism, all we have to do is define the rest of the two-level model structure on $\presheaf\C$.

The non-obvious bit is how to define the coercion map $|-| : \fibrant{\presheaf\Ty} \to \strict{\presheaf\Ty}$.  Fortunately, most of the hard work is already contained in the proof of \cref{thm:regular-yoneda}.

For all contexts $P : \presheaf\C$, and $A : \presheaf\Ty(P)$, set:
$$
|A|_\Gamma(x) :\equiv y(A_\Gamma(x))_\Gamma(\id).
$$

From naturality of $y$, it follows that:
$$
|A|[x] = y(A_\Gamma(x)).
$$

Now, consider a pair $(A, B)$ in the fibrant fragment. Its related pair is given (\cref{lem:related-pairs}) by $(|A|, |B|)$, where we have used the isomorphism of \cref{cor:ty-ext} implicitly.

Now we compute:
\begin{align*}
|\Pi_A B|_\Gamma(x)
& = y((\Pi_A B)_\Gamma(x))_\Gamma(\id) \\
& = y(\Pi_{A_\Gamma(x)} \widetilde B_\Gamma(x))_\Gamma(\id) \\
& = (\Pi_{y(A_\Gamma(x)} y'(\widetilde B_\Gamma(x)))_\Gamma(\id) \\
& = (\Pi_{|A|[x]} |B|[x^+])_\Gamma(\id) \\
& = ((\Pi_{|A|} |B|)[x])_\Gamma(\id) \\
& = (\Pi_{|A|} |B|)_\Gamma(x).
\end{align*}

It follows that $|-|$ preserves $\Pi$ types strictly.  A similar verification for $\Sigma$ and the unit type shows that $|-|$ is a model morphism, concluding the proof.
\end{proof}

\subsection{Conservativity}\label{sec:conservativity}

An important consequence of the results of \cref{sec:regularisation} is the following \emph{conservativity} result.

\begin{thm}\label{thm:conservativity}
Let $\Phi$ be a regular algebraic type former, and $\Psi$ a set-theoretic type former over $\T$. Assume that the category of $(\Phi,\Psi)$-models has an initial object $\mathcal{S}$, and let $\mathcal{H}$ be the initial $\Phi$-model.
Let $H : \mathcal{H} \to \mathcal{S}$ be the unique morphism to the fibrant fragment of $\mathcal{S}$.

Let $\Gamma : \mathcal{H}$ and $A : \Ty(\Gamma)$.  If $H(A)$ is inhabited in $\mathcal{S}$, then $A$ is inhabited in $\mathcal{H}$.
\end{thm}
\begin{proof}
Consider the diagram:
$$
\xymatrix{
\mathcal{H} \ar[rr]^y \ar[dr]_H & &
\presheaf{\mathcal{H}} \\
& \mathcal{S}, \ar[ur]
}
$$

where $\presheaf{\mathcal{H}}$ is regarded as a $(\Phi,\Psi)$-model as in \cref{thm:two-level-presheaves}.

Since $\mathcal{H}$ is the initial $\Phi$-model, this diagram commutes weakly by \cref{prop:algebraic-2-initiality}.

Therefore, if $H(A)$ is inhabited in $\mathcal{S}$, it is also inhabited in $\presheaf{\mathcal{H}}$, hence in $\mathcal{H}$, since the Yoneda embedding is full.
\end{proof}

Theorem \cref{thm:conservativity} states that to prove a proposition or construct a value in a model, it is enough to prove it or construct it in a corresponding two-level model.  The type formers $\Phi$ and $\Psi$ appearing in \cref{thm:conservativity} specify the choice of structure for the fibrant fragment and strict fragment of the two level model, respectively.
They are both type formers over $\T$, because they share the common structure of a model of type theory, which, according to \cref{def:two-level}, has to be preserved by the coercion morphism from fibrant to strict types.

Note that individual type formers outside of the common fragment in $\T$ may be duplicated across $\Phi$ and $\Psi$.  For example, both $\Phi$ and $\Psi$ could contain the type former for binary sums $\Phi^{\textsc{sum}}$ introduced in \cref{sec:rf-examples}.  This is not a problem, but it is important to note that the common type formers outside of $\T$ need not be preserved by the coercion morphism.

Regularity of $\Phi$ is important, because without it we cannot make sure that preservation of the basic type former $\T$ is strict.  It could be possible to define a weaker notion of two-level model of type theory that, unlike \cref{def:two-level}, does not require the basic type formers to be preserved strictly by the coercion morphism.  In that case, it would be possible to remove the regularity assumption from the hypotheses of \cref{thm:conservativity}.

The intended application of \cref{thm:conservativity} is to a setting where $\Phi$ contains the type formers of a theory like $\hott$, and $\Psi$ the ones for a version of strict type theory, either something like our $\RF$, or alternatively a theory with just UIP and function extensionality (\cref{sec:rf-examples}).  In \cref{sec:two-level-intro} we will describe such a setting in detail.

\section{Two-level type formers}

When building a two-level system, one can specify type formers $\Phi$ and $\Psi$ over $\T$, and that gives a notion of $(\Phi, \Psi)$-model that one can work with.

This way, the type formers of $\Phi$ and $\Psi$, except for their $\T$ fragment, are completely independent, which means that the strict and fibrant fragment of a $(\Phi,\Psi)$-model do not interact outside of their common model of type theory.

Sometimes, however, it might be desirable to put structures on top of a two-level model that make full use of the two fragments.  To make this possible, we will define a notion of \emph{two-level type former}.

\begin{defn}
A \emph{two-level $\RF$-category} is an $\RF$-category with an additional universe $(\fibrant\U, \fibrant\El)$, and a morphism of universes $\fibrant\U \to \U$.
\end{defn}

To avoid confusion, we will denote the first universe in an $\RF$-category with $(\strict\U, \strict\El)$.  We will often keep the morphism $\fibrant\U \to \strict\U$ implicit when writing out types and terms in a two-level $\RF$-category.

Similarly to what we did in \cref{sec:type-formers}, we can define a category $\mathcal{RF}^2$ of two-level $\RF$-categories, and show that it has an initial object $\RF_0^2$.

Consequently, we get the corresponding notions of \emph{two-level type former} and \emph{two-level structure} for a two-level CwF.

Furthermore, there is a two-level type former $\T^2$ corresponding to the statements that both $\fibrant\U$ and $\strict\U$ are $\T$-universes, and that the map $\fibrant\U \to \strict\U$ is a $\T$-universe morphism.

Correspondingly, for a two-level type former $\Phi$ over $\T^2$, we get a corresponding notion of two-level $\Phi$-model.

Note that a type former $\Phi$ can be regarded as a two-level type former in two ways, either by lifting it to $\fibrant\U$ or $\strict\U$.  If $\Phi$ is over $\T$, then either of its liftings to $\RF_0^2$ are over $\T^2$.

In particular, given type formers $\Phi$ and $\Psi$ over $\T$, we can lift $\Phi$ to a two-level type former $\fibrant\Phi$ on $\fibrant\U$, $\Psi$ to a two-level type former $\strict\Psi$ on $\strict\U$, and obtain a two-level type former $\fibrant\Phi \times \strict\Psi$ over $\T^2$.  Then $(\Phi,\Psi)$-models are the same as two-level $\fibrant\Phi \times \strict\Psi$-models.

We can then prove a more general version of \cref{thm:conservativity} for models of two-level type formers.

\begin{defn}\label{def:set-theoretic2}
Let $\Phi$ be a type former over $\T$, and $\Psi$ a two-level type former over $\fibrant\Phi \times_T \T^2$.  We say that $\Psi$ is \emph{set-theoretic} if for all regular $\Phi$-models $\C$, the presheaf category $\presheaf\C$ is a two-level $\Psi$-model.
\end{defn}

\begin{thm}\label{thm:conservativity2}
Let $\Phi$ be a regular algebraic type former over $\T$, and $\Psi$ a set-theoretic two-level type former over $\Phi \times_\T \T^2$.
Assume that the category of two-level $\Psi$-models has an initial object $\mathcal{S}$, and let $\mathcal{H}$ be the initial $\Phi$-model.
Let $H : \mathcal{H} \to \mathcal{S}$ be the unique morphism to the fibrant fragment of $\mathcal{S}$.

Let $\Gamma : \mathcal{H}$ and $A : \Ty(\Gamma)$.  If $H(A)$ is inhabited in $\mathcal{S}$, then $A$ is inhabited in $\mathcal{H}$.
\end{thm}
\begin{proof}
Completely analogous to the proof of \cref{thm:conservativity}.
\end{proof}

\chapter{Type theory with strict equality}
\label{chap:type-theory}

In this \namecref{chap:type-theory}, we fix a specific two-level model of type theory, and work internally in it.  One is free to assume that this model is the initial one equipped with the prescribed type structures, but this is not strictly necessary, so we will not make that assumption.

The style used in the following mimics that employed in \cite{hott-book} to develop HoTT \emph{internally}.  We will make use of the same ideas, although our notation is consistent with the rest of the thesis, and follows the conventions described in \cref{sec:notation,sec:notational-conventions}.

Our main purpose for this \namecref{chap:type-theory} is to develop enough fundamentals of two-level type theory to be able to define certain basic notions that will enable us to express the idea of ``infinite structure'' or ``infinite tower of coherence conditions'', as explained in \cref{sec:infinite-structures}.

\section{Introduction}\label{sec:two-level-intro}

Let $\Phi_0$ be a ``basic'' type former over $\T$.  For concreteness, define $\Phi_0$ as:
$$
\Phi_0 :\equiv \Phi^{\textsc{ieq}} \times \Phi^{\textsc{sum}} \times \Phi^{\textsc{empty}} \times \Phi^\nat
$$
since these are the type structures that will be assumed to exist in the following.

The type former $\Phi_0$ represents structures that will be present both in the fibrant and in the strict fragment of our theory.  However, since the definition of two-level model only requires the two $\T$-structures to be compatible, the two $\Phi_0$-structures will behave very differently, in general.

For the rest of the \namecref{chap:type-theory}, fix a $(\Phi, \Psi)$ model of type theory $\A$, where:

\begin{itemize}
\item $\Psi$ is the type former over $\Phi$ obtained by adding function extensionality ($\Phi^{\textsc{funext}}$), and requiring that the equality in $\Phi$ satisfy $\textsc{UIP}$;
\item the strict fragment of $\A$ admits a system of $\Psi$-universes indexed by some finite ordinal $n$;
\item the fibrant fragment of $\A$ admits a system of \emph{univalent}  $\Phi$-universes indexed by $n$;
\item the two systems of universes can be extended to a system of universes indexed by $\omega \times 2$.
\end{itemize}

The idea of the universe setup is that the two systems of universes live in the two different fragments, but for any $i : n$, the $i$-th fibrant universe is ``contained'' in the $i$-th strict universe.

A crucial observation is that we can find a two-level type former $\Psi'$ so that the initial two-level $\Psi'$-model satisfies all the above conditions, and at the same type all the hypotheses of \cref{thm:conservativity2}.

To make this possible, we have to set up the type formers for our universes so that all the fibrant universes are contained in the first strict one.  This makes the two-level type former $\Psi'$ for set-theoretic.

Therefore, if we assume $\A$ to coincide with the initial two-level $\Psi'$-model, we are allowed to interpret all the results of this \namecref{chap:type-theory} to ordinary $\hott$, thanks to \cref{thm:conservativity2}.

As mentioned above, the type structures for the strict and the fibrant fragments are not required to match (outside of $\T$).  However, it is possible to assume that parts of them do.

In particular, the language is (at least apparently) more expressive if we require the $\Phi^{\textsc{sum}}$, $\Phi^{\textsc{empty}}$ and $\Phi^\nat$-structures to match.  A model where this happens has been referred to as \emph{strong} in \cite{two-level-type-theory}.

One substantial disadvantage of working in a strong model is that \cref{thm:conservativity2} does not apply.  It appears that strong two-level models constitute a proper extension of $\hott$, which means that adopting their language implies having to depart from $\hott$ itself.
Therefore, we will \emph{not} make this assumption in the following.

\subsection{Differences with HTS}\label{sec:hts}

Although our two-level theory is inspired by HTS \cite{hts}, and shares many of its features and motivations, there are some substantial differences between the two systems.

Probably the most important difference is that HTS assumes that natural numbers, binary sums and the empty type in the fibrant fragment can eliminate to arbitrary types.  In other words, coercion from fibrant to strict types preserves those type formers.
As we observed above, the extra assumptions would break the proof of our conservativity result (\cref{thm:conservativity}).  Furthermore, they are not strictly necessary for the development that follows.

Another fundamental difference is that HTS assumes the reflection rule for equality in the strict fragment.  From a semantic point of view, this is a completely unproblematic assumption, and in fact it is within the scope of \cref{thm:conservativity}, since equality in presheaf categories does validate the reflection rule.

However, systems with equality reflection seem to be much harder to study from a meta-theoretical point of view, and consequently harder to implement.  Although most of the current implementation efforts for proof assistants based on Martin-L\"{o}f type theory do not include equality reflection, there have been recent attempts at developing a system within which something like HTS could potentially be realised \cite{andromeda}.

In practice, lack of a reflection rule for strict equality does not seem to be a big hurdle when reasoning within a two-level system \emph{informally}.  Of course, formalising proofs in a proof assistant could potentially be made easier by not having to manually manage rewrites along equality witnesses, but we have no reason to believe that a system that replaces reflection with simply $\textsc{uip}$ would be any less practical for actual formalisation of results
based on a two-level theory.

Finally, universes in the strict fragment of our system are not assumed to be fibrant types, like in HTS.  In some variations of HTS, universes of strict types are even assumed to be contractible.  This is motivated by their interpretation in the simplicial set model (\cref{sec:simplicial-model}). However, universes in presheaf categories are clearly not fibrant in the two-level CwF structure that we constructed in \cref{sec:two-level-presheaves}, so we will not make this assumption.

\section{Basic notions}\label{sec:basic-notions}

We will adopt some specific conventions when working internally in a two-level theory.

As in any two-level model of type theory, we have a distinction between fibrant and strict types.  Technically, they are completely disjoint sets, only connected by the coercion morphism $\fibrant\Ty \to \strict\Ty$.

However, we will sometimes refer to \emph{being fibrant} as a property of a strict type: such a type will be called fibrant if there is a fibrant type that coerces to it.

We will keep the universe hierarchies of the two fragments distinct, by writing $\U$ for a generic fibrant universe, and $\strict\U$ for a strict one.  Similarly to how we dealt with universes in the metatheory (\cref{sec:notation}), we will not write explicit subscripts to identify a universe within a hierarchy, and instead adhere to the convention called ``typical ambiguity''~\cite{feferman:typical-ambiguity}.

Similarly, we will use the superscript $\mathrm{s}$ to denote type formers for the strict fragment, and no superscript at all for their fibrant counterparts.  For example $\strict 0$ is the strict empty type, $A \strict{+} B$ is a strict binary sum, etc.
For strict $\Pi$, $\Sigma$ and unit types, we are free to omit the subscript, since they behave identically to their fibrant versions.
Furthermore, strict equality will be written as $x \seq y$, and fibrant equality simply as $x = y$.

We will follow the same convention for defined notions.  For example, we will write $A \strict\simeq B$ to denote \emph{strict isomorphism}, defined as follows:
\begin{align*}
A \strict\simeq B
& :\equiv (f : A \to B) \\
& \times (g : B \to A) \\
& \times ((a : A) \to g (f a) \seq a) \\
& \times ((b : B) \to f (g b) \seq b).
\end{align*}

Of particular importance for the following are the \emph{finite ordinals} given by $\fin n : \U$.  They are defined by induction on the natural number argument $n : \nat$:
\begin{align*}
\fin 0 & :\equiv 0 \\
\fin{n+1} & :\equiv 1 + \fin n. \\
\end{align*}

Of course, we also get the corresponding strict type $\strict{\fin n}$, indexed over the strict natural numbers, with the analogous strict definition.

We conclude this \namecref{sec:basic-notions} with the following observation, showing that, in order to develop a system with two different notions of equality, one really needs the separation between fibrant and strict types.

\begin{lem}\label{lem:all-fibrant-collapse}
Assume that the coercion morphism $\fibrant\Ty \to \strict\Ty$ is an isomorphism.  Then strict and fibrant equality coincide up to equivalence, hence in particular fibrant equality satisfies \textsc{uip}.
\end{lem}
\begin{proof}
For any type $A$, and $a,b : A$, it follows from the assumption that the strict equality type $a \seq b$ is fibrant.  Therefore, we can define a function:
$$
f : a \seq b \to a = b,
$$
and it is easy to show that $f$ is the inverse of the usual coercion $a = b \to a \seq b$.

Therefore, strict and fibrant equality are strictly isomorphic types.
\end{proof}

\section{Fibrant replacement}\label{sec:fibrant-replacement}

It is natural to ask whether we could extend our theory with a \emph{fibrant replacement} operation, allowing us to convert any type into its ``closest'' fibrant approximation.

In fact, it is not hard to give a definition for a fibrant replacement type former in $\RF_0^2$:
\begin{align*}
\Phi^R : (A : \strict \U)
\to & (R : \fibrant\U) \\
\times & (\eta : A \to R) \\
\times & (\mathsf{elim} : (X : \fibrant\U) \to (A \to X) \to R \to X) \\
\times & ((X : \fibrant\U)(f : A \to X)(a : A) \to \mathsf{elim}^X(f)(\eta(a)) = f(a))
\end{align*}

The type former $\Phi^R$ expressed quite faithfully the idea of ``replacing'' a strict type with a fibrant approximation: given a strict type $A$, we get a fibrant type $R$, together with a function $\eta : A \to R$, and a universal property stating that, for any fibrant type $X$, to define a function $R \to X$ all we need it to define a function $A \to X$.

In fact, a fibrant replacement type former is quite similar to the propositional truncation operation, only, of course, it makes types \emph{fibrant} rather than propostional.

Having fibrant replacement in the theory would make a lot of constructions easier, and it does seem justifiable, since many of the known models of $\hotts$, being Quillen model categories, are indeed equipped with a very similar operation.

For example, a type former along the lines of $\Phi^R$ is considered in \cite{boulier:hts}, where the authors construct a model structure on a universe of strict types using fibrant replacement.

Unfortunately, it turns out that the fibrant replacement operation in models of $\hotts$ cannot be internalised as a $\Phi^R$-structure, as the following theorem shows.

\begin{thm}\label{thm:fibrant-replacement-inconsistent}
Assume the existence of a fibrant replacement type structure $R$, as given by
$\Phi^R$.  Then every fibrant type is a set.
\end{thm}
\begin{proof}
Let $A$ be a fibrant type, and $x : A$. Since the type $(r : x \seq x) \to r = \refl$ is inhabited, so is its fibrant replacement.  Therefore, by path induction, we get that for all $x, y : A$ and $p : x = y$:
$$
R((r : x \seq y) \to r = p).
$$

However, if $p : x = x$, the type $(r : x \seq x) \to r = p$ clearly imples that $\refl = p$, hence, by the elimination property of $R$ and the fibrancy of $\refl = p$, so does its fibrant replacement.

It therefore follows that for all $p : x = x$, we have $\refl = p$, i.e.\ $A$ is a set.
\end{proof}

\section{Reedy fibrant diagrams}
\label{sec:presheaves-reedy}

In this \namecref{sec:presheaves-reedy} we will demonstrate how a two-level system can be used to derive results about $\hott$ by going outside of the fibrant fragment.  This is analogous to how in homotopy theory one can get results that are invariant under homotopy equivalence, even when certain constructions are performed on concrete spaces and do not only depend on their homotopy type.

Specifically, we will define Reedy fibrant diagrams $I \to \U$ for an inverse category $I$, and show that they have limits in $\U$ if $I$ is finite.  This is an internalised version of some of the results in~\cite{shulman:inverse}.

\subsection{Essentially fibrant types and fibrations}

As a preparation for our sample application of the two-level system, we remark that for a strict type $A : \strict\U$, asking that $A$ be fibrant is quite a strong requirement.  It is often sufficient that there exists a fibrant type $B : \U$ and a strict isomorphism $A \strict\simeq B$.  If this is the case, we say that $A$ is \emph{essentially fibrant}.

In \cref{sec:basic-notions}, we have defined the fibrant finite ordinals $\fin n$, for $n : \nat$, and their strict counterparts $\strict{\fin n}$, for $n : \strict\nat$.

\begin{defn}
A type $I$ is said to be \emph{finite} if there exists a number $n : \strict\nat$ and a strict isomorphisms $I \strict\simeq \strict{\fin n}$.
\end{defn}

Note that $\fin n$ is not in general finite.

\begin{lem}\label{lem:fin-prod-finite}
Let $I$ be finite and $X : I \to \U$ be a family of fibrant types. Then, $(i:I) \to X(i)$ is essentially fibrant.
\end{lem}
\begin{proof}
Essential finiteness gives us a cardinality $n$ on which we can do induction. If $n$ is $\strict 0$, then $(i:I) \to X(i)$ is strictly isomorphic to the unit type.  Otherwise, we have an finite $I'$ such that $f : \unit +^s I' \simeq^s I$, and $(i : I) \to X(i)$ is strictly isomorphic to
\begin{equation*}
 X(f(\mathsf{inl}\ 1)) \times ((i : I') \to X(f(\mathsf{inr}\ i))),
\end{equation*}
which is finite by the induction hypothesis.
\end{proof}

Similar to essential fibrancy, we have the following definition:

\begin{defn}
Let $p : E \to B$ be a function.  We say that $p$ is a \emph{fibration} if there is a family $F : B \to \U$ such that the fibre of $p$ over any $b : B$ is strictly isomorphic to $F(b)$, that is,
\begin{equation*}
(b : B) \to \left(F(b) \strict\simeq (e:E) \times (p(e) \seq b) \right).
\end{equation*}
\end{defn}

Any fibrant type family $F : B \to \U$ gives rise to a fibration $p : E \to B$, as it is easy to see that the first projection $(\smsimple B F) \to B$ satisfies the given condition.  Indeed, any strict fibration is isomorphic over $B$ to a strict fibration of this form.  This often allows us to assume that a given fibration has the form of a projection.

\subsection{Strict Categories}\label{sec:strict-categories}

We can define categories in a two-level system in much the same way as precategories are defined in~\cite{hott-book}, except that we can use strict equality to express the laws.
Since strict equality does not suffer from coherence issues, this notion of category is well-behaved even when morphisms form a higher type, or even if they are not fibrant at all.

\newcommand{\obj}[1]{\vert #1 \vert}

\begin{defn}[strict category] \label{def:strictcat}
A \emph{strict category} $\C$ is given by:
\begin{itemize}
\item a type $\obj \C$ of \emph{objects};
\item for all pairs of objects $x, y : \obj \C$, a type $\C(x, y)$ of \emph{arrows} or \emph{morphisms};
\item for all objects $x : \obj \C$, an \emph{identity} arrow
$\id : \C(x, x)$;
\item for all objects $x, y, z : \obj \C$, a \emph{composition} function
$$\circ : \C(y, z) \to \C(x, y) \to \C(x, z).$$
\end{itemize}
With the usual categorical laws holding strictly, meaning that we have:

\begin{itemize}
\item for all object $x, y : \obj \C$ and morphisms $f : \C(x, y)$, strict equalities
\begin{align*}
& \mathsf{idl} : f \circ \mathsf{id} \seq f \\
& \mathsf{idr} : \mathsf{id} \circ f \seq f;
\end{align*}
\item for all objects $x, y, z, w : \obj \C$, and morphisms $f : \C(x, y)$, $g : \C(y, z)$, $h : \C(z, w)$, a strict equality
\begin{equation*}
 \mathsf{assoc} : h \circ (g \circ f) \seq (h \circ g) \circ f.
\end{equation*}
\end{itemize}

We say that a strict category $\C$ is \emph{locally fibrant} if $\C(x, y)$ is
a fibrant type for all objects $x, y$. We say that $\C$ is \emph{fibrant} if it is locally fibrant and the type of objects is a fibrant type.  Finally, we say that $\C$ is \emph{finite} if the type of objects $\obj \C$ is finite.
\end{defn}

The usual theory of categories can be reproduced in the context of strict categories.  It is not hard to define corresponding notions of \emph{functor}, \emph{natural transformation}, \emph{limits}, \emph{adjunctions}, and so on.

From now on, we will refer to strict categories simply as \emph{categories}. If $\C$ is a category, we will often abuse notation and use $\C$ itself to denote its type of objects.

Another important notion is the following:
\begin{defn}[reduced coslice]
Given a category $\C$ and an object $x : \C$, the \emph{reduced coslice} $x \sslash \C$ is the full subcategory of non-identity arrows in the coslice category $x \slash \C$.
A concrete definition is the following.
The objects of $x \sslash \C$ are triples of the following type:
\begin{equation*}
(y : \obj \C) \times (f : \C(x,y)) \times \left((p : x \seq y) \to \neg \left(\trans p f \seq \id \right)\right),
\end{equation*}
where $\transf p$ denotes the $\mathsf{transport}$ function $\C(x,y) \to \C(y,y)$, obtained from the eliminator of strict equality.
Morphisms between $(y,f,s)$ and $(y',f',s')$ are elements $h : \C(y,y')$ such that $h \circ f \seq f'$ in $\C$.

Note that we have a ``forgetful functor'' $\mathsf{forget} : x \sslash \C \to \C$, given by the first projection on objects as well as on morphisms.
\end{defn}

\subsection{Limits and colimits}\label{sec:limits-colimits}

Much of what is known about the category of sets in classical category theory
can be extended to the category of strict types in a given universe.

For example, the following result translates rather directly:

\begin{lem} \label{lem:all-strict-limits}
The universe $\strict\U$, regarded as a category in the usual way, has all small limits.
\end{lem}
\begin{proof}
Let $\C$ be a category with $\obj \C : \strict\U$ and $\C(x,y) : \strict\U$ (for all $x,y$), and let $X : \C \to \strict\U$ be a functor.

We define $L$ to be the type of natural transformations $1 \to X$, where $1 : \C \to \strict\U$ is the constant functor on $\unit$.  Clearly, $L : \strict\U$, and a routine verification shows that $L$ satisfies the universal property of the limit of $X$.
\end{proof}

Unfortunately, for colimits the situation is not as pleasant.  We can certainly show that $\strict\U$ has coproducts, since they can be obtained directly using the strict $\Sigma$ type structure, but only using our assumptions on the strict fragment of the system, we cannot prove that pushouts exist in $\strict\U$.

It would be possible to add pushouts as an additional strict type former.  This type former would be set-theoretic, since presheaf models do have arbitrary colimits, so it would not invalidate the assumptions of \cref{thm:conservativity2}.  Since we will not need arbitrary colimits in the following, we choose to not take this route, and maintain a traditional set of type formers for the strict fragment.

\subsection{Inverse Categories}

Classically, \emph{inverse categories} are defined as categories which do not contain an infinite sequence of nonidentity arrows (see~\cite{shulman:inverse}).

For simplicity, we restrict ourselves to those which have \emph{height} at most $\omega$, and where a \emph{rank function} is given explicitly.
This allows us to perform all constructions constructively, without having to deal with ordinals beyond $\omega$.

\newcommand{\natop}{\op{(\strict\nat)}}

First, consider the category $\natop$ which has $n : \strict\nat$ as objects, and $\natop(n,m) :\equiv n \strict > m$.

The predicate $\strict > : \strict\nat \to \strict\nat \to \strict\U$ is defined in the familiar way, and it is a \emph{strict proposition}, i.e.
$$
(p,q : n \strict > m) \to p \seq q
$$

\begin{defn}
We say that a category $\C$ is an \emph{inverse category} if there is a functor $\varphi : \C \to \natop$ which ``creates identities''; i.e.\ if we have $f : \C(x,y)$ and $\varphi_x \seq \varphi_y$, then we also have $p : x \seq y$ and $\trans p f \seq \mathsf{id}$.
\end{defn}

\subsection{Reedy Fibrant Limits}

We saw in \cref{sec:limits-colimits} that $\strict\U$ has all small limits.
Unfortunately, the same does not hold for the category $\U$ of fibrant types.
Even pullbacks of fibrant types are not fibrant in general (but see Lemma~\ref{prop:fibrant-pullback}).
If we have a functor $X : \C \to \U$, we can always regard it as a functor $X : \C \to \strict\U$, where it does have a limit.
If this limit happens to be essentially fibrant, we say that $X$ has a \emph{fibrant limit}.
Clearly, this limit will then be a limit of the original diagram $C \to \U$, since $\U$ is a full subcategory of $\strict\U$.

Of course, the category $\U$ has general \emph{homotopy limits}.  For example, given a diagram:
$$
\xymatrix{
  & A \ar[d]^f \\
B \ar[r]^g & C,
}
$$
we can form the corresponding homotopy pullback by taking:
$$
P :\equiv (a : A) \times (b : B) \times (f\ a = g\ b),
$$
which is fibrant by construction.

It could in principle be possible to use homotopy limits everywhere in place of strict limits, which would therefore work around the question of the existence of strict limits in $\U$.  However, definining homotopy limits for general (or even inverse) diagrams already requires some machinery to handle arbitrarily high towers of coherence data, hence we cannot tackle it at this point.

\begin{lem}\label{prop:fibrant-pullback}
The pullback of a fibration $E \to B$ along any function $f : A \to B$ is a fibration.
\end{lem}
\begin{proof}
We can assume that $E$ is of the form $\sm{b:B} C(b)$ and $p$ is the first projection.
Clearly, the first projection of $\sm{a:A}C(f(a))$
satisfies the universal property of the pullback.
\end{proof}

Lemma \ref{prop:fibrant-pullback} makes it possible to construct fibrant limits of
certain ``well-behaved'' functors from inverse categories.
The so-called \emph{matching objects} play an important role.
\begin{defn}[matching object; see {\cite[Chp.~11]{shulman:inverse}}]
Let $\C$ be an inverse category, and $X : {\C} \to \U$ a functor. 
For any $z : \C$, we define the \emph{matching object} $M_z^X$ to be the (not necessarily fibrant) limit of
the composition
$z \sslash \C \xrightarrow{\mathsf{forget}} \C \xrightarrow{X} \U \subset \strict\U$.
\end{defn}

\begin{defn}[Reedy fibrant diagram; see {\cite[Def. 11.3]{shulman:inverse}}]
 Let $\C$ be an inverse category and $X : \C \to \U$ be a functor.
 We say that $X$ is \emph{Reedy fibrant} if, for all $z : \C$, the canonical map $X_z \to M_z^X$ is a fibration.
\end{defn}

Using this definition, we can make precise the claim that we can construct fibrant limits of certain well-behaved diagrams.  The following theorem is an internal version of the corresponding result in \cite[Lemma 11.8]{shulman:inverse}.

\begin{thm} \label{thm:fibrant-limits}
Let $\C$ be an finite inverse category. 
Then, every Reedy fibrant $X : \C \to \U$ has a fibrant limit.
\end{thm}
\begin{proof}
By induction on the cardinality of $\C$.  If the type of objects is empty, the limit is the unit type.

Otherwise, let us consider the rank functor $\varphi : \C \to \natop$.
We choose an object $z : \C$ such that $\varphi_z$ is maximal; this is possible (constructively) since $\C$ is assumed to be finite.  In particular, $z$ has no incoming arrow (apart from $\id$).

Let us call $\C'$ the category that we get if we remove $z$ from $\C$;
that is, we set
$$
\obj {\C'} :\equiv (x : \obj \C) \times (\neg (x \seq z)).
$$

Clearly, $\C'$ is still finite and inverse.
Let $X : \C \to \U$ be Reedy fibrant.
We can write down the limit of $X$ (i.e.\ the type of natural transformations to the constant functor) explicitly as
\begin{equation}
(c : (y : \obj \C) \to X_y) \times ((y,y' : \obj \C)(f : \C(y,y')) \to c_y[f] \seq c_{y'}).
\end{equation}
Using that $\obj \C \strict\simeq \obj{\C'} \strict+ 1$, and the fact that $z$ has no incoming non-identity arrows, this type is strictly isomorphic to
\begin{equation} \label{eq:limit-2}
\begin{aligned}
& (c_z : X_z) \times (c : (y : \obj {\C'}) \to X_y) \times \\
& \left((y : \obj {\C'})(f : \C(z,y)) \to c_z[f] \seq c_y\right) \times \\
& \left((y,y' : \obj {\C'})(f : \C(y,y')) \to c_y[f] \seq c_{y'} \right).
\end{aligned}
\end{equation}

Let us write $L$ for the limit of $X$ restricted to $\C'$,  $p$ for the canonical map $p : L \to M_z^X$, and $q$ for the map $X_z \to M_z^X$.

Then, \eqref{eq:limit-2} is strictly isomorphic to
\begin{equation}\label{eq:limit-3}
(c_z : X_z) \times (d : L) \times (p(d) \seq q(c_z))
\end{equation}
This is the pullback of the cospan
$$
\xymatrix{
L \ar[r]^-p & M_z^X & X_z \ar[l]_-q.
}
$$
By Reedy fibrancy of $X$, the map $q$ is a fibration. 
Thus, by Lemma~\ref{prop:fibrant-pullback}, the map from \eqref{eq:limit-3} to $L$ is a fibration.

By the induction hypothesis, $L$ is essentially fibrant.
This implies that \eqref{eq:limit-3} is essentially fibrant, as it is the domain of a fibration whose codomain is essentially fibrant.
\end{proof}

If $\C$ is an inverse category, we will denote by $\C^{<n}$ the full subcategory of $\C$ consisting of all those objects of rank less than $n$.  Correspondingly, for a given diagram $X$ over $\C$, we will denote by $X|n$ the restriction of $X$ to $\C^{<n}$.

\newcommand{\tdeltop}[1]{\left(\op{\Delta_+}\right)^{<#1}}

\subsection{Fibrant Limits and Semi-Simplicial Types}
\label{sec:semi-simplicial-types}

If $X$ is a Reedy fibrant diagram over $\C :\equiv \tdeltop n$, we can restrict $X$ to $n \sslash \C$, then take the limit of the corresponding functor.  With a slight abuse of notation, we will denote such limit by $M_n^X$, even though $X$ is not defined at $n$.

Note that a diagram $X$ over $\tdeltop{n+1}$ is Reedy fibrant if and only if its restriction to $\tdeltop n$ is Reedy fibrant and the map $X_n \to M_n^X$ is a fibration.  Hence, to give a Reedy fibrant diagram over $\tdeltop{n+1}$ is the same as to give a Reedy fibrant diagram $X$ over $\tdeltop{n}$, together with a fibration $Y$ over $M_n^X$.  We will refer to this extended diagram as $\langle X, Y \rangle$.

By mutual induction on the natural number $n$, we can define a type $\sst_n$, and a function $\ssk_n$ from $\sst_n$ to diagrams over $\tdeltop n$.  We start with with $\sst_0 :\equiv \unit$ and $\ssk_0(1)$ set to the trivial diagram over $\tdeltop 0$.

Then, we set
\begin{align*}
& \sst_{n+1} :\equiv \sm{X : \sst_n} (M_{n}^{\ssk_{n} X} \to \U) \\
& \ssk_{n+1}(X, Y) :\equiv \langle \ssk_n(X), Y \rangle.
\end{align*}

Above, we write $M_n^A$ to mean the fibrant type, given by Theorem \ref{thm:fibrant-limits}, which is strictly isomorphic to the matching object of $A$ at $n$ (which would otherwise only be a strict type).

For any strict natural number $n : \strict\nat$, elements of $\sst_n$ are Reedy fibrant $n$-semi-simplicial types.  Since $\sst_n$ is fibrant, this gives an internal representation of semi-simplicial types in $\hott$.

Unfortunately, unless we add some form of $\omega$-limits to the fibrant fragment of our system, we cannot use the family $\sst$ to obtain a fibrant type of general semi-simplicial types (i.e. with simplices of arbitrarily high dimension).

\section{Reedy-Fibrant Replacement}\label{sec:replacement}

The goal of the current section is to show that any strict functor $X$ from an \emph{admissible} inverse category $\C$ to $\U$ has a fibrant replacement; that is, we can construct a Reedy fibrant diagram which is equivalent in a suitable sense.  This construction is an internalisation of the known analogous construction in traditional mathematics (see e.g.~\cite[Lemma 11.10]{shulman:inverse} or \cite{riehl2014theory}.

Note that this notion of fibrant replacement does not contradict the impossibility result of \cref{sec:fibrant-replacement}.  In fact, all the types involved in the construction of a Reedy fibrant replacement are already fibrant: the replacement only happens at the level of diagrams.

\begin{lem}\label{lem:mapping-cocylinder}
Let $f : A \to B$ be a function between fibrant types. Then there exists a fibrant type $N$, an equivalence $i : A \to N$, and a fibration $p : N \to B$, such that $f \seq p \circ i$.
\end{lem}
\begin{proof}
Let $N :\equiv (a : A) \times (b : B) \times (f a = b)$.  The function $i$ is given by $i(a) :\equiv (a, f(a), \refl)$, while $p$ is simply the projection into the component of type $B$.

The function $i$ is clearly the inverse of the projection into the component of type $A$, hence $i$ is an equivalence. Furthermore, $p$ is a fibration, being a projection from a $\Sigma$ type.

The equation $f \seq p \circ i$ holds definitionally.
\end{proof}

We will refer to the type $N$ constructed in the proof of lemma \ref{lem:mapping-cocylinder} as the \emph{mapping cocylinder} of $f$.

\begin{defn}
Let $\C$ be an inverse category.  We say that $\C$ is \emph{admissible} if, for all $n : \C$, Reedy fibrant diagrams over the reduced coslice $x \sslash \C$ have a fibrant limit.
\end{defn}

The main example of an admissible inverse category is $\deltop$. This follows from Theorem \ref{thm:fibrant-limits} and the fact that all the reduced coslices of $\deltop$ are finite.

\begin{defn}
Let $X, Y$ be diagrams over a category $\C$. A natural transformation $f : X \to Y$ is said to be an \emph{equivalence} if, for all $n : \C$, the function $f_n : X_n \to Y_n$ is an equivalence.
\end{defn}

\begin{thm}\label{thm:reedy-fibrant-replacement}
Let $X$ be a diagram over an admissible inverse category $\C$. Then there exists a Reedy fibrant diagram $Y$, and an equivalence $\eta : X \to Y$.
\end{thm}
\begin{proof}
We will construct, by induction on the natural number $n$, a Reedy fibrant diagram $Y^{(n)}$ over $\C^{<n}$, and an equivalence $\eta^{(n)} : X|n \to
Y^{(n)}$.

For $n=0$ there is nothing to construct, so assume the existence of $Y^{(n)}$, and fix any object $x :\C$ of rank $n+1$.  The forgetful functor $i_x: x \sslash \C \to \C$ factors through $\C^{<n}$, hence we can consider the composition $Y^{(n)} \circ i$, which is again a Reedy fibrant diagram, and take its limit $L$.

The map $\eta^{(n)}$ induces a map $X_x \to L$. Define $Y^{(n+1)}_x$ to be the mapping cocylinder of this map.  For any object $y$ of rank $n$ or less, define $Y^{(n+1)}_y$ as $Y^{(n)}_y$, and for any morphism $f : \C(x, y)$, the corresponding function $Y^{(n+1)}_x \to Y^{(n+1)}_y$ is given by the projection from the mapping cocylinder, followed by a map of the universal cone of the limit $L$.  The action of $Y^{(n)}$ on morphisms between objects of ranks $n$ or less is defined to be the same as that of $Y^{(n)}$.

It is easy to see that those definitions make $Y^{(n+1)}$ into a diagram that extends $Y^{(n)}$ to objects of rank $n+1$.  We can also extend $\eta^{(n)}$ by defining $\eta^{(n+1)}(x)$ to be the embedding of $X_x$ into the mapping cocylinder $Y^{(n+1)}_x$, which is an equivalence by lemma \ref{lem:mapping-cocylinder}.

Reedy-fibrancy of $Y^{(n+1)}$ follows immediately from the construction, since $L$ is exactly the matching object of $Y^{(n+1)}$ at $x$.

To conclude the proof, we glue together all the $Y^{(n)}$ and $\eta^{(n)}$ into a single diagram $Y$ and natural transformation $\eta$.  Clearly, $Y$ is Reedy fibrant, and $\eta$ is an equivalence.  \end{proof}

\section{Semi-Segal types}

One of the most promising applications of a homotopy type theory with strict equality is the possibility of constructing and working with algebraic objects comprising infinite towers of coherence conditions.

Semi-semplicial types, introduced in \cref{sec:semi-simplicial-types}, represent the most fundamental of those objects, and a basis on which to build more complex and directly useful structures.

In this section, we will define the notion of \emph{semi-Segal type} and use it to model $(\infty,1)$-semicategories internally in $\hotts$.  The following definitions and results are mostly based on the theory of \emph{Segal spaces}~\cite{segal-spaces}, which can, to a certain extent, be thought of as the special case obtained when the model we are working on happens to be the simplicial model (\cref{sec:simplicial-model}).

The caveat here is that, as noted in \cref{sec:limits-colimits}, the category of simplicial sets is much richer, in terms of \emph{strict} categorical structure, than what we get to see when working from within type theory.  In particular, we noted that the lack of colimits in the formulation of $\hotts$ that we adopted makes it really hard (and perhaps impossible) to reproduce the theory of diagrams over general Reedy categories.

Therefore, we cannot hope for a well-behaved theory of Segal types, and we instead settle for the weaker notion of semi-Segal type, which means that we cannot directly model higher categories equipped with identity morphisms, but only semi-category-like structures.

Fortunately, a rich theory can be developed nonetheless.  For example, the notion of \emph{completeness}, which superficially seems to require the presence of degeneracies in the underlying simplicial type, can actually be defined for semi-Segal types (\cref{def:complete-semi-segal-type}).

\subsection{Preliminaries}\label{sec:segal-prelim}

We begin with some definitions concerning semi-simplicial types. Note that a map $f : \Delta^+(n, m)$ is uniquely determined by the finite strictly increasing sequence $f(0), f(1), \ldots, f(n)$.  In the following, we will use the notation $\facemap_{f(0), f(1), \ldots, f(n)}$ to denote the face map $X_m \to X_n$ of a semi-simplicial type $X$ corresponding to the map $f$.

For example, $\facemap_{012} : X_3 \to X_2$ is the face map corresponding to the inclusion $[2] \to [3]$.

For all $n$, the $n+1$ face maps $X_n \to X_0$ will therefore be denoted by $\facemap_i$, for $i : \fin{n+1}$.  Finally, for all $i : \fin n$, let us write $\lface_i$ for $\facemap_{i,i+1}$.

\begin{defn}
Let $X$ be a semi-simplicial type. The \emph{$n$-th spine} of $X$ is the type:
$$
S_n(X) :\equiv (x : \fin n \to X_1) \times
((i : (\fin{n-1})) \to \facemap_1(x_i) = \facemap_0(x_{i+1})).
$$
\end{defn}

The $n$-th spine of $X$ can be regarded as the type of ``paths'' of length $n$ in the graph underlying $X$.  Note that $S_n(X)$ is a fibrant type, for all semi-simplicial types $X$.

\begin{lem}
Let $X$ be a semi-simplicial type.  For all $n : \nat$, the family of face maps $\lface_i : X_n \to X_1$ determines a map $\phi_n : X_n \to S_n(X)$.
\end{lem}
\begin{proof}
It follows from the definition of $\lface_i$ that:
$$
\facemap_1 \circ \lface_i \seq \facemap_{i+1} \seq \facemap_0 \circ \lface_{i+1},
$$
therefore $\phi_n$ can be defined simply as:
$$
\phi_n(x) :\equiv (\lambda i. \lface_i(x), \lambda i. \refl).
$$
\end{proof}

The map $\phi_n$ defined above is called the $n$-th Segal map of $X$.

\begin{defn}\label{semi-segal-type}
A \emph{semi-Segal type} is a semi-simplicial type $X$ such that all the Segal maps $\phi_n$ are equivalences.  A morphism of semi-Segal types is simply a morphism of the underlying semi-simplicial types.
\end{defn}

For any $n$, the type expressing the fact that $\phi_n$ is an equivalence is called the \emph{$n$-th Segal condition}, and it is a fibrant, propositional type.  In fact, being an equivalence is always a proposition (\cite{hott-book}).

For a semi-simplicial type $X$, the structure of a semi-Segal type is therefore a fibrant proposition, so in particular it is invariant under levelwise equivalence of semi-simplicial types.

We will say that a semi-Segal type is \emph{Reedy fibrant} if the underlying semi-simplicial type is.
Note that that the Segal maps of a Reedy fibrant semi-simplicial type are fibrations.
The following proposition shows that it is quite easy to obtain Reedy fibrant semi-Segal types.

\begin{prop}
Let $X$ be a semi-Segal type.  Then the Reedy fibrant replacement of $X$ (\cref{thm:reedy-fibrant-replacement}) is a Reedy fibrant semi-Segal type.
\end{prop}
\begin{proof}
Let $Y$ be the Reedy fibrant replacement of $X$.  Since $\eta : X \to Y$ is a levelwise equivalence, we get commutative squares:
$$
\xymatrix{
  X_n \ar[r] \ar[d] &
  Y_n \ar[d] \\
  X_1 \times_{X_0} \cdots \times_{X_0} X_1 \ar[r] &
  Y_1 \times_{Y_0} \cdots \times_{Y_0} Y_1,
}
$$

where the horizontal maps are equivalences induced by $\eta$, and the vertical maps are the Segal maps.  Since the Segal maps of $X$ are equivalences, it follows that those of $Y$ are equivalences as well.
\end{proof}

The relationship between semi-Segal types and $(\infty,1)$-categorical structures on types becomes clear when we analyse the first few levels of their semi-simplicial structure.

Let $X$ be a semi-Segal type. We can think of the type $X_0$ as the type of \emph{objects} of $X$. Since $X$ is Reedy fibrant, we have a fibration $\overline X_1$ over $M_0(X) = X_0 \times X_0$.  The type $\overline X_1(x,y)$ can be thought of as the type of \emph{morphisms} between two objects $x$ and $y$.

So far, we have only singled out a graph.  The algebraic nature of semi-Segal types arises from the invertibility of the Segal maps.  For all $n$, let $\psi_n : S_n(X) \to X_n$ be an inverse of $\phi_n$.  Given morphisms $f : \overline X_1(x,y)$ and $g : \overline X_1(y, z)$, we can define their \emph{composition} $g \circ f :\equiv \facemap_{02}(\psi(f,g))$, where $(f,g)$ denotes the element of $S_2(X)$ determined by $f$ and $g$.

With some work, this composition operation can be shown to be \emph{weakly associative}, i.e. there exists a family of \emph{associators}, witnessing equalities between $h \circ (g \circ f)$ and $(h \circ g) \circ f$, for all triples of composable morphisms $f$, $g$ and $h$.

In fact, consider the homotopy pullback:
$$
\xymatrix{
X_2 \times_{X_1} X_2 \ar[r]^-{\pi_1} \ar[d]_{\pi_2} &
X_2 \ar[d]^{\facemap_{02}} \\
X_2 \ar[r]_{\facemap_{01}} &
X_1.
}
$$

Using the equivalence $\phi_2$ twice, we can easily construct an equivalence $\tau : X_2 \times_{X_1} X_2 \to S_3(X)$.  If $u : X_2 \times_{X_1} X_2$ is such that $\tau(u) = (f,g,h)$, then it is not hard to check that $\facemap_{02}(\pi_2(u)) = h \circ (g \circ f)$.

Note that the functions $\facemap_{012}, \facemap_{123} : X_3 \to X_2$ determine a well-defined map $p : X_3 \to X_2 \times_{X_1} X_2$, and $\tau \circ p = \phi_3$.  It follows from the 2-out-of-3 property of equivalences that $p$ is also an equivalence.

Now, let $t = \psi_3(f,g,h)$.  We have that $(f,g,h) = \phi_3(t) = \tau(p(t))$, hence $p(t) = u$. Therefore, $h \circ (g \circ f) = \facemap_{02}(\pi_2(u)) = \facemap_{03}(t)$.  Using a different pullback, one can show that, similarly, $(h \circ g) \circ f = \facemap_{03}(t)$, which implies the required equality.

It is perhaps not surprising that similar arguments, using the Segal conditions at successively higher levels, show the existence of coherence conditions for the semi-categorical structures built so far.  For example, at level 4 one can obtain a family of \emph{pentagonators}, witnessing the commutativity of the following diagram of equalities, for all quadruples of composable morphisms $f,g,h,k$:

$$
\xymatrix{
& k \circ (h \circ (g \circ f)) \ar[dl] \ar[dr] & \\
k \circ ((h \circ g) \circ f) \ar[d] & & (k \circ h) \circ (g \circ f) \ar[d] \\
(k \circ (h \circ g) \circ f) \ar[rr] & & ((k \circ h) \circ g) \circ f.
}
$$

\subsection{Nerve of a strict category}\label{sec:segal-examples}

The most fundamental examples of semi-Segal types are given by strict categories.  In principle, only \emph{semi-categories} are required, since the identities do not play any role in the construction of the corresponding semi-Segal type.  However, we will not be concerned with the extra generality.

Let us recall that in \cref{sec:strict-categories} we defined a locally fibrant category as a strict category $\C$, such that for all objects $x, y$ of $\C$, the type $\C(x, y)$ is fibrant.

\begin{lem}\label{lem:strict-cat-segal}
A locally fibrant category $\C$ determines a Reedy fibrant semi-Segal type.
\end{lem}
\begin{proof}
We define a semi-simplicial type $X$ using a familiar \emph{nerve} construction:
$$
X_n :\equiv (x : \strict{\fin{n+1}} \to \C) \times ((i : \strict{\fin n} \to \C(x_i, x_{i+1})).
$$

Face maps are defined in the usual way.  First, given two indices $i,j : \strict{\fin{n+1}}$, with $p: i < j$, an pair $(x, f) : X_n$ determines a morphism $f_p : \C(x_i, x_j)$ obtained by composing all the $f_k$ with $i \leq k < j$.  The composed morphism $f_p$ can easily be defined by induction over the inequality $p$.
Note that $X_n$ is fibrant thanks to \cref{lem:fin-prod-finite}.

Now, let $i^+ : i < i+1$. A map $\sigma : \Delta_+(n,m)$ can be used to obtain an inequality $\sigma(i^+) : \sigma(i) < \sigma(i+1)$.  We can then define $\sigma^* : X_m \to X_n$ as follows:
$$
\sigma^*(x, f) :\equiv (\lambda i. x_{\sigma(i)}, \lambda i. f_{\sigma(i^+)}).
$$

It is easy to show that $X$, as defined above, is indeed a semi-simplicial type.

The Segal condition can be shown by directly constructing the equivalence between $n$-spines and $n$-simplices.  The type of $n$-spines of $X$ is:
$$
(f : \strict{\fin n} \to X_1) \times ((i : \strict{\fin{n-1}} \to
\facemap_1(f_i) = \facemap_0(f_{i+1}))).
$$

Expanding the definitions, we get the equivalent type:
\begin{align*}
& ((x^0, x^1 : \strict{\fin n} \to \C) \\
\times & (f : (i : \strict{\fin n}) \to \C(x^0_i, x^1_i)) \\
\times & ((i : \strict{\fin{n-1}}) \to x^1_i = x^0_{i+1}).
\end{align*}

We now split $x^0$ into the pair of $s = x^0_0$ and the rest of the sequence, and similarly split $x^1$ into the pair consisting of the beginning of the sequence and $t = x^1_{n-1}$.  With some index manipulation, this yields the equivalent type:
\begin{align*}
& (s, t : \C) \\
\times & ((x^0 x^1 : \strict{\fin{n-1}} \to \C) \\
\times & (g : \to \C(s, x^0_0))) \\
\times & (h : \to \C(x^1_{n-2}, t))) \\
\times & (f : (i : \strict{\fin{n-2}}) \to \C(x^0_i, x^1_{i+1})) \\
\times & ((i : \strict{\fin{n-1}}) \to x^1_i = x^0_i).
\end{align*}

The last component of the previous type states that $x^0$ and $x^1$ are equal. Therefore, we can contract them into a single sequence $x$:
\begin{align*}
& (s, t : \C) \\
\times & ((x : \strict{\fin{n-1}} \to \C) \\
\times & (g : \to \C(s, x_0))) \\
\times & (h : \to \C(x_{n-2}, t))) \\
\times & (f : (i : \strict{\fin{n-2}}) \to \C(x_i, x_{i+1})).
\end{align*}

Now we can join $s$ at the beginning of $x$ and $t$ at the end, to get exactly the type $X_n$ as defined above.  Examining the equivalence $X_n \to S_n(X)$ obtained by chaining the above steps reveals that it is exactly given by the Segal map, thereby proving that $X$ is a semi-Segal type.  \end{proof}

We call the the semi-Segal type $X$ obtained from a locally fibrant category $\C$ using \cref{lem:strict-cat-segal} the \emph{pre-nerve} of $\C$, and we call \emph{nerve} its Reedy fibrant replacement.

It is important to note that the ``weak'' categorical structure arising from the nerve $X$ of a locally fibrant category $\C$ matches precisely with the categorical structure on $\C$ itself.

Clearly, the objects and morphisms of $X$ are the same as those of $\C$.  Let us now consider composition.  Let $f : \C(x_0, x_1)$ and $g : \C(x_1, x_2)$ be two composable morphisms.  Their composition as morphisms of the semi-Segal type $X$ is given by applying the face map $\facemap_{02}$ to the 2-simplex corresponding to the pair $(f,g)$ through the Segal equivalence.  It follows from the definition of the semi-simplicial structure on $X$ that this is indeed the composition $g \circ f$, as expected.

\Cref{lem:strict-cat-segal} can be applied to a universe regarded as a strict category.  We will denote the nerve of a univalent universe as $\TYPE$, leaving implicit the specific universe used, as usual.

\subsection{Maps of semi-Segal types}

The definition of semi-Segal types as semi-simplicial types satisfying a (propositional) property makes it extremely easy to define the corresponding notion of morphism.

\begin{defn}\label{def:semi-segal-map}
A \emph{semi-Segal map} is a morphism between the underlying semi-simplicial types of two semi-Segal types.
\end{defn}

A semi-Segal map can be regarded as the appropriate generalisation of the notion of \emph{functor} between categories.  In particular, we can regard a semi-Segal map between the nerves of two strict categories as a \emph{weak semi-functor} between them.

It is important to note that the notion of semi-Segal map between arbitrary semi-Segal types is not fibrant, hence not invariant under equivalence.  For example, a map between the \emph{pre-nerves} of two strict categories is the same thing as an ordinary (strict) functor between them, while a semi-Segal map between the nerves is a much weaker notion.

\subsection{Completeness}

In the classical theory of Segal spaces, completeness can be understood as the property that the internal notion of \emph{equivalence} in a Segal space can be recovered by only looking the path spaces of its space of points.

In $\hott$, completeness is also a very natural property, corresponding to an internal form of univalence for a categorical structure.  In~\cite{aks}, completeness is considered such a fundamental property that the term \emph{category} is reserved for those structures that possess it (while those that do not are referred to as \emph{precategories}).

However, it is clear that, in order to define completeness in the setting of semi-Segal types, we first need to derive a notion of \emph{equivalence}, which might appear to be problematic, since semi-Segal types have no identity morphisms.

Fortunately, there is a way to work around this issue:

\begin{defn}\label{def:segal-equivalence}
Let $X$ be a Reedy fibrant semi-Segal type, and $f : \overline X_1(x, y)$ be a morphism. We say that $f$ is an \emph{equivalence} if, for all objects $z : X_0$, the maps:
\begin{align*}
f \circ - : \overline X_1(z, x) \to X_1(z, y) \\
- \circ f : \overline X_1(y, z) \to X_1(x, z),
\end{align*}
given by left and right composition with $f$ respectively, are equivalences of types.
\end{defn}

It is easy to see that, if $X$ is the nerve of a strict category $\C$, then $f$ is an equivalence if and only if it is a ``homotopy equivalence'' in $\C$, i.e. if there exists a morphism $g$ in $\C$ in the opposite direction such that $g \circ f = \id$ and $f \circ g = \id$.  Note that we are using fibrant equality here, so a homotopy equivalence is not the same as a categorical isomorphism.

The property of being an equivalence for a morphism $f : X_1$ is a mere proposition, denoted $\isequiv(f)$, hence it determines a \emph{subtype} of $X_1$:
$$
\Equiv^X :\equiv (f : X_1) \times \isequiv(f).
$$

\begin{defn}\label{def:complete-semi-segal-type}
A Reedy fibrant semi-Segal type $X$ is said to be \emph{complete} if the function
$$
\Equiv^X \to X_0,
$$
that maps an equivalence to its first endpoint, is an equivalence of types.
\end{defn}

\begin{prop}\label{prop:universe-complete-segal}
$\TYPE$ is a complete semi-Segal type.
\end{prop}
\begin{proof}
A function $f : X \to Y$ is an equivalence in $\TYPE$ if and only if it is an equivalence of types.  Therefore, completeness of $\TYPE$ follows immediately from univalence.
\end{proof}

A semi-Segal type does not have a built-in notion of identity, but nevertheless, certain morphisms can behave as identities:

\begin{defn}
Let $u : \overline X_1(x, x)$ be a morphism in a Reedy fibrant semi-Segal type $X$.  We say that $u$ is a \emph{unit} if for all $g : \overline X_1(y, x)$ we have that $u \circ g = g$, and for all $h : \overline X_1(x, z)$ we have that $h \circ u = h$.
\end{defn}

Interestingly, completeness is enough for a semi-Segal type to possess units.

\begin{prop}[see {\cite[Lemma 1.4.5]{harpaz:semi-segal}}]
Let $X$ be a complete semi-Segal type.  Then for all objects $x : X_0$ there exists a unit $u : \overline X_1(x,x)$.
\end{prop}
\begin{proof}
By completeness, we get an object $y : X_0$, and an equivalence $f : \overline X_1(x, y)$.  Since $f$ is an equivalence, we can find $u : \overline X_1(x, x)$ such that $f \circ u = f$.  We will show that $u$ is a unit.

If $g : \overline X_1(z, x)$, then $f \circ u \circ g = f \circ g$.  Using the fact that $f$ is an equivalence again, we get that $u \circ g = g$.

Finally, let $h : \overline X_1(x, z)$.  We can find $h' : \overline X_1(y, z)$ such that $h' \circ f = h$. Then $h \circ u = h' \circ f \circ u = h' \circ f = h$, as required.
\end{proof}

\section{Further work}

We have only scratched the surface of what is possible to achieve in a two-level system.

In particular, the notion of semi-Segal type appears to be a quite promising candidate for the role of $(\infty,1)$-categories in type theory.  This thesis only presented the very basic definitions and result, but there is much left to be developed in this area.

\bibliographystyle{plain}
\bibliography{master}

\end{document}

%% file: header.tex
\usepackage{amsthm}
\pdfoutput=1
\usepackage[pdftex, all, 2cell]{xy}
\usepackage{cleveref}
\usepackage{enumitem}
\usepackage{mathtools}
\usepackage{todonotes}
\usepackage{proof}
\usepackage{stmaryrd}
\usepackage{xspace}

\newtheorem{prop}{Proposition}[section]
\newtheorem{thm}[prop]{Theorem}
\newtheorem{lem}[prop]{Lemma}
\newtheorem{cor}[prop]{Corollary}

\theoremstyle{definition}
\newtheorem{defn}[prop]{Definition}
\theoremstyle{remark}
\newtheorem{example}[prop]{Example}
\newtheorem{remark}[prop]{Remark}

\crefname{prop}{proposition}{propositions}
\crefname{thm}{theorem}{theorems}
\crefname{lem}{lemma}{lemmas}
\crefname{cor}{corollary}{corollaries}
\crefname{defn}{definition}{definitions}
\crefname{remark}{remark}{remarks}
\crefname{equation}{\!}{\!}

\newcommand{\A}{\mathcal{A}}
\newcommand{\B}{\mathcal{B}}
\newcommand{\C}{\mathcal{C}}
\newcommand{\D}{\mathcal{D}}
\newcommand{\E}{\mathcal{E}}
\newcommand{\I}{\mathcal{I}}
\renewcommand{\S}{\mathcal{S}}
\newcommand{\T}{\mathcal{T}}

\newcommand{\Ty}{\mathrm{Ty}}
\newcommand{\bTy}{\mathbf{Ty}}
\newcommand{\Tm}{\mathrm{Tm}}
\newcommand{\ext}{\mathsf{ext}}

\newcommand{\U}{\mathcal{U}}
\newcommand{\El}{\mathsf{El}}
\newcommand{\el}{\mathsf{el}}
\newcommand{\V}{\mathcal{V}}
\newcommand{\cat}[1]{\mathsf{#1}}
\newcommand{\id}{\operatorname{id}}
\newcommand{\Id}{\mathsf{Id}}
\newcommand{\op}[1]{#1^{\mathrm{op}}}
\newcommand{\refl}{\mathsf{refl}}

\newcommand{\term}{1}

\renewcommand{\prod}{\origprod\limits}

\renewcommand{\set}{\mathsf{Set}}
\newcommand{\sset}{\mathsf{sSet}}

\newcommand{\seq}{\underset{s}{=}}

\newcommand{\nat}{\mathbb{N}}
\newcommand{\strict}[1]{#1^\mathrm{s}}
\newcommand{\fibrant}[1]{#1^\mathrm{f}}
\newcommand{\hott}{\textsc{HoTT}}
\newcommand{\hotts}{\strict\hott}

\newcommand{\fin}[1]{\mathsf{Fin}_{#1}}

\newcommand{\TYPE}{\mathsf{TYPE}}

\newcommand{\facemap}{\sigma}
\newcommand{\lface}{\mathsf{line}}

\newcommand{\Eq}{\mathsf{Eq}}
\renewcommand{\eq}{\mathsf{eq}}

\newcommand{\interp}[1]{\llbracket #1 \rrbracket}
\newcommand{\presheaf}[1]{\widehat{#1}}
\newcommand{\deltplus}{\Delta_+}
\newcommand{\deltop}{\op{\deltplus}}

\newcommand{\RF}{RF}

\newcommand{\isequiv}{\mathsf{isEquiv}}
\newcommand{\Equiv}{\mathsf{Equiv}}
\newcommand{\ua}{\mathsf{ua}}
\newcommand{\coerce}{\mathsf{coerce}}

\renewcommand{\b}[1]{\mathbf{#1}}

\newenvironment{morphism}%
  {\begin{equation}\begin{cases}}
  {\end{cases}\end{equation}}
\newenvironment{morphism*}%
  {\begin{equation*}\begin{cases}}
  {\end{cases}\end{equation*}}

\newcommand{\prd}[1]{\Pi_{#1}\mkern1mu}

\newcommand{\sm}[1]{\Sigma \left(#1\right).\,}
\newcommand{\smsimple}[1]{\Sigma_{#1}}

\newcommand{\trans}[2]{\ensuremath{{#1}_{*}\mathopen{}\left({#2}\right)\mathclose{}}\xspace}

\newcommand{\transf}[1]{\ensuremath{{#1}_{*}}\xspace} 

\newcommand{\unit}{1}

\newcommand{\sst}{\mathsf{SST}}
\newcommand{\ssk}{\mathsf{SSK}}

\newcommand{\ordinal}[1]{[\mathsf{#1}]}